\documentclass[letterpaper,11pt]{article}

\pdfoutput=1

\usepackage{hyperref}
\usepackage[margin=1in]{geometry}
\usepackage{fullpage}

\usepackage{times}
\usepackage{color}
\hypersetup{
     colorlinks=true,
     linkcolor=blue!50!black,
     citecolor = green!30!black,
     urlcolor=black
}

\usepackage{verbatim}
\usepackage[english]{babel}
\usepackage[utf8]{inputenc}

\usepackage{enumerate}

\usepackage{algorithm}
\usepackage[noend]{algpseudocode}
\usepackage{url}

\usepackage{ifthen}

\usepackage{graphicx}
\usepackage[font=small]{caption}
\usepackage[font=small]{subcaption}
\usepackage{tikz}
\usetikzlibrary{calc}
\usetikzlibrary{shapes.geometric}
\usetikzlibrary{arrows}
\usetikzlibrary{decorations.pathreplacing, decorations.pathmorphing}
\usetikzlibrary{patterns}
\usetikzlibrary{backgrounds}
\usetikzlibrary{scopes}
\usetikzlibrary{arrows, automata, positioning}
\usetikzlibrary{decorations.markings}

\usepackage{amssymb,amsmath,amsthm}
\usepackage{mathtools}
\usepackage{thm-restate}

\newtheorem{theorem}{Theorem}
\numberwithin{theorem}{section}

\newtheorem{lemma}[theorem]{Lemma}

\newtheorem{corollary}[theorem]{Corollary}
\newtheorem{observation}[theorem]{Observation}

\newtheorem{definition}[theorem]{Definition}

\newtheorem{claim}[theorem]{Claim}

    \def\squareforqedb{\hbox{$\blacksquare$}}
\def\qedb{\ifmmode\squareforqedb\else{\unskip\nobreak\hfil
    \penalty50\hskip1em\null\nobreak\hfil\squareforqedb
    \parfillskip=0pt\finalhyphendemerits=0\endgraf}\fi}

    \newenvironment{claimproof}[1][\proofname]{
        \pushQED{\qed}\normalfont \trivlist
        \item\relax{\itshape#1{.}}\hspace\labelsep\ignorespaces
    }{
        \popQED\endtrivlist 
    }

\usepackage[colorinlistoftodos,textsize=tiny,textwidth=2cm,color=green!50!gray]{todonotes} 

\usepackage{comment}

\usepackage{xfrac}
\newcommand{\eps}{\varepsilon}
\renewcommand{\epsilon}{\varepsilon}
\newcommand{\dist}{{\mathrm{dist}}}

\newcommand{\opt}{\mathrm{opt}}

\newcommand{\tmerge}{t_{\mathrm{merge}}}
\newcommand{\drop}{\mathrm{drop}}
\newcommand{\mst}{\mathrm{mst}}

\newcommand{\deltaval}{0.00858}
\newcommand{\gammaval}{0.0774}
\newcommand{\betaval}{7.249}
\newcommand{\Mval}{1.941792}
\newcommand{\alphaval}{2.081}
\newcommand{\lambdaval}{0.02004}
\newcommand{\muval}{0.377}
\newcommand{\gapval}{1.9988}\newcommand{\epsval}{10^{-7}}

\ifdefined\DEBUG
\newcommand{\highlight}[1]{\textcolor{orange!80!black}{#1}}
\else
\newcommand{\highlight}[1]{#1}
\fi

\def\Cscr{\mathcal{C}}

\def\Iscr{\mathcal{I}}
\def\Lscr{\mathcal{L}}

\def\Sscr{\mathcal{S}}

\makeatletter
\def\@fnsymbol#1{\ensuremath{\ifcase#1\or *\or \ddagger\or
    \mathsection\or \mathparagraph\or \|\or **\or \dagger\dagger
    \or \ddagger\ddagger \else\@ctrerr\fi}}
\makeatother

\title{The Bidirected Cut Relaxation for Steiner Tree\\has Integrality Gap Smaller than 2.}

\author{
Jaros\l aw Byrka\thanks{University of Wroc\l aw.  \href{mailto:jaroslaw.byrka@cs.uni.wroc.pl}{jaroslaw.byrka@cs.uni.wroc.pl}, supported by NCN grant number 2020/39/B/ST6/01641.
}  
\and
Fabrizio Grandoni\thanks{IDSIA, USI-SUPSI. 
 \href{mailto:fabrizio.grandoni@idsia.ch}{fabrizio.grandoni@idsia.ch}. Partially supported by the SNF Grant $200021\_200731/1$.
}
\and
Vera Traub\thanks{
Research Institute for Discrete Mathematics and Hausdorff Center for Mathematics, University of Bonn.\\
 \href{mailto:traub@dm.uni-bonn.de}{traub@dm.uni-bonn.de}.
}
}

\date{}

\begin{document}

\maketitle

\thispagestyle{empty}
\addtocounter{page}{-1}
\enlargethispage{-1cm}

\begin{abstract}
\noindent
The Steiner tree problem is one of the most prominent problems in network design. Given an edge-weighted undirected graph and a subset of the vertices, called terminals, the task is to compute a minimum-weight tree containing all terminals (and possibly further vertices). 
The best-known approximation algorithms for Steiner tree involve enumeration of a (polynomial but) very large number of candidate components and are therefore slow in practice. 

A promising ingredient for the design of fast and accurate approximation algorithms for Steiner tree is the bidirected cut relaxation (BCR): bidirect all edges, choose an arbitrary terminal as a root, and enforce that each cut containing some terminal but not the root has one unit of fractional edges leaving it. 
BCR is known to be integral in the spanning tree case [Edmonds'67], i.e., when all the vertices are terminals. For general instances, however, it was not even known whether the integrality gap of BCR is better than the integrality gap of the natural undirected relaxation, which is exactly 2.
We resolve this question by proving an upper bound of $\gapval$ on the integrality gap of BCR.
\end{abstract}
\clearpage

\section{Introduction}
\label{sec:introduction}

In the classical Steiner tree problem we are given an undirected graph $G=(V,E)$, with edge costs $c:E\rightarrow \mathbb{R}_{\ge 0}$, and a subset of vertices $R\subseteq V$ (the \emph{terminals}). Our goals is to compute a tree $T$ of minimum cost $c(T):=\sum_{e\in E(T)}c(e)$ which contains $R$ (and possibly other vertices, called \emph{Steiner nodes}). Steiner tree is one of the best-studied problems in Computer Science and Operations Research, with great practical and theoretical relevance.

\paragraph{Existing Approximation Algorithms.} The Steiner tree problem is NP-hard, indeed it is one of the earliest problems which were shown to belong to this class. More precisely, it is NP-hard to approximate below a factor $\frac{96}{95}$ \cite{CC08}. Several (polynomial-time) constant-factor approximation algorithms are known for this problem. The simple \emph{minimum spanning tree heuristic} gives a $2$-approximation. Without loss of generality, we can replace $G$ with its metric closure. Consider the subgraph  $G[R]$ of $G$ induced by the terminals, and  return a minimum-cost spanning tree of it (a so-called  \emph{terminal MST}). Mehlhorn \cite{Mehlhorn88} gives a nearly-linear-time $\tilde{O}(m+n)$ implementation of this algorithm. The same approximation factor is obtained by Agrawal, Klein, and Ravi \cite{AKR91} (indeed for the more general \emph{Steiner forest} problem) using the primal-dual method (see also \cite{GW95}). This result also implies that the following Undirected Cut Relaxation \eqref{eq:undir-lp-tree} for Steiner tree has integrality gap at most~$2$. 
For a vertex set $U\subseteq V$, let $\delta(U)=\{\{u,v\}\in E: |U\cap \{u,v\}|=1\}$. Then the Undirected Cut Relaxation is the following linear program:
\begin{equation}\label{eq:undir-lp-tree}\tag{UCR}
\begin{aligned}
        \min & & \sum_{e\in E} c(e) \cdot x_e & \\
        \text{s.t.} &  & \sum_{e\in \delta(U)} x_e  &\ \geq\ 1 & & \text{for all } U \subseteq V \text{ with } R\cap U \neq \emptyset \text{ and } R \setminus U \neq \emptyset    \\
        & & x_e &\ \geq\ 0 & & \text{for all } e\in E.
\end{aligned}
\end{equation}

All the known better-than-$2$ approximation algorithms for Steiner tree involve one way or the other the notion of \emph{components}. A component is a connected subgraph of $G$ connecting a subset $X$ of terminals. The component is \emph{full} if its leaves coincide with the terminals it connects. Any minimal Steiner tree can be decomposed into full components, however such components might be arbitrarily large. Borchers and Du \cite{BD95} proved that considering Steiner trees whose full components contain at most $h$ terminals is sufficient to obtain a $1+\frac{1}{\lfloor \log_2 h\rfloor}$ approximation (in particular, we may assume that restricting to components with at most $h=2^{\lceil 1/\eps\rceil}$ terminals increases the cost of a cheapest Steiner tree by at most a factor $1+\eps$). We also remark that the cheapest component connecting a given subset of $h=O(1)$  terminals can be computed in polynomial time.
(Actually this can be done even for $h=O(\log n)$ \cite{DW71}.)

A series of works exploits these ingredients within (variants of) a relative-greedy approach \cite{Z93,zelikovsky_1996_better, KZ97, PS00}, culminating with a $1.55$ approximation \cite{RZ05}. The basic idea is to start with a terminal MST, which provides an initial $2$ approximation $S$. Then one iteratively improves $S$ with the following process. Consider a component $K$ with terminals $X$, and the subgraph $S/X$ obtained from the contraction of $X$. Now drop from $S/X$ a subset of edges of maximum cost while maintaining the connectivity (this can be done by computing a minimum spanning tree of $S/X$). When the cost $\drop(X)$ of dropped edges is larger than the cost $c(K)$ of the component, this induces a cheaper Steiner tree.  
Different algorithms differ in the greedy choice of the improving components. 

The current best $\ln 4+\eps<1.39$ approximation by Byrka, Grandoni, Rothvo{\ss}, and Sanit{\`a}
\cite{BGRS13} is based on a different approach. The authors consider a \emph{hypergraphic relaxation} HYP of the problem based on \emph{directed components} over $2^{\lceil 1/\eps\rceil}$ terminals. This relaxation is exploited within an \emph{iterative randomized rounding} framework: sample one component with probability proportional to its value in an optimal LP solution, contract it, and iterate. Goemans, Olver, Rothvo{\ss}, and Zenklusen \cite{GORZ12} later proved a matching upper bound on the integrality gap of HYP. The same approximation factor was obtained recently by Traub and Zenklusen \cite{TZ22} using a (non-oblivious) local search approach.

All the mentioned better-than-$2$ approximation algorithms for Steiner tree (and in particular the current-best ones) need to enumerate over a (polynomial but) large set\footnote{All known better-than-2 approximation algorithms enumerate over at  least $\Omega(k^3)$ components, where $k$ denotes the number of terminals. The currently best known approximation algorithms enumerate over much larger sets of components.} of candidate components: this makes them mostly unpractical. This raises the following natural question: 
 Is there a very fast (ideally, nearly-linear-time) algorithm for Steiner tree with approximation factor (ideally, substantially) smaller than $2$?
 
\paragraph{The Bidirected Cut Relaxation.} One ingredient that might play a crucial role towards achieving the above goal is the Bidirected Cut Relaxation (BCR). BCR is one of the oldest and best-studied relaxations for Steiner tree, see for example \cite{GM93}.

Let us bidirect all the edges of the given graph $G$: in more detail, replace each undirected edge $e=\{u,v\}\in E$ with two oppositely directed edges $(u,v),(v,u)$, both with cost $c(e)$. Let $\overrightarrow{E}$ be this set of directed edges. Choose an arbitrary terminal $r\in R$ as a root. For $U\subseteq V$, let $\delta^+(U)=\{(u,v)\in \overrightarrow{E}: u\in U,v\notin U\}$. 
Then BCR is the following linear programming relaxation:
\begin{equation}\label{eq:bcr-tree}\tag{BCR}
\begin{aligned}
        \min & & \sum_{e\in \overrightarrow{E}} c(e) \cdot x_e & \\
        \text{s.t.} &  & \sum_{e\in \delta^+(U)} x_e  &\ \geq\ 1 & & \text{for all } U \subseteq V\setminus \{r\}  \text{ with } R \cap U \neq \emptyset \\
        & & x_{e} &\ \geq\ 0 & & \text{for all } e\in \overrightarrow{E}.
\end{aligned}
\end{equation}
\ref{eq:bcr-tree} is indeed a relaxation of the Steiner tree problem because we can orient every Steiner tree towards the root $r$ and then consider the incidence vector of this oriented tree to obtain a feasible solution of \ref{eq:bcr-tree}.
Like \ref{eq:undir-lp-tree}, also \ref{eq:bcr-tree} has a linear number of variables. 
However, it is provably stronger than \ref{eq:undir-lp-tree} in several interesting special cases. For example, a famous result by Edmonds \cite{E67} shows that \ref{eq:bcr-tree} is integral in the minimum spanning tree case, i.e., when $R=V$. 
Notice that, in contrast, \ref{eq:undir-lp-tree} has integrality gap $2$ already on such instances: simply consider a cycle of length $n$ with unit-cost edges; the optimal fractional solution to \ref{eq:undir-lp-tree} sets all the variables to $\frac{1}{2}$, at total cost $\frac{n}{2}$, while the minimum spanning tree costs $n-1$. 

We also know that \ref{eq:bcr-tree} has integrality gap at most $\frac{4}{3}$ on quasi-bipartite instances \cite{CDV11} (improving on an earlier $\frac{3}{2}$ bound from~\cite{RajagopalanV99}), and at most $\frac{991}
{732}$ on Steiner claw free instances \cite{FKOS16,hyattdenesik_et_al:LIPIcs.ICALP.2023.79} (see Section~\ref{sec:related_work} for more details). 
The current best lower bound on the integrality gap of \ref{eq:bcr-tree} is $\frac{6}{5}$, recently shown in \cite{Vicari20}, improving on the earlier lower bound of $\frac{36}{31}$ \cite{BGRS13}. 
Notice that $\frac{6}{5}<\ln(4) \approx 1.386$ and even $\frac{6}{5}<\frac{991}
{732} \approx 1.354$.
The best-known upper bound on the integrality gap of \ref{eq:bcr-tree} prior to our work was~$2$, the same as for \ref{eq:undir-lp-tree}.

The interest in \ref{eq:bcr-tree} goes beyond its possible implications for Steiner tree approximation. If the integrality gap of \ref{eq:bcr-tree} would be sufficiently small, this might
help to find improved approximation for \emph{prize-collecting Steiner tree}: this is the generalization of Steiner tree where one is allowed to leave some terminals disconnected, however one has to pay a penalty for each such terminal.
Another generalization of Steiner tree which might benefit from a better understanding of \ref{eq:bcr-tree} is \emph{Steiner forest}: given a set of pairs of terminals, find the cheapest forest such that each such pair belongs to the same connected component of the forest. 
A natural attempt to improve on the existing $2$-approximation algorithms for this problem \cite{AKR91,GW95,J01}, a well-known open problem in the area, might be to consider some multi-root variant of \ref{eq:bcr-tree}\footnote{Hypergraphic relaxations seem less promising for Steiner forest because in the Steiner forest problem we may assume that every vertex is a terminal.}. 
However, this attempt would be hopeless if the integrality gap of \ref{eq:bcr-tree} was $2$ already for the Steiner tree case. 
For all the mentioned reasons, a natural and important open question is: 

\smallskip
\begin{center}
{ Is the integrality gap of \ref{eq:bcr-tree} smaller than $2$?}    
\end{center}
\smallskip

\noindent We answer this question affirmatively:
\begin{restatable}{theorem}{mainthm}\label{thm:main}
The integrality gap of \ref{eq:bcr-tree} is at most $\gapval$.    
\end{restatable}

We remark that we did not optimize the exact upper bound on the integrality gap, and rather aimed at making the proof as simple as possible.

To prove Theorem~\ref{thm:main}, we give a procedure that yields both a Steiner tree and a dual solution for \ref{eq:bcr-tree}.
We show that the ratio between the cost of these is at most $\gapval$, which implies the desired upper bound on the integrality gap of \ref{eq:bcr-tree}.
Our construction combines the idea of contracting components, as used for example in prior relative greedy approaches, with techniques from primal-dual algorithms.

A key challenge here is that classical primal-dual algorithms always produce dual solutions with laminar support\footnote{For a description of the dual LPs of \eqref{eq:undir-lp-tree} and \eqref{eq:bcr-tree}, see Section~\ref{sec:outline}.
A solution $y$ to one of these LPs has \emph{laminar support} if the support $\mathcal{L} \coloneqq \{ U\subseteq V: y_U > 0\}$ of $y$ is a laminar family, i.e., for all sets $A,B\in \mathcal{L}$ we have $A\subseteq B$, $B\subseteq A$, or  $A\cap B = \emptyset$.
}.
However, there is a family of instances of the Steiner tree problem, where the ratio between the cost of an optimum Steiner tree and the maximum value of a dual solution with laminar support can be arbitrarily close to $2$ (see Appendix~\ref{sec:non-laminar}). 
Hence, for our approach to work, we necessarily need to construct dual solutions whose support is \emph{not} laminar.

Intuitively, we first consider a terminal MST together with the  dual solution produced by the primal-dual algorithm for \ref{eq:undir-lp-tree} (see \cite{AKR91, GW95}) and then try to construct a corresponding dual solution for \ref{eq:bcr-tree}, where each set is grown by a $(1+\delta)$ factor more than in the dual solution for \ref{eq:undir-lp-tree} (for some positive $\delta > 0$).
We show that in this way we can either construct a good dual solution to \ref{eq:bcr-tree}, or we can identify a good component to contract, hence improving the primal solution, i.e., the current Steiner tree.
We provide a more detailed outline of our proof techniques in Section~\ref{sec:outline}.

While the  focus of our work is on proving an upper bound on the integrality gap of \ref{eq:bcr-tree} and we do not discuss algorithmic aspects and  running times in detail, we remark  that our proof techniques give rise to a polynomial-time algorithm. 
In contrast to previous algorithms for Steiner tree with approximation factor smaller than $2$, this algorithm does not require enumerating subsets of terminals to construct the improving components. 
In our approach it is the dual growing procedure that identifies a new improving component to be added (which can be of arbitrary size).

\subsection{Further Related work}
\label{sec:related_work}

Chakrabarty, K{\"{o}}nemann, and Pritchard \cite{CKP10} prove that HYP and \ref{eq:bcr-tree} are equivalent on \emph{quasi-bipartite} instances of Steiner tree, i.e., the special case when there are no edges between Steiner vertices. Chakrabarty, Devanur, and Vazirani~\cite{CDV11}, using a primal-dual construction, show that the integrality gap of \ref{eq:bcr-tree} is at most 4/3 on quasi-bipartite instances. 
Feldman, K{\"o}nemann, Olver, and Sanit{\`a}~\cite{FKOS16} later extended the equivalence between HYP and \ref{eq:bcr-tree} to \emph{Steiner-claw-free} instances of Steiner tree, i.e., instances where each Steiner vertex is adjacent to at most $2$ other Steiner vertices. Notice that this implies an upper bound of $\ln 4$ on the integrality gap of \ref{eq:bcr-tree} on such instances by the result in \cite{GORZ12}.
They also showed that, on small instances built around a single Steiner claw, \ref{eq:bcr-tree} can be weaker than HYP. 
Recently, Hyatt-Denesik, Ameli, and Sanit{\`a} improved the upper bound on the integrality gap of HYP and \ref{eq:bcr-tree} for \emph{Steiner-claw-free} instances to $\frac{991}
{732} < 1.354$ (see proof of Theorem 5 in \cite{hyattdenesik_et_al:LIPIcs.ICALP.2023.79}).

Vicari~\cite{Vicari20} improved the  lower bound on the integrality gap of \ref{eq:bcr-tree} from the previously know $\frac{36}{31}$ from~\cite{BGRS13} to $\frac{6}{5}$.
He also analyzed a strengthening of \ref{eq:bcr-tree} by adding degree constraints.
Filipecki and Van Vye~\cite{FV20} computationally studied a multi-commodity flow strengthening of \ref{eq:bcr-tree}. 

Goemans and Williamson \cite{GW95} present a primal-dual $2$-approximation for prize-collecting Steiner tree. This was slightly improved to $1.9672$ by Archer, Bateni, Hajiaghayi, and Karlof \cite{ABHK11}, and very recently to $1.79$ by 
Ahmadi, Gholami, Hajiaghayi, Jabbarzade, and Mahdavi~\cite{AGHJM24}. 
A $3$-approximation for the prize-collecting generalization of Steiner forest can be obtained with the technique by Bienstock, Goemans, Simchi-Levi, and Williamson \cite{BGSW93}. Goemans later observed that a similar approach gives a $2.542$ approximation (see, e.g., \cite{HJ06} for a detailed proof). Very recently this approximation was improved to $2$ by Ahmadi, Gholami, Hajiaghayi, Jabbarzade, and Mahdavi~\cite{AGHJM24soda}. K{\"{o}}nemann, Olver, Pashkovich, Ravi, Swamy, and Vygen \cite{KOP0SV17} proved that a natural linear programming relaxation for the prize-collecting Steiner forest problem (generalizing \ref{eq:undir-lp-tree}) has integrality gap larger than~$2$.

The \emph{Steiner network} problem is a generalization of Steiner tree (and Steiner forest) where we are given pairwise vertex connectivity requirements $r(u,v)\geq 0$, and the task is to compute a cheapest subgraph of $G$ such that each such pair of vertices $u,v$ is $r(u,v)$-edge connected. In a celebrated result, Jain \cite{J01} obtained a $2$-approximation for this problem using the \emph{iterative rounding} technique. Since then, breaching the $2$-approximation barrier, even just for relevant special cases of Steiner network, became an important open problem. This was recently achieved for some problems in this family, such as \emph{connectivity augmentation} \cite{BGJ23sicomp,CTZ21,TZ21,TZ22,TZ23} and \emph{forest augmentation} \cite{GJT22stoc}. Among the special cases for which $2$ is still the best-known factor, we already mentioned the Steiner forest problem. Another interesting special case is the \emph{minimum-weight 2-edge connected spanning subgraph} problem.

In the directed Steiner tree problem (DST) we are given a directed edge-weighted graph, a root $r$ and a set of terminals $R$. Our goal is to find a minimum-weight arborescence rooted at $r$ that includes $R$.
A classical result by Charikar, Chekuri, Cheung, Dai, Goel, Guha, and Li \cite{CCCDGGL99} gives a $O(k^{\eps}/\eps^3)$ approximation for DST in time $O(n^{1/\eps})$. This implies a $k^{\eps}$ approximation in polynomial time and a $O(\log^3 k)$ approximation in quasi-polynomial time. The quasi-polynomial time approximation factor was improved to $O(\log^2 k/\log\log k)$ by Grandoni, Laekhanukit and Li \cite{GLL19}, who also present a matching inapproximability result. 
We highlight that, while an analogue of \ref{eq:bcr-tree} would be a valid relaxation for DST, all known integrality gap upper bounds for \ref{eq:bcr-tree} only apply to the undirected Steiner tree problem.
In fact, the natural analogue of \ref{eq:bcr-tree} for DST has integrality gap $\Omega(\sqrt{|R|})$ \cite{ZK02} and $\Omega(|V|^{\delta})$ for some constant $\delta >0$ \cite{LL22}.

\section{Outline of our approach}
\label{sec:outline}

In this section we provide an outline of our techniques and explain the motivation behind our construction.
We start with a brief recap of the classical primal-dual approach from \cite{AKR91} in Section~\ref{sec:primal-dual}.
Then we explain our dual growth procedure to construct a dual solution for \ref{eq:bcr-tree}.
In Section~\ref{sec:contracting_components} we then discuss under which conditions our dual growth procedure is successful and constructs a sufficiently good dual solution.
We also discuss how to handle situations where this is not the case.
Finally, in Section~\ref{sec:outline_analysis} we provide a brief overview of the analysis of our procedure.

\subsection{The primal-dual algorithm for Steiner tree}
\label{sec:primal-dual}

The starting point of our approach is the famous primal-dual algorithm from \cite{AKR91} that yields a $2$-approximation for the Steiner tree problem (and the more general Steiner forest problem). 
We will describe it here for Steiner tree from a perspective that will help us to explain our new techniques later on, introducing some useful notation along the way. 
The primal-dual algorithm works with the Undirected Cut Relaxation~\eqref{eq:undir-lp-tree}.
As mentioned in the introduction (Section~\ref{sec:introduction}), this LP relaxation is well-known to have an integrality gap of exactly $2$.
The upper bound of $2$ follows for example from the analysis of the  primal-dual algorithm~\cite{AKR91} (but there are many other proofs, too).

In the following we will assume without loss of generality that $G$ is a complete graph and the edge costs satisfy the triangle inequality (by taking the metric closure).
Then the primal-dual algorithm can always return a terminal MST, i.e., a minimum-cost spanning tree of the graph $G[R]$ induced by the set $R$ of terminals (and we assume without loss of generality that this is indeed the case).
We write $\mst(G[R])$ to denote the cost of a terminal MST.
For $t\in \mathbb{R}_{\ge 0}$, let $G^t[R]$ be the graph $G[R]$  restricted to the edges of cost at most $2t$.
Then, we define $\mathcal{S}^t$ to be the partition of $R$ into the vertex sets of the connected components of $G^t[R]$.
At any time $t$, the primal-dual algorithm has included the edges of  a minimum cost spanning tree in $G[S]$ for every set $S\in \mathcal{S}^t$ (hence proceeding like Kruskal's algorithm for minimum spanning trees).
We denote by $t_{\max}$ the first time $t$ where $G^t[R]$ is connected.
Then the total cost of the returned terminal MST is 
$\mst(G[R]) = 2  \int_{0}^{t_{\max}}  \Bigl(|\mathcal{S}^t| - 1\Bigr)\ dt$.

The primal-dual algorithm also constructs a solution $y$ to the dual linear program of \ref{eq:undir-lp-tree}:
\begin{equation}\label{eq:dual-undir-lp-tree}\tag{Dual-UCR}
\begin{aligned}
        \max & & \sum_{\substack{U \subseteq V: \\
        R\cap U, R \setminus U \neq \emptyset}} y_U & \\
        \text{s.t.} &  & \sum_{U : e\in \delta(U)} y_U  &\ \leq c(e)\ & & \text{for all } e \in E   \\
        & & y_U &\ \geq\ 0 & & \text{for all } U \subseteq V \text{ with }  R\cap U, R \setminus U \neq \emptyset.
\end{aligned}
\end{equation}
The cost of a terminal MST will be no more than twice the cost of the constructed dual solution $y$ and hence
no more than twice the value of~\ref{eq:undir-lp-tree}.

The construction of the dual solution $y$ in the primal-dual algorithm can be described as follows.
We maintain a feasible dual solution $y$ and start with $y$ being the all-zero vector at time $t=0$.
We say that an edge $e$ is tight if the corresponding constraint in \ref{eq:dual-undir-lp-tree} is tight, i.e., $\sum_{U : e\in \delta(U)} y_U  = c(e)$.
Then at any time $t$, we grow the dual variables $y_{U_S}$ for all sets $S\in \mathcal{S}^t$, where
\[
U_S \coloneqq \bigl\{ v\in V : v \text{ is reachable from $S$ by tight edges (at the current time $t$)}\bigr\},
\]
where $v$ being reachable from $S$ means that $v$ is reachable from at least one $s\in S$.
By the definition of the set $U_S$, no edge in $\delta(U_S)$ is tight and hence the constraints of \ref{eq:dual-undir-lp-tree} will not be violated.
Moreover, one can show that for any time $t \le t_{\max}$, the sets $U_S$ with $S\in \mathcal{S}^t$ are disjoint.
In particular, both $U_S \cap R$ and $U_S \setminus R$ will be nonempty and hence $y$ remains a feasible solution to \ref{eq:dual-undir-lp-tree}.
At time $t_{\max}$ the dual solution $y$ has value $\int_{0}^{t_{\max}}  |\mathcal{S}^t|\ dt$, implying that the cost $\mst(G[R])= 2\int_{0}^{t_{\max}} \Bigl(|\mathcal{S}^t| - 1\Bigr) \ dt$ of the terminal MST is at most twice the value of $y$.

\subsection{Growing dual solutions for BCR}\label{sec:dual_growing}

Instead of the linear programming relaxation~\ref{eq:undir-lp-tree} with integrality gap $2$, we now consider the stronger relaxation~\ref{eq:bcr-tree}.
We remark that any feasible solution to \ref{eq:bcr-tree} can be converted into a feasible solution to \ref{eq:undir-lp-tree} of the same cost by omitting the orientation of the edges, i.e., setting $x_{\{v,w\}} \coloneqq x_{(v,w)} + x_{(w,v)}$ for every edge $\{v,w\}\in E$.
It is well-known that the value of \ref{eq:bcr-tree} is independent of the choice of the root~$r$ \cite{GM93}.
The dual of \ref{eq:bcr-tree} is 
\begin{equation}\label{eq:dual-bcr-tree}\tag{Dual-BCR}
\begin{aligned}
        \max & & \sum_{\substack{U \subseteq V\setminus \{r\}: \\
        R\cap U\neq \emptyset}} y_U & \\
        \text{s.t.} &  & \sum_{U : e\in \delta^+(U)} y_U  &\ \leq\ c(e) & & \text{for all } e \in \overrightarrow{E}   \\
        & & y_U &\ \geq\ 0 & & \text{for all } U \subseteq V\setminus \{r\} \text{ with }  R\cap U\neq \emptyset.
\end{aligned}
\end{equation}
We observe that taking any solution $y$ to \ref{eq:dual-undir-lp-tree} and omitting the variables corresponding to sets~$U$ with $r\in U$ yields a feasible solution to \ref{eq:dual-bcr-tree}.
In particular, we can take $y$ to be the solution to \ref{eq:dual-undir-lp-tree} produced by the primal-dual algorithm (see Section~\ref{sec:primal-dual}) and 
let $\overline{y}$ be the solution to \ref{eq:dual-bcr-tree} resulting from $y$ by omitting the variables corresponding to sets~$U$ with $r\in U$.
Because $\mathcal{S}^t$ is a partition of the terminal set,  at any time $t$ we omit the dual variable corresponding to $U_S$ for exactly one set $S \in \mathcal{S}^t$.
Therefore, the value of the solution $\overline{y}$ to \ref{eq:dual-bcr-tree} is $\int_{0}^{t_{\max}}  \Bigl( |\mathcal{S}^t|-1\Bigr)\ dt$ and hence exactly equal to half the cost of a terminal MST.

If the cost of an optimum Steiner tree is significantly less than the cost of a terminal MST, then our solution $\overline{y}$ to \ref{eq:dual-bcr-tree} is good enough to prove that the integrality gap of \ref{eq:bcr-tree} is less than $2$.
The difficult case is when the cost of an optimum Steiner tree is almost equal to the cost of a terminal MST.
In this case we need to improve the construction of our dual solution.

Therefore, let us first consider the special case where our Steiner tree instance is an instance of the minimum spanning tree problem, i.e., the special case where every vertex is a terminal.
In this special case it is well-known that \ref{eq:bcr-tree} is integral \cite{E67}. One way to see this is as follows.
For two terminals $s,\tilde{s} $ we define $\tmerge(s,\tilde{s})$ to be the first time where $s$ and $\tilde{s}$ belong to the same part of the partition $\mathcal{S}^t$.
Then for a directed edge $(s,\tilde{s})\in \overrightarrow{E}$ with $s \neq r$ we have\footnote{
If $s= r$, there are no variables $\overline{y}_U$ for sets $U$ with $(s,\tilde{s})\in \delta^+(U)$ in \ref{eq:dual-bcr-tree} and hence the constraint corresponding to $(s,\tilde{s})$ can be ignored.}
\[
 \sum_{U : (s,\tilde{s})\in \delta^+(U)} \overline{y}_U \ =\ \tmerge(s,\tilde{s}) \ =\ \tfrac{1}{2}  \sum_{U : \{s,\tilde{s}\}\in \delta(U)} y_U  \ \le\ \tfrac{1}{2} \cdot c(\{s, \tilde{s}\}).
\]
Hence, in the special case where all vertices are terminals, we can simply multiply the dual solution $\overline{y}$ with a factor $2$ and maintain a feasible solution to \ref{eq:dual-bcr-tree}. 
This scaled-up dual solution $2 \overline{y}$ has a value equal to the cost of a terminal MST.

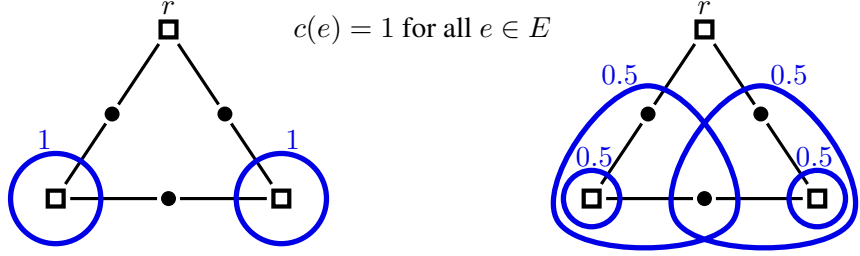
\begin{figure}
\begin{center}
\begin{tikzpicture}[scale=0.75]
\tikzset{terminal/.style={
ultra thick,draw,fill=none,rectangle,minimum size=0pt, inner sep=3pt, outer sep=2.5pt}
}
\tikzset{steiner/.style={
fill=black,circle,inner sep=0em,minimum size=0pt, inner sep=2pt, outer sep=1.5pt}
}

\tikzset{dual/.style={line width=2pt}}

\begin{scope}[every node/.style=terminal]
\node (s1) at (0,0) {};
\node (s2) at (4,0) {};
\node (r) at (2,3) {};
\end{scope}
\node[above=2pt] at (r) {$r$};

\begin{scope}[every node/.style=steiner]
\node (v1) at (1,1.5) {};
\node (v2) at (2,0) {};
\node (v3) at (3,1.5) {};
\end{scope}

\begin{scope}[very thick]
\draw (s1) -- (v1) -- (r) -- (v3) -- (s2) -- (v2) -- (s1);
\end{scope}

\begin{scope}[blue!90!black]
\draw[dual] (s1) ellipse (0.8 and 0.8);
\draw[dual] (s2) ellipse (0.8 and 0.8);
\node at (-0.2,1.05) {$1$};
\node at (4.2,1.05) {$1$};
\end{scope}

\node at (6.5,3) {$c(e)=1$ for all $e\in E$};

\begin{scope}[shift={(9.5,0)}]
\begin{scope}[every node/.style=terminal]
\node (s1) at (0,0) {};
\node (s2) at (4,0) {};
\node (r) at (2,3) {};
\end{scope}
\node[above=2pt] at (r) {$r$};

\begin{scope}[every node/.style=steiner]
\node (v1) at (1,1.5) {};
\node (v2) at (2,0) {};
\node (v3) at (3,1.5) {};
\end{scope}

\begin{scope}[very thick]
\draw (s1) -- (v1) -- (r) -- (v3) -- (s2) -- (v2) -- (s1);
\end{scope}

\begin{scope}[blue!90!black]
\draw[dual] (s1) ellipse (0.5 and 0.5);
\draw[dual] (s2) ellipse (0.5 and 0.5);
\node at (0.05,0.75) {$0.5$};
\node at (3.95,0.75) {$0.5$};
\draw [dual] plot [smooth cycle, tension=0.95] coordinates { (-0.6,-0.4) (1,2) (2.5,-0.4)};
\draw [dual] plot [smooth cycle, tension=0.95] coordinates { (1.5,-0.4) (3,2) (4.6,-0.4)};
\node at (0.5,2.2) {$0.5$};
\node at (3.5,2.2) {$0.5$};
\end{scope}
\end{scope}
\end{tikzpicture}
\caption{\label{fig:scaling_up_problem}
An instance of the Steiner tree problem that arises from a minimum-cost spanning tree problem by subdividing edges.
Terminals are shown as squares and Steiner nodes as circles; all edges have cost $1$. 
On the left, we see in blue the dual solution $\overline{y}$ resulting from the dual solution $y$ computed by the primal-dual algorithm by omitting the set containing the root $r$.
Scaling up $\overline{y}$ by a factor of $2$ does not yield a feasible dual solution to \ref{eq:bcr-tree} because the constraints corresponding to outgoing edges of the two terminals in $R\setminus \{r\}$ are already tight with respect to $\overline{y}$.
On the right, we see in blue the dual solution $z$ resulting from our dual growth process for \ref{eq:bcr-tree}. 
Scaling up $z$ by a factor of $2$ yields an optimum dual solution for \ref{eq:bcr-tree}. 
Notice that our construction does not provide a dual solution with laminar support.
}
\end{center}
\end{figure}

Let us next consider instances that arise from minimum spanning tree instances by subdividing edges, i.e., instances where every non-terminal vertex has degree two in the graph $G$ (before taking the metric closure).
Such instances are of course equivalent to minimum spanning tree instances.
However, for such an instance we cannot simply scale up the dual solution $\overline{y}$ by a factor $2$, as Figure~\ref{fig:scaling_up_problem} illustrates.
To obtain a solution $z$ to \ref{eq:dual-bcr-tree} with twice the value of $\overline{y}$ also on such instances, we consider the following dual growth procedure for \ref{eq:bcr-tree}, which on instances of the minimum spanning tree problem again yields the solution $2 \overline{y}$.

We construct a dual solution $z\in \mathbb{R}_{\ge 0}^{2^V}$.
We highlight that \ref{eq:dual-bcr-tree} has no variables $z_U$ for sets $U$ with $U\cap R = \emptyset$ or $r\in U$ and we therefore say that $z$ is a feasible solution to \ref{eq:dual-bcr-tree} only if these variables are equal to zero (and could thus be omitted).
It will nevertheless be convenient to keep these variables when explaining the construction of $z$.
To describe the construction we need the following definition.
\begin{definition}
For $\delta\ge 0$,  we say that an edge $e\in \overrightarrow{E}$ is \emph{$\delta$-tight} if
\[
\sum_{U\subseteq V: e\in \delta^+(U)} z_U = \tfrac{1}{1+\delta}\cdot c(e).
\]
\end{definition}
Let us fix $\delta=1$ for a moment.
We again consider the partition $\mathcal{S}^t$ of $R$ into the vertex sets of the connected components of the undirected graph $G^t[R]$. To construct $z$, we start from $z$ being the all-zero vector at time $t=0$.
Then at any time $t\in [0, t_{\max})$, for each $S\in {\cal S}^t$ with $r\notin S$ we grow the dual variable $z_{U_S}$ for
\[
U_S \coloneqq \bigl\{ v\in V : v \text{ is reachable from }S\text{ by }\delta\text{-tight edges in $\overrightarrow{E}$ at time } t\bigr\}.
\]
This yields a feasible dual solution $z$ of the same value as $\overline{y}$;
 see Figure~\ref{fig:scaling_up_problem} for an example.
However, in contrast to $2\overline{y}$, the vector $2z$ is a feasible solution to \ref{eq:dual-bcr-tree} (in our special case of instances arising from minimum spanning tree instances by subdividing edges).

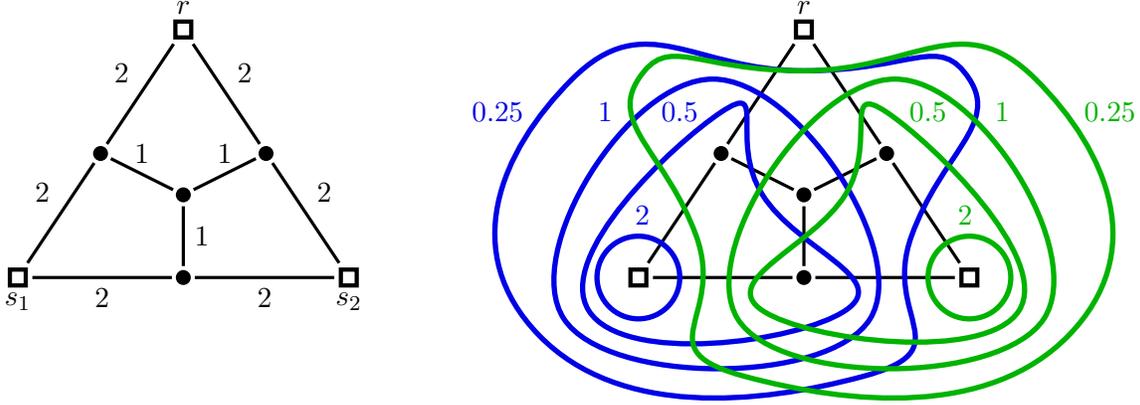
\begin{figure}
\begin{center}
\begin{tikzpicture}[scale=1.1]
\tikzset{terminal/.style={
ultra thick,draw,fill=none,rectangle,minimum size=0pt, inner sep=3pt, outer sep=2.5pt}
}
\tikzset{steiner/.style={
fill=black,circle,inner sep=0em,minimum size=0pt, inner sep=2pt, outer sep=1.5pt}
}

\tikzset{dual/.style={line width=2pt}}

\begin{scope}[every node/.style=terminal]
\node (s1) at (0,0) {};
\node (s2) at (4,0) {};
\node (r) at (2,3) {};
\end{scope}
\node[above=2pt] at (r) {$r$};
\node[below=2pt] at (s1) {$s_1$};
\node[below=2pt] at (s2) {$s_2$};

\begin{scope}[every node/.style=steiner]
\node (v1) at (1,1.5) {};
\node (v2) at (2,0) {};
\node (v3) at (3,1.5) {};
\node (v4) at (2,1) {};
\end{scope}

\begin{scope}[very thick]
\draw (s1) -- node[auto] {$2$} (v1);
\draw (v1) -- node[auto] {$2$} (r);
\draw (r) -- node[auto] {$2$} (v3);
\draw (v3) -- node[auto] {$2$} (s2);
\draw (s2) -- node[auto] {$2$} (v2);
\draw (v2) -- node[auto] {$2$} (s1);
\draw (v4) -- node[above] {$1$} (v1);
\draw (v4) -- node[auto] {$1$} (v2);
\draw (v4) -- node[above] {$1$} (v3);
\end{scope}

\begin{scope}[shift={(7.5,0)}]

\begin{scope}[every node/.style=terminal]
\node (s1) at (0,0) {};
\node (s2) at (4,0) {};
\node (r) at (2,3) {};
\end{scope}
\node[above=2pt] at (r) {$r$};

\begin{scope}[every node/.style=steiner]
\node (v1) at (1,1.5) {};
\node (v2) at (2,0) {};
\node (v3) at (3,1.5) {};
\node (v4) at (2,1) {};
\end{scope}

\begin{scope}[very thick]
\draw (s1) -- (v1) -- (r) -- (v3) -- (s2) -- (v2) -- (s1);
\draw (v4) -- (v1);
\draw (v4) -- (v2);
\draw (v4) -- (v3);
\end{scope}

\begin{scope}[blue!90!black]
\draw[dual] (s1) ellipse (0.5 and 0.5);
\node at (0.05,0.75) {$2$};
\draw [dual] plot [smooth cycle, tension=1] coordinates { (-0.6,-0.4) (1,2) (1.5,1) (2.5,-0.4)};
\node at (0.5,2) {$0.5$};
\draw [dual] plot [smooth cycle, tension=1] coordinates { (-0.9,-0.5) (0.9,2.4) (2.8,-0.5)};
\node at (-0.4,2) {$1$};
\draw [dual] plot [smooth cycle, tension=1] coordinates { (-1.2,-0.7) (-0.8,2.4) (2,2.5) (4,2.4) (3.3,0.5) (2.7,-1.2)};
\node at (-1.7,2) {$0.25$};
\end{scope}

\begin{scope}[green!70!black]
\draw[dual] (s2) ellipse (0.5 and 0.5);
\node at (3.95,0.75) {$2$};
\draw [dual] plot [smooth cycle, tension=0.95] coordinates { (1.5,-0.4) (2.5,1) (3,2)  (4.6,-0.4)};
\node at (3.5,2) {$0.5$};
\draw [dual] plot [smooth cycle, tension=1] coordinates { (4.9,-0.5) (3.1,2.4) (1.2,-0.5)};
\node at (4.4,2) {$1$};
\draw [dual] plot [smooth cycle, tension=1] coordinates { (5.2,-0.7) (4.8,2.4) (2,2.5) (0,2.4) (0.7,0.5) (1.3,-1.2)};
\node at (5.7,2) {$0.25$};
\end{scope}

\end{scope}
\end{tikzpicture}
\vspace*{-2mm}
\caption{\label{fig:cannot_grow_factor_2}
The left part of the figure shows a known instance of the Steiner tree problem where \ref{eq:bcr-tree} is not integral; see e.g.\ \cite{Vicari20}.
The numbers next to the edges show the edge costs and terminals are shown as squares.
A terminal MST (in the metric closure of the instance) has cost $8$ and this is also the cost of an optimum Steiner tree.
On the right we see the dual solution $(1+\delta)z$ that our dual growth procedure yields for $\delta=\frac{7}{8}$, which in this example is an optimum dual solution.
In this example, for every time $t\in [0, t_{\max}) = [0,2)$, the partition $\mathcal{S}^t$ consists of all singleton sets $\{s\}$ with $s\in R$.
The blue sets are sets $U_{\{s_1\}}$  (at different times of the algorithm) and the green ones sets $U_{\{s_2\}}$.
For $\delta > \frac{7}{8}$, the sets $U_{\{s_1\}}$ and $U_{\{s_2\}}$ would for some time $t < t_{\max}$ contain the root vertex $r$ .
}
\end{center}
\end{figure}

Let us now consider general instances of the Steiner tree problem where the cost of an optimum Steiner tree is equal to the cost $\mst(G[R])$ of a terminal MST.
One can still apply the above construction of $z$ to such instances.
However, for $\delta=1$ this might result in growing variables $z_U$ for sets~$U$ containing the root~$r$ and hence does not lead to a feasible solution of \ref{eq:dual-bcr-tree}. 
See Figure~\ref{fig:cannot_grow_factor_2} for an example.
However, for a sufficiently small but positive value $\delta$, we obtain a vector $z$ for which $(1+\delta)z$ is a feasible solution to \ref{eq:dual-bcr-tree} of value $\tfrac{1+\delta}{2}\cdot \mst(G[R])$ (see Lemma~\ref{lem:dual_feasible} and Section~\ref{sec:outline_analysis}). 
Proving this is the main technical challenge in our analysis.

\subsection{Contracting components}
\label{sec:contracting_components}

So far we considered the setting where the cost of an optimum Steiner tree is equal to the cost of a terminal MST.
However, we also need to be able to improve the dual solution $\overline{y}$ if this is only approximately the case, i.e., if the cost of a terminal MST is at most $1+\gamma$ times the cost $\opt$ of an optimum Steiner tree for some small constant $\gamma >0$.
In this case, we might not be able to directly apply the above construction for any $\delta >0$, as illustrated in Figure~\ref{fig:example_local_improvement}.
Another example is given in Figure~\ref{fig:why_prima_dual_does_not_work}, which also illustrates why simply not growing dual variables for sets containing the root (and otherwise proceeding as before) does not yield sufficiently good dual solutions.

To address this issue, we proceed as follows.
We fix a sufficiently small but positive constant $\delta$.
Then we prove that if a terminal MST is not only globally almost optimal, i.e., $\mst(G[R]) < (1+\gamma)\cdot  \opt$, but in a certain sense also ``locally almost optimal'', then our dual growing procedure yields a feasible solution to \ref{eq:dual-bcr-tree} of value  $\tfrac{1+\delta}{2}\cdot \mst(G[R])$.
To explain what we mean by a terminal MST being ``locally almost optimal'', we need the notion of components.
We say that a \emph{component} is any connected subgraph of~$G$ and we say that a component~$K$ \emph{connects} a vertex set~$X$ if all vertices of~$X$ are vertices of~$K$.\footnote{ 
We highlight that our notion of components is more general than the notion of full components used in prior work.
We also do not have any restriction on the size of components. 
In particular, the whole graph and an optimum Steiner tree are examples of components.
}
If a component~$K$ connects a set~$X$ of terminals, then we can obtain a Steiner tree by computing a minimum terminal spanning tree~$T$ in the graph $G[R]/X$ (resulting from $G[R]$ by contracting $X$)  and taking the union of $T$ and the component $K$ (possibly omitting some superfluous edges).
This Steiner tree has cost at most $\mst(G[R]/X) + c(K)$, where $c(K)$ denotes the total cost of the edges of the component $K$.
We write 
\[
\drop(X) \coloneqq \mst(G[R]) - \mst(G[R]/X) 
\]
to denote the decrease of the cost of a terminal MST when contracting~$X$, i.e., the total cost of the edges that we can drop from a terminal MST when contracting $X$.
Now we consider an instance to be ``locally almost optimal'' if for no component $K$ connecting $X$, the cost $c(K)$ of the component is significantly less than $\drop(X)$.
This is captured by the following definition.
\begin{definition}
Let $\gamma > 0$.
An instance of the Steiner tree problem is \emph{locally $\gamma$-MST-optimal} if  for every component $K$ connecting a set $X$ of terminals, we have
$\drop(X) < (1+\gamma)\cdot c(K)$.
\end{definition}
Note that this is a stronger condition than asking for a MST being ``globally almost optimal'', i.e., $\mst(G[R])< (1+\gamma) \cdot \opt$.
For example, the instances shown in Figure~\ref{fig:example_local_improvement} and Figure~\ref{fig:why_prima_dual_does_not_work} are globally almost optimal, but not  locally $\gamma$-MST-optimal.
The main technical statement we prove in this paper is the following theorem.

\begin{figure}
\begin{center}

\begin{tikzpicture}[scale=0.85]
\tikzset{terminal/.style={
ultra thick,draw,fill=none,rectangle,minimum size=0pt, inner sep=3pt, outer sep=2.5pt}
}
\tikzset{steiner/.style={
fill=black,circle,inner sep=0em,minimum size=0pt, inner sep=2pt, outer sep=1.5pt}
}

\tikzset{dual/.style={line width=2pt}}

\def\rad{2}
\def\numk{5}
\def\numn{12}

\node[steiner] (v) at (0.5,0) {};
\node[terminal] (r) at (2.5,0) {};
 \draw[very thick] (r) -- (v);
 
 \node[above left] at (r) {$r$};
 
\begin{scope}[every node/.style={terminal}]
\foreach \i in {1,...,\numk} {
  \pgfmathsetmacro\r{90+(\i-0.5)*360/(2*\numk)}
      \node (s\i) at (\r:\rad) {};
      \draw[very thick] (s\i) -- (v);
}
\end{scope}

 \node[above left] at (s1) {$s_1$};
 \node[below left] at (s\numk) {$s_k$};

\begin{scope}[every node/.style={terminal}, shift={(3,0)}]
\foreach \i in {1,...,\numn} {
  \pgfmathsetmacro\r{-90+(\i-0.5)*360/(2*\numn)}
      \node (t\i) at (\r:\rad) {};
      \draw[very thick] (t\i) -- (r);
}
\end{scope}
 \node[above left] at (t\numn) {$\overline{s}_1$};
 \node[below left] at (t1) {$\overline{s}_q$};

\end{tikzpicture}
\caption{\label{fig:example_local_improvement}
Consider the Steiner tree instance $(G,R)$ that arises as the metric closure of the depicted graph where all shown edges have cost $1$.
The terminal set $R$ is shown by squares.
Then $\opt=k+q+1$ and $\mst(G[R]) = 2k + q$.
For any $\gamma > 0$, we can choose $q$ large enough compared to $k$ to obtain $\mst(G[R]) < (1+\gamma) \cdot \opt$.
However, if we run our dual growth procedure for any value of $\delta > \frac{1}{k}$, we would grow dual variables corresponding to sets $U_S$ containing the root $r$ (for $S=\{s_i\}$ with $i\in \{1,\dots,k\}$).
}
\end{center}
\end{figure}
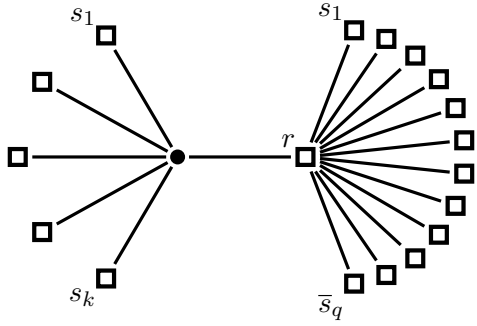

\begin{figure}
\begin{center}

\begin{tikzpicture}[xscale=1.3, yscale=0.9]
\tikzset{terminal/.style={
ultra thick,draw,fill=none,rectangle,minimum size=0pt, inner sep=3pt, outer sep=2.5pt}
}
\tikzset{steiner/.style={
fill=black,circle,inner sep=0em,minimum size=0pt, inner sep=2pt, outer sep=1.5pt}
}

\tikzset{dual/.style={line width=2pt}}

\def\numk{4}
\def\numn{7}
\def\numm{6}

\node[terminal] (s1) at (2,2) {};
\node[steiner] (v1) at (2,1) {};
\draw[very thick] (s1) -- (v1);

\node[terminal] (s\numn) at (6,2) {};
\node[steiner] (v\numn) at (6,1) {};
\draw[very thick] (s\numn) -- (v\numn);

 \node[above=2pt] at (s1) {$s_1$};
 \node[above=2pt] at (s\numn) {$s_q$};
  \node at (4,1.5) {$\cdots$};
 \node[left=2pt] at (v1) {$v_1$};
 \node[right=2pt] at (v\numn) {$v_q$};

\foreach \i in {1,2,\numk } {
  \pgfmathsetmacro\z{ ( \numn -  \numk )  / 2 + \i +  0.5}
      \node[terminal] (t\i) at (\z, -1) {};
}

\node[below=2pt] at (t1) {$\tilde{s}_{1}$};
\node[below=2pt] at (t4) {$\tilde{s}_{k}$};
\node at (5,-1) {$\dots$};

\node[terminal] (t0) at (2.1, -1) {};
\node[below=5pt] at (t0) {$r$};

\foreach \i in {0,1,2,\numk} {
       \draw[very thick] (v1) -- (t\i);
       \draw[very thick] (v\numn) -- (t\i);
}

\begin{scope}[shift={(7,0)}]

\node[steiner] (v1) at (1,1) {};

\foreach \i in {2,...,\numm} {
       \node[terminal] (s\i) at (\i,2) {};
       \node[steiner] (v\i) at (\i,1) {};
       \draw[very thick] (s\i) -- (v\i);
}

 \node[above=2pt] at (s2) {$s_2$};
 \node[above=2pt] at (s\numm) {$s_q$};
 \node[left=2pt] at (v1) {$v_1$};
 \node[right=2pt] at (v\numm) {$v_q$};

\node[terminal] (t0) at (3.5, -1.5) {};
\node[below=5pt] at (t0) {$r$};

\foreach \i in {1,...,\numm} {
       \draw[very thick] (v\i) -- (t0);
}

\end{scope}

\end{tikzpicture}
\caption{\label{fig:why_prima_dual_does_not_work}
Consider the Steiner tree instance $(G,R)$ that arises as the metric closure of the left depicted graph where all shown edges have cost $1$.
The terminal set $R$ is shown by squares.
Then $\opt=k+2q$ and $\mst(G[R]) = 2k + 2q$.
For any $\gamma > 0$, we can choose $q$ large enough compared to $k$ to obtain $\mst(G[R]) < (1+\gamma) \cdot \opt$.
However, if we run our dual growth procedure for any value of $\delta > \frac{1}{k+1}$, we would grow dual variables corresponding to sets $U_S$ containing the root $r$.\\
Note that here, for any time $t \ge \frac{k+2}{(k+1) \cdot (1+\delta)}$, even all dual variables $U_S$ we grow would contain the root $r$. 
Thus, if we stopped growing those dual variables containing the root (or simply set them to zero at the end) to ensure feasibility, then the dual solution we obtain is not sufficiently good to prove an integrality gap below two. (We obtain a lower bound of $\frac{k+2}{k+1}\cdot (k+q)$ on the value of \ref{eq:bcr-tree} and can choose $k$ and $\frac{q}{k}$ arbitrarily large.)\\
We address this issue by first identifying an improving component to contract and only then applying our dual growth procedure.
For example, we might consider a component $K$ with vertex set $\{s_1, v_1, r, \tilde{s}_1, \dots, \tilde{s}_k\}$. Then contracting the set $X$ of terminals connected by $K$, we obtain the (metric closure of) the instance shown on the right.
}
\end{center}
\end{figure}

\begin{theorem}\label{thm:main_dual}
Let $\gamma =\gammaval$ and $\delta=\deltaval$.
Then for any locally $\gamma$-MST-optimal instance $(G,R)$ of the Steiner tree problem, the value of \ref{eq:bcr-tree} is at least $\tfrac{1+\delta}{2} \cdot \mathrm{mst}(G[R])$.
\end{theorem}

To prove Theorem~\ref{thm:main_dual}, we will use the dual growing procedure described in Section~\ref{sec:dual_growing}.
We provide a brief overview of how we analyze this procedure in Section~\ref{sec:outline_analysis}.
To see that Theorem~\ref{thm:main_dual} implies the desired bound on the integrality gap of \ref{eq:bcr-tree}, we proceed as follows.
Consider an instance $\Iscr$ of the Steiner tree problem.
If there exists a component $K$ connecting a terminal set $X$ with $\drop(X) \ge (1+\gamma)\cdot c(K)$, we contract $X$ and iterate until we obtain a  locally $\gamma$-MST-optimal instance~$\Iscr'$.\footnote{
This is not a polynomial-time algorithm, but this is not important for proving an upper bound on the integrality gap of \ref{eq:bcr-tree}.
We remark that it is also possible to work with a polynomial-time algorithm instead and we refer to Section~\ref{sec:steiner_tree} for details.
}
If the cost of a terminal MST decreased by a lot through these iterative contractions, then we obtain a Steiner tree for the original instance~$\Iscr$ that is significantly cheaper than the cost of a terminal MST: this immediately implies an integrality gap smaller than $2$. Otherwise the cost of a terminal MST in the locally $\gamma$-MST-optimal instance $\Iscr'$ arising from the iterative contractions is almost the same as the cost of a terminal MST in the original instance~$\Iscr$.
Because the value of \ref{eq:bcr-tree} can only decrease by contracting terminals, the lower bound on the value of \ref{eq:bcr-tree} for~$\Iscr'$ is also a lower bound on the value of  \ref{eq:bcr-tree} for~$\Iscr$, implying that a terminal MST in the instance $\Iscr$ has cost
roughly $\frac{2}{1+\delta}$ times the value of \ref{eq:bcr-tree}, i.e., significantly less than twice the value  of \ref{eq:bcr-tree}.
We refer to the proof of Theorem~\ref{thm:main} in Section~\ref{sec:construction_tree_and_dual_solution} for details.

\subsection{Proving feasibility of our dual solution to BCR}\label{sec:outline_analysis}

We now provide an overview of how we prove our main technical statement, Theorem~\ref{thm:main_dual}.
We construct a solution to \ref{eq:dual-bcr-tree} of value $\frac{1+\delta}{2} \cdot \mst(G[R])$ as follows.
Given a locally $\gamma$-MST-optimal instance $(G,R)$ of the Steiner tree problem, we apply the dual growth procedure explained in Section~\ref{sec:dual_growing}.
This yields a vector $z\in \mathbb{R}_{\ge 0}^{2^V}$ that satisfies the constraints of \ref{eq:dual-bcr-tree} and where $\sum_{U\subseteq V} z_U = \frac{1+\delta}{2} \cdot \mst(G[R])$.
In order to obtain the desired solution to \ref{eq:dual-bcr-tree} of value $\frac{1+\delta}{2} \cdot \mst(G[R])$, the main difficulty is to prove $z_U= 0$ for all sets $U$ containing the root vertex $r$.
In other words, we need to show that whenever at time $t$ we grow a dual variable $z_{U_S}$ for a set $S\in \mathcal{S}^t$, we have $r\notin U_S$.

To this end, we prove that at time $t$ for any set $S\in \mathcal{S}^t$ with $r\notin S$, all vertices in the set $U_S$ have distance at most $(1+\beta\delta)\cdot t$ from $S$, where $\beta \coloneqq \betaval$.
Because the root $r$ is in a different part of the partition $\mathcal{S}^t$ than the terminals in $S$, the distance of $r$ to each terminal in $S$ is at least $2t$, i.e., we have 
\[
c(r,S)\ \coloneqq\ \min_{s\in S} c(r,s)\ \ge\ 2t \ >\ (1+\beta\delta)\cdot t
\]
 and we can thus conclude $r\notin U_S$.

To prove the claimed upper bound of $(1+\beta\delta)\cdot t$ on the distance $c(u,S)$ of vertices $u\in U_S$ from the terminal set $S$, we proceed as follows.
If $u\in U_S$, this means that there exists a path of $\delta$-tight edges from a terminal $s\in S$ to a vertex $u$.
There might be several such paths and we consider one such path $P$ with a particular property (which we call \emph{$S$-tight}; see Section~\ref{sec:S-tight}) and upper bound its length, i.e, the total cost of its edges.

To upper bound the length of $P$, we first analyze the structure of the terminal sets $S$ that ``contributed'' to~$P$.
We say that a terminal set $S$ \emph{contributed} to a $\delta$-tight edge $e$ if some dual variable $z_{U_S}$ contributed to the $\delta$-tightness of the edge $e$, i.e., if we have grown a set $U_S$ corresponding to $S$ with $e\in\delta^+(U_S)$ (at some earlier time).
Then we say that the sets \emph{contributing to the path $P$} are the sets contributing to some edge of~$P$.
\begin{figure}
\begin{center}
\begin{tikzpicture}[very thick, yscale=0.4, xscale=0.9]
\tikzset{terminal/.style={
ultra thick,draw,fill=none,rectangle,minimum size=0pt, inner sep=3pt, outer sep=2.5pt}
}
\tikzset{steiner/.style={
fill=black,circle,inner sep=0em,minimum size=0pt, inner sep=2pt, outer sep=1.5pt}
}

\node at (-2,6) {(a)};
\node[terminal] (s) at (0.3,2.2) {};
\node[below=2pt] at (s) {$s$};

\begin{scope}[densely dashed]
\draw (0.3,2) ellipse (0.5 and 1);
\draw (1.5,2) ellipse (2 and 2.3);
\draw (1.8,2) ellipse (2.6 and 3.1);
\draw (3.5,2) ellipse (5 and 4.1);
\end{scope}

\begin{scope}[blue!90!black]
\draw (2,2.2) ellipse (0.3 and 0.3);
\draw (2,2.2) ellipse (0.5 and 0.6);
\draw (2,2.2) ellipse (0.7 and 0.9);
\node at (2.5,1) {$\mathcal{C}^1$};
\end{scope}
\begin{scope}[red!80!black]
\draw (6,2.2) ellipse (0.5 and 0.6);
\draw (6,2.2) ellipse (0.8 and 1.2);
\draw (6,2.2) ellipse (1.2 and 1.6);
\node at (7,0.6) {$\mathcal{C}^2$};
\end{scope}
\begin{scope}[green!70!black]
\draw (11,2.5) ellipse (0.6 and 0.8);
\draw (11,2.5) ellipse (1.3 and 2);
\draw (11,2.5) ellipse (1.9 and 3);
\node at (12.5,-0.5) {$\mathcal{C}^3$};
\end{scope}

\node at (-2,-3.7) {(b)};

\begin{scope}[shift={(1.5,-1.7)}, xscale=1.2]
\node[terminal] (start) at (-1.5, -4) {};
\node[left=3pt] at (start) {$s$};
\node[inner sep=0pt, outer sep=0pt] (a1) at (0.3,-4) {};
\node[inner sep=0pt, outer sep=0pt] (a2) at (1,-4) {};
\node[inner sep=0pt, outer sep=0pt] (b1) at (3,-4) {};
\node[inner sep=0pt, outer sep=0pt] (b2) at (4,-4) {};
\node[inner sep=0pt, outer sep=0pt] (c1) at (7,-4) {};

\node (c2) at (8.5,-4) {};

\begin{scope}[->, >=latex, ultra thick]
\draw[densely dashed] (start) to (a1);
\draw[blue!90!black] (a1) to node[above=2pt] {$P_1$} (a2);
\draw[densely dashed] (a2) to (b1);
\draw[red!80!black] (b1) to node[above=2pt] {$P_2$} (b2);
\draw[densely dashed] (b2) to (c1);
\draw[green!70!black] (c1) to node[above=2pt] {$P_3$} (c2);
\end{scope}
\end{scope}
\end{tikzpicture}
\caption{\label{fig:example_contributing_sets}
The top part (a) of the figure shows the terminal sets that contributed to an $S$-tight path $P$ that starts at a terminal $s\in S$.
The sets $S$ with $s\in S$ are shown in black and dashed.
All other sets contributing to $P$ can be partitioned into chains $\mathcal{C}^1$, $\mathcal{C}^2$, and $\mathcal{C}^3$, where any two sets belonging two different chains $\mathcal{C}^i$ are disjoint.
The bottom part (b) of the figure illustrates the path $P$ with subpaths $P_1$, $P_2$, and $P_3$.
The sets from a chain $\mathcal{C}^i$ contribute only on edges of $P_i$.
The sets containing the terminal $s$ can contribute to any edge of $P$ (including the edges of $P_1$, $P_2$, and~$P_3$).
}
\end{center}
\end{figure}
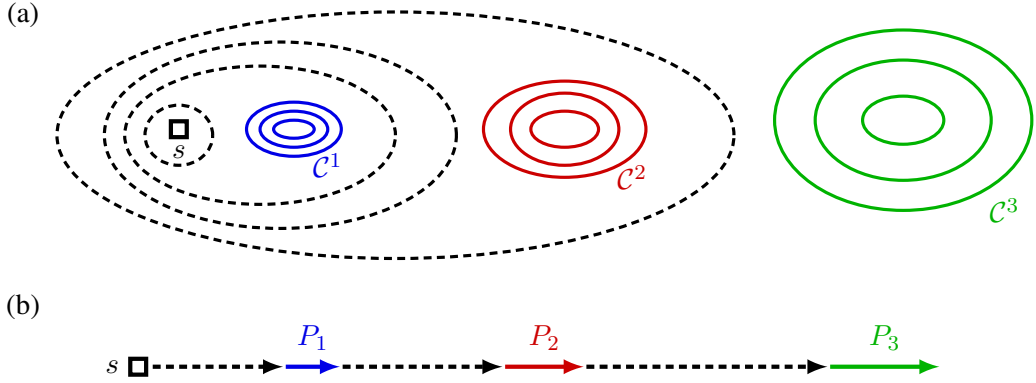

Among these sets that contributed to the path $P$ we distinguish between sets containing the terminal~$s$ and other sets.
If all the sets that contributed to the path $P$ contain the terminal~$s$, it is relatively easy to show that the length of $P$ is at most $(1+\delta)\cdot t$ (see Section~\ref{sec:base_contribution}).
The difficulty is to analyze the contribution of the sets not containing~$s$.
The first step of our analysis is to show that the collection of sets that contributed to~$P$ but do not contain~$s$ is highly structured.
It can be partitioned into chains $\mathcal{C}^1, \dots, \mathcal{C}^k$ such that any two sets from different chains $\mathcal{C}^i$ are disjoint (see Section~\ref{sec:structure_contributing} and the top part of Figure~\ref{fig:example_contributing_sets} for an illustration).
Moreover, we can show that each chain $\mathcal{C}^i$ contributes only on a subpath of the path $P$ and these subpaths for different chains are disjoint\footnote{In our analysis, we will actually allow these subpaths to overlap in one single edge, which by some preprocessing step subdividing edges will have negligible length.} (see Section~\ref{sec:separating} and the bottom part of Figure~\ref{fig:example_contributing_sets} for an illustration).
This  structure allows us to analyze the contribution of every chain separately, which we do using a carefully chosen potential function (see Section~\ref{sec:single_chain}).
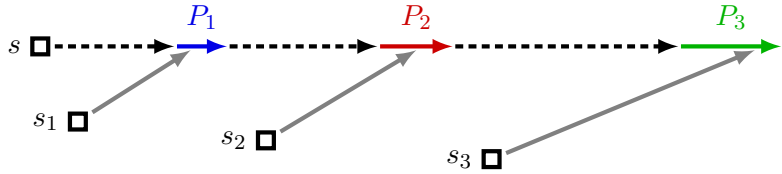
\begin{figure}
\begin{center}
\begin{tikzpicture}[very thick, yscale=0.5, xscale=1]
\tikzset{terminal/.style={
ultra thick,draw,fill=none,rectangle,minimum size=0pt, inner sep=3pt, outer sep=2.5pt}
}
\tikzset{steiner/.style={
fill=black,circle,inner sep=0em,minimum size=0pt, inner sep=2pt, outer sep=1.5pt}
}

\node[terminal] (start) at (-1.5, -4) {};
\node[left=3pt] at (start) {$s$};
\node[terminal] (s1) at (-1, -6) {};
\node[left=3pt] at (s1) {$s_1$};
\node[terminal] (s2) at (1.5, -6.5) {};
\node[left=3pt] at (s2) {$s_2$};
\node[terminal] (s3) at (4.5, -7) {};
\node[left=3pt] at (s3) {$s_3$};

\node[inner sep=0pt, outer sep=0pt] (a1) at (0.3,-4) {};
\node[inner sep=0pt, outer sep=0pt] (a2) at (1,-4) {};
\node[inner sep=0pt, outer sep=0pt] (b1) at (3,-4) {};
\node[inner sep=0pt, outer sep=0pt] (b2) at (4,-4) {};
\node[inner sep=0pt, outer sep=0pt] (c1) at (7,-4) {};
\node (c2) at (8.5,-4) {};

\coordinate (m1) at (0.5,-4.1);
\coordinate (m2) at (3.5,-4.1);
\coordinate (m3) at (8,-4.1);

\begin{scope}[->, >=latex, ultra thick]
\draw[densely dashed] (start) to (a1);
\draw[blue!90!black] (a1) to node[above=2pt] {$P_1$} (a2);
\draw[densely dashed] (a2) to (b1);
\draw[red!80!black] (b1) to node[above=2pt] {$P_2$} (b2);
\draw[densely dashed] (b2) to (c1);
\draw[green!70!black] (c1) to node[above=2pt] {$P_3$} (c2);
\draw[gray] (s1) to (m1);
\draw[gray] (s2) to (m2);
\draw[gray] (s3) to (m3);
\end{scope}

\end{tikzpicture}
\caption{\label{fig:example_building component}
An illustration of the component $K$ that we construct for the path $P$ from Figure~\ref{fig:example_contributing_sets}.
We add the grey paths from terminals $s_1$, $s_2$, and $s_3$ to vertices of $P_1$, $P_2$, and $P_3$, respectively.
Every vertex $s_i$ is contained in a set $S_i$ from the chain $\mathcal{C}^i$.
}
\end{center}
\end{figure}

Finally, let us comment on how we make use of the fact that our Steiner tree instance is locally $\gamma$-MST-optimal in this analysis.
We prove the claimed length bound on $P$ by an induction over time.
For every path~$P$ whose length we analyze, we construct a component $K$ connecting the endpoints of $P$ (one of which is a terminal) and possibly further terminals.
We do this by considering some sets $S'$ contributing to an edge $(v,w)$ of the path~$P$.
We observe that $v$ was contained in $U_{S'}$ at some earlier time (when $S$ contributed to~$e$ by growing $z_{U_{S'}}$) and this allows us to use the induction hypothesis to show that the distance of  $S'$ to~$v$ is small.
We can now add a short path from $S'$ to $v$ to our component. See Figure~\ref{fig:example_building component} for an illustration.

Together with the component~$K$ we also construct a dual certificate that yields a lower bound on $\drop(X)$ for the set~$X$ of terminals connected by $K$ (see the notion of \emph{drop certificates} introduced in  Section~\ref{sec:drop_certificates}).
Intuitively, whenever some set $S'$ not containing $s$ (or more precisely a chain $\mathcal{C}^i$) contributes to the path $P$, this allows us to improve the component $K$ we obtain.
Because our Steiner tree instance is locally $\gamma$-MST-optimal, we know that any component $K$ connecting a terminal set~$X$ we obtain in this way cannot be ``too good'', i.e., we know that $(1+\gamma)\cdot c(K) > \drop(X)$.
This allows us to limit the total contribution of sets $S'$ not containing $s$ to the path $P$.

The details of our analysis are given in Sections  \ref{sec:feasibility_dual} and \ref{sec:distance_bound}.

\section{Constructing the Steiner tree and the dual solution to BCR}
\label{sec:construction_tree_and_dual_solution}

In this section we describe in detail how we construct a Steiner tree and a dual solution for \ref{eq:bcr-tree} whose cost differ by at most a factor $\gapval$.
We first describe the construction of a good Steiner tree and prove how our main result (Theorem~\ref{thm:main}) follows from Theorem~\ref{thm:main_dual} in Section~\ref{sec:steiner_tree}.
Then in Section~\ref{sec:construction_dual_solution} we describe the construction of our dual solution that we will use to prove Theorem~\ref{thm:main_dual}.

\subsection{Constructing a good Steiner tree}\label{sec:steiner_tree}

We now show how the bound on the integrality gap for locally $\gamma$-MST-optimal Steiner tree instances from Theorem~\ref{thm:main_dual} implies the upper bound on the integrality gap of  \ref{eq:bcr-tree} in the general case.

\mainthm*
\begin{proof}We provide an (exponential-time) procedure to construct a Steiner tree of cost at most $\gapval$ times the value of \ref{eq:bcr-tree}.
We are given an instance $(G,R)$ of the Steiner tree problem, where without loss of generality $G$ is a complete graph with metric edge costs (otherwise, take the metric closure).
We start by computing a terminal MST, i.e, a minimum-cost spanning tree in $G[R]$.

While there exists a component $K$ connecting a set $X\subseteq R$ of terminals with $\drop(X) \geq (1+\gamma) \cdot c(K)$, we select one such component $K$, contract $X$, recompute a terminals MST, and iterate the process. Let $(G', R')$ be the Steiner tree instance at the end of the process which, by construction, is locally $\gamma$-MST-optimal.

We can construct a Steiner tree $T$ for the original instance $(G,R)$ by taking a terminal MST\footnote{Technically, we take a set of edges in the input instance corresponding to the edges of such terminal MST.} in the instance $(G',R')$ and adding the components $K$ that we selected in the course of our procedure.
Note that each of the selected components $K$ connecting a terminal set $X$ satisfies $\drop(X) \geq (1+\gamma) \cdot c(K)$, hence the total cost of these components is at most $\frac{1}{1+\gamma} \cdot (\mathrm{mst}(G[R]) - \mathrm{mst}(G'[R']))$. Define $\rho = \frac{\mathrm{mst}(G'[R'])}{\mathrm{mst}(G[R])}$.
We conclude that
\[
c(T)\ \leq\ \mathrm{mst}(G'[R']) + \tfrac{1}{1+\gamma} \cdot (\mathrm{mst}(G[R]) - \mathrm{mst}(G'[R']))=(\rho+\tfrac{1-\rho}{1+\gamma})\cdot  \mathrm{mst}(G[R]).
\]

Let us next lower bound the value of \ref{eq:bcr-tree}. This value is at least $\tfrac{1}{2}\cdot \mst(G[R])$ (see Section~\ref{sec:dual_growing}). By Theorem~\ref{thm:main_dual},  the value of \ref{eq:bcr-tree} for the instance $(G', R')$ is at least $\tfrac{1+\delta}{2} \cdot \mathrm{mst}(G'[R'])=\tfrac{1+\delta}{2}\cdot \rho \cdot \mathrm{mst}(G[R])$.
Because the contraction of vertex sets can only decrease the value of \ref{eq:bcr-tree}, the latter is also a valid lower bound on the value of \ref{eq:bcr-tree} for the original instance $(G, R)$. 
Altogether, the value of \ref{eq:bcr-tree} is at least
\[
\max \Bigl\{ \tfrac{1}{2}\cdot \mst(G[R]), \  \tfrac{1+\delta}{2} \cdot   \rho \cdot\mst(G[R]) \Bigr\}.
\]
Therefore, the integrality gap of \ref{eq:bcr-tree} is at most
\[
\max_{\rho\in [0,1]} \min \Bigl\{ 
2 \Bigl(\rho + \tfrac{1-\rho}{1+\gamma}\Bigr), \  
\tfrac{2}{(1+\delta)\rho} \cdot  \Bigl(\rho + \tfrac{1-\rho}{1+\gamma}\Bigr)
\Bigr\}
\ =\ \highlight{2 \cdot \frac{1+\gamma+\delta}{(1+\gamma)\cdot(1+\delta)}\ <\ 1.9988},
\]
where the maximum is attained for $\rho =\frac{1}{1+\delta}$.
\end{proof}

Note that the above argument only proves a bound on the integrality gap of \ref{eq:bcr-tree}. 
It does not show how to identify the improving components used in the proof of Theorem~\ref{thm:main} in polynomial time.
One can perform this step approximately, by restricting to components connecting a constant number of terminals due to a classical result by Borchers and Du~\cite{BD95}.
While this leads to a polynomial running time, the running time will be very high.
We thus highlight that our proof of Theorem~\ref{thm:main_dual} yields another possibility to turn the proof of Theorem~\ref{thm:main} into an algorithm computing a Steiner tree of cost at most $\gapval$ times the value of \ref{eq:bcr-tree}.

As already indicated in Section~\ref{sec:outline_analysis}, our proof of Theorem~\ref{thm:main_dual} has the form of an algorithm that either produces a good \ref{eq:dual-bcr-tree} solution, or identifies a component that contradicts the assumption of the instance being locally $\gamma$-MST-optimal. 
This leads to a \emph{scale-or-contract} algorithm (which may be seen in an analogy to the round-or-cut method that has been used very successfully for various combinatorial optimization problems), where the procedure in the proof of Theorem~\ref{thm:main} can be implemented as follows. 
Try to construct a good \ref{eq:dual-bcr-tree} solution, and if it fails take the component that manifests the reason of the failure, contract the set of terminals the component connects, and iterate in the contracted graph. 
Then we can define $(G', R')$ to be the instance on which the construction of a good solution to \ref{eq:dual-bcr-tree} succeeded and analyze the algorithm as in the proof of Theorem~\ref{thm:main}.

In the interest of clarity of the argument that \ref{eq:bcr-tree} has integrality gap less than $2$, which is the focus of this work, we omit the details of the argument that improving components may be found efficiently.

\subsection{Dual growth process for BCR}
\label{sec:construction_dual_solution}

In order to prove Theorem~\ref{thm:main_dual}, fix a locally $\gamma$-MST optimal instance $(G,R)$ of the Steiner tree problem, where $G=(V,E,c)$ denotes the given edge-weighted graph and $R$ denotes the set of terminals.
We assume without loss of generality that $G$ is a complete graph with metric edge costs.
Recall that for $t\in \mathbb{R}_{\ge 0}$,  we defined $G^t[R]$ to denote the graph $G[R]$  restricted to the edges of cost at most $2t$.
Then, $\mathcal{S}^t$ is the partition of $R$ into the vertex sets of the connected components of $G^t[R]$ and $t_{\max}$ is the first time $t$ where $G^t[R]$ is connected.

To prove Theorem~\ref{thm:main_dual}, we construct a solution to \ref{eq:dual-bcr-tree} of value $\tfrac{1+\delta}{2} \cdot  \mst(G[R])$.
In order to simplify the analysis of our construction, we consider a graph $G^{\epsilon}$ that arises from $G$ by subdividing edges.
We will choose some sufficiently small $\epsilon > 0$ and subdivide edges of $G$ so that every edge of the resulting graph $G^{\epsilon}$ has cost at most $\epsilon$.
We write $G^{\epsilon}=(V^{\epsilon}, E^{\epsilon},c)$
to denote the resulting weighted graph.

We will construct a solution to \ref{eq:dual-bcr-tree} for $(G^{\epsilon}, R)$, which gives rise to a solution to \ref{eq:dual-bcr-tree} for the original instance $(G,R)$ by undoing the subdivision of edges.
More precisely, if $z$ is a solution to \ref{eq:dual-bcr-tree} for $(G^{\epsilon}, R)$, we construct a solution $\overline{z}$ to \ref{eq:dual-bcr-tree} for $(G,R)$ by setting
\[
   \overline{z}_U \coloneqq \sum_{\substack{W \subseteq V^{\epsilon}: \\ W\cap V = U}} z_W
\]
for every set $U\subseteq V\setminus \{r\}$ with $U\cap R \neq \emptyset$.
If $z$ is a feasible solution to  \ref{eq:dual-bcr-tree} for $(G^{\epsilon}, R)$, then $\overline{z}$ is a feasible solution to \ref{eq:dual-bcr-tree} for $(G,R)$ of the same value.

We choose $\epsilon'>0$ small enough (for example $\epsilon' \coloneqq \epsval$) and let
\begin{equation}\label{eq:definition_eps}
\epsilon\ \coloneqq\ \min \Bigl\{ \tfrac{\epsilon'}{2} \cdot c(\{s,s'\})\ \colon\ s,s'\in R, \ s\neq s' \Bigr\}.
\end{equation}
Note that the fact that $(G,R)$ is locally $\gamma$-MST-optimal (and $\gamma> 0$) implies that the distance $c(\{s,s'\})$ between two terminals is positive because otherwise the component $K$ connecting $s$ and $s'$ by a direct edge satisfies
$ \drop(\{s,s'\}) = c(\{s,s'\}) \ge (1+\gamma) \cdot c(K)$.
Therefore, we indeed have $\epsilon > 0$.
\bigskip

We now describe our construction of a solution to \ref{eq:dual-bcr-tree} for $(G^{\epsilon}, R)$, where we will proceed as outlined in Section~\ref{sec:dual_growing}.
We write $\overrightarrow{E}^{\epsilon}$ to denote the set of edges resulting from $E^{\epsilon}$ by bidirecting every edge, i.e, $\overrightarrow{E}^{\epsilon}$ contains the directed edges $(v,w)$ and $(w,v)$ for every edge $\{v,w\}\in E^{\epsilon}$.
We define the cost of a directed edge $(v,w)$ to be the same as the cost of the corresponding undirected edge $\{v,w\}$.

For every time $t\in [0,t_{\max}]$ we define a feasible dual solution $z^t$ to  \ref{eq:dual-bcr-tree} for $(G^{\epsilon}, R)$, starting with $z^t=0$ at time $t= 0$.
The vector $(1+\delta)\cdot z^{t_{\max}}$ will be the desired dual solution of value $\tfrac{1+\delta}{2} \cdot  \mst(G[R])$.

For convenience of notation, we view $z^t$ as a vector in $\mathbb{R}^{2^V}$ (although \ref{eq:dual-bcr-tree} has variables only for subsets containing at least one terminal, but not the root) and we will ensure $z^t_U=0$ for sets $U$ for which \ref{eq:dual-bcr-tree} has no variable.

\begin{definition}
We say that an edge $e\in \overrightarrow{E}^{\epsilon}$ is \emph{$\delta$-tight} at time $t$ if 
\[
\sum_{U\subseteq V^{\epsilon}: e\in \delta^+(U)} z^t_U = \frac{1}{1+\delta}\cdot c(e).
\]
\end{definition}

For a set $S$ of terminals and a time $t$, we define
\[
U_S^t \coloneqq \Bigl\{ v\in V^{\epsilon} : v \text{ is reachable from }S\text{ by }\delta\text{-tight edges at time } t\Bigr\},
\]
where $v$ being reachable from $S$ by $\delta$-tight edges means that there exists an $s$-$v$ path consisting of $\delta$-tight edges for at least one terminal $s\in S$.

Then we construct the dual solution $z^t$ starting with $z^0 = 0$ by growing  at every time $t$ the dual variables corresponding to sets $U_S^t$ for all $S\in\mathcal{S}^t$ with $r\notin S$.
Formally, we then have for every time $t\in [0, t_{\max}]$, 
\[
    z^t \coloneqq \int_{0}^t  \ \sum_{S\in \mathcal{S}^{t'}: r \notin S} \chi^{U^{t'}_S} \quad dt',
\]
where $\chi^{U^{t'}_S}$ is the vector $y \in \mathbb{R}^{2^V}$ with $y_U = 1$ for $U=U^{t'}_S$ and $y_U = 0$ otherwise.

By the definition of the sets $U^t_S$, we never grow a dual variable corresponding to a set $U$ with an outgoing $\delta$-tight edge. 
Hence, we maintain the invariant that for every edge $e \in \overrightarrow{E} ^{\epsilon}$, we have
\begin{equation}\label{eq:constraint_with_slack}
 \sum_{U\subseteq V^{\epsilon} : e\in \delta^+(U)} z^t_U  \ \leq\ \tfrac{1}{1+\delta} \cdot c(e).
\end{equation}
Because the set $U^t_S$ always contains the terminal set $S\subseteq R$, we only grow dual variables corresponding to sets containing at least one terminal.
The most difficult part of our analysis is to prove that we never grow dual variables corresponding to sets $U$ with $r\in U$.
In Section~\ref{sec:feasibility_dual}, we will prove the following.

\begin{restatable}{lemma}{dualfeasible}\label{lem:dual_feasible}
Let $\gamma =\gammaval$ and $\delta =\deltaval$.
Then for every set $U\subseteq  V^{\epsilon}$ with $r\in U$ we  have $z^{t_{\max}}_U = 0$.
\end{restatable}

From Lemma~\ref{lem:dual_feasible} we conclude that $z^{t_{\max}}_U = 0$ for every set $U\subseteq V^{\epsilon}$ for which \ref{eq:dual-bcr-tree} for $(G^{\epsilon}, R)$ does not have a variable (i.e., where $U\cap R = \emptyset$ or $r\in U$).
Therefore, by \eqref{eq:constraint_with_slack}, the vector $(1+\delta) \cdot z^{t_{\max}}$ is a feasible solution to \ref{eq:dual-bcr-tree} for $(G^{\epsilon}, R)$.
It has value
\[
(1+\delta) \sum_{U\subseteq V^{\epsilon}} z^{t_{\max}}_U \ =\ (1+\delta) \int_{0}^{t_{\max}} \Bigl(|\mathcal{S}^t| - 1 \Bigr)\ dt =\ \tfrac{1+\delta}{2} \cdot \mst(G[R]),
\]
where we refer to Section~\ref{sec:primal-dual} for an explanation of the last equality.
In order to prove Theorem~\ref{thm:main_dual} and thus Theorem~\ref{thm:main}, it now remains to prove Lemma~\ref{lem:dual_feasible}.

\begin{table}[t]
\centering
\fbox{
\begin{minipage}{0.29\textwidth}
\vspace*{-4mm}
\begin{align*}
\delta\ =&\ \deltaval \\
\gamma\ =&\ \gammaval \\
\eps'\ =&\ \epsval \\
\beta\ =&\ \betaval  \\
\alpha\ =&\ \alphaval \\
\lambda\ =&\ \lambdaval \\
\mu\ =&\ \muval \\
\end{align*}
\vspace*{-12mm}
\end{minipage}
\begin{minipage}{0.49\textwidth}
\centering
\begin{align*}
M\ =&\ \frac{\frac{2}{1+\gamma}-1-\beta\delta}{\frac{(1+\beta\delta)^2}{1-\beta\delta}-(\frac{2}{1+\gamma}-1-\beta\delta)}\ \approx\ \Mval
\end{align*}
\end{minipage}
}
\caption{A list of  the constants appearing in our proofs.}
\label{table:constants}
\end{table}

In our proof of Lemma~\ref{lem:dual_feasible} we will define a few constants (besides the already mentioned $\delta$, $\gamma$, and $\epsilon'$).
For easier reference, we provide an overview of all these constants in Table~\ref{table:constants}.

\section{Feasibility of the dual solution (Proof of Lemma~\ref{lem:dual_feasible})}\label{sec:feasibility_dual}

In this section we describe several key concepts we use to prove Lemma~\ref{lem:dual_feasible}.
We start with a few basic observations about the partitions $\mathcal{S}^t$ of the terminal set $R$ in Section~\ref{sec:terminal_partitions}.
Then we introduce the concept of \emph{drop certificates}, which we will use to derive lower bounds on $\drop(X)$ for terminal sets $X\subseteq R$ (Section~\ref{sec:drop_certificates}).
Finally, in Section~\ref{sec:S-tight} we introduce the concept of \emph{$S$-tight paths} and argue why it suffices to bound the length of such paths in order to prove Lemma~\ref{lem:dual_feasible}.
In Section~\ref{sec:distance_bound} we will then prove such an upper bound on the length of $S$-tight paths.

\subsection{The partitions $\mathcal{S}^t$ of the terminals}
\label{sec:terminal_partitions}

Recall that $(G,R)$ is a locally $\gamma$-MST-optimal instance of the Steiner tree problem, where $G=(V,E,c)$ denotes the given complete weighted graph and $R$ denotes the set of terminals.
For $t\in \mathbb{R}_{\ge 0}$,  we defined $G^t[R]$ to denote the graph $G[R]$  restricted to the edges of cost at most $2t$.
Then, $\mathcal{S}^t$ is the partition of $R$ into the vertex sets of the connected components of $G^t[R]$ and $t_{\max}$ is the first time $t$ where $G^t[R]$ is connected.

We view $\mathcal{S}^t$ as a partition evolving over time $t$.
Then $\mathcal{S}^t$ only changes at finitely many times $t$, because every time it changes the number of its parts decreases by one.
We write $\mathcal{S}$ to denote the union of all set families $\mathcal{S}^t$ for $t\in [0,t_{\max})$.
Whenever $\mathcal{S}^t$ changes, some of its parts get merged into a larger part.
Therefore, the union $\mathcal{S}$ of all collections $\mathcal{S}^t$ is a laminar family. (This means that for every two sets $A,B \in \mathcal{S}$, we have $A\subseteq B$, $B\subseteq A$, or $A\cap B = \emptyset$.)

If $S\in \mathcal{S}^t$ and $r\notin S$, we say that $S$ is \emph{active} at time $t$.
We observe that for every set $S\in \mathcal{S}$, the set of times at which $S$ is active is a (half-open) interval.

\begin{definition}
For a set $S\in \mathcal{S}$ with $r\notin S$, we write $[a^S, d^S)$ to denote the interval of times $t$ with $S\in \mathcal{S}^t$.
We call $a^S$ the  \emph{activation time} of $S$ and $d^S$ the \emph{deactivation time} of $S$.
\end{definition}

If at some time $t$ the terminal set $S$ is active and we have $e\in  \delta^+(U_S^t)$, then we say that $S$ is \emph{contributing} to $e$ at time $t$.

Recall that for two terminals $s_1,s_2$ we write $\tmerge(s_1,s_2)$ to denote the first time $t \in [0,t_{\max}]$ at which $s_1$ and $s_2$ belong to the same part of the partition $\mathcal{S}^t$.
We highlight that these merge times satisfy the following properties.

\begin{observation}\label{obs:deactivation_time_lower_bounds_merge_time}
For a set $S\in \mathcal{S}$ with $r\notin S$ and two terminals $s_1$ and $s_2$ with $s_1 \in S$, we have $\tmerge(s_1,s_2) \ <\ d^S$ if and only if $s_2\in S$.
\end{observation}

\begin{observation}\label{obs:merge_time_max}
For any three terminals $s_1,s_2,s_3$, we have
\[
\tmerge(s_1,s_3) \le  \max \bigl\{ \tmerge(s_1,s_2),\  \tmerge(s_2,s_3)  \bigr\}.
\]
\end{observation}

We also observe that the minimum deactivation time of a set  $S\in \mathcal{S}$ with $r\notin S$ is half of the minimum distance of two terminals.
Therefore, by the choice of $\epsilon$ in  \eqref{eq:definition_eps}, we have 
\begin{equation}\label{eq:choice_of_epsilons}
\epsilon  \ \le\ \epsilon' \cdot d^S
\end{equation}
for every set  $S\in \mathcal{S}$ with $r\notin S$.

Recall that the graph $G^{\epsilon}=(V^{\epsilon}, E^{\epsilon},c)$ arose from $G$ by the subdivision of edges.
We write $\dist(v,w)$ to denote the shortest-path distance (with respect to the edge costs $c$) between any two vertices $v$ and $w$ in the graph $G^{\epsilon}$.
We highlight that for any two vertices $v$ and $w$ of the original graph $G$ this distance is the same as the cost of the edge $\{v,w\}$ of $G$, where we use that the edge costs of the complete graph $G$ satisfy the triangle inequality.
We observe that by the definition of the partitions $\mathcal{S}^t$ for $t\in [0,t_{\max}]$, we have the following lower bound on the distance of two terminals.

\begin{observation}\label{obs:dist_lb_merge_time}
For any two terminals $s_1,s_2$, we have
\[
\dist(s_1,s_2) \ge 2 \cdot \tmerge(s_1,s_2).
\]
\end{observation}

\subsection{Drop certificates for lower bounds on the cost of components}
\label{sec:drop_certificates}

In this section we derive the key property of locally $\gamma$-MST-optimal instances that we will use to prove Lemma~\ref{lem:dual_feasible}.
Recall that $\drop(X)$ denotes the total cost of the terminal MST edges we could drop when adding a component connecting the terminals in~$X$.
We now introduce the notion of a \emph{drop certificate} for~$X$, which provides a lower bound on $\drop(X)$ as we will show in Lemma~\ref{lem:drop_certificates} below.
See Figure~\ref{fig:drop_certificate} for an example.

\begin{definition}
We call $\tilde{\mathcal{S}} \subseteq \mathcal{S}$ a \emph{drop certificate} for a terminal set $X \subseteq R$ if
\begin{itemize}\itemsep0pt
\item $|\tilde{\mathcal{S}}| = |X| - 1$, and
\item for all $s_1,s_2 \in X$ with $s_1 \neq s_2$, there is a set $S \in \tilde{\mathcal{S}}$ with $|S\cap \{s_1,s_2\}| =1$.
\end{itemize}
We also say that $2 \sum_{S\in \tilde{\mathcal{S}}} d^{S}$  is the \emph{value} of the drop certificate $\tilde{\mathcal{S}}$.
\end{definition}

\begin{figure}
\begin{center}
\begin{tikzpicture}[scale=1.1]
\tikzset{terminal/.style={
ultra thick,draw,fill=none,rectangle,minimum size=0pt, inner sep=3pt, outer sep=2.5pt}
}
\tikzset{steiner/.style={
fill=black,circle,inner sep=0em,minimum size=0pt, inner sep=2pt, outer sep=1.5pt}
}
\tikzset{dual/.style={green!75!black, line width=2pt}}

\begin{scope}[every node/.style=terminal]
\node at (0.5, 1) {};
\node at (2, 1) {};
\node at (3.5, 1) {};
\node at (4, 2.5) {};
\node at (6.1, 2.5) {};
\node at (6.1,1) {};
\node at (7.5, 1) {};
\node at (7.5, 2.5) {};
\end{scope}

\begin{scope}[dual]
\draw (2,1) ellipse (2.5 and 1.1);
\draw (0.5, 1) ellipse (0.6 and 0.6);
\draw (2, 1) ellipse (0.6 and 0.6);
\draw (6.1, 2.5) ellipse (0.55 and 0.55);
\draw (7.5, 1) ellipse (0.6 and 0.6);
\draw (7, 1) ellipse (1.5 and 0.9);
\draw (6.8, 1.5) ellipse (2 and 2);
\end{scope}
\end{tikzpicture}
\vspace*{-2mm}
\caption{\label{fig:drop_certificate}
An example of a drop certificate $\tilde{\mathcal{S}} $ (green) for a set $X$ of terminals (black squares).
}
\end{center}
\end{figure}
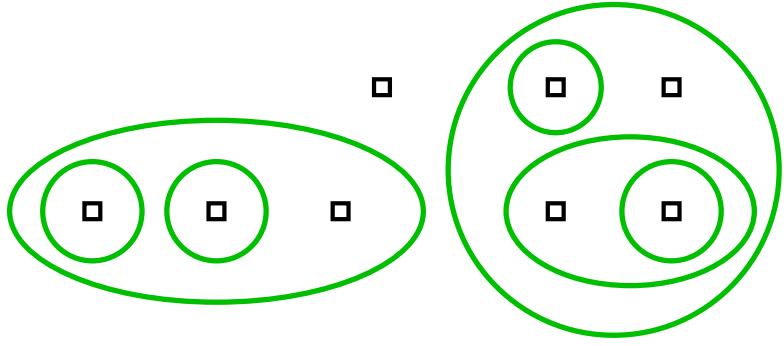

We next show that the value of any drop certificate for $X$ is a lower bound on $\drop(X)$, which justifies the name ``drop certificate''.
Even more, we show that the maximum value of a drop certificate for $X$ is equal to $\drop(X)$.
We remark that for the remainder of our analysis it would be sufficient to only prove that the $\drop(X)$ is at least the maximum value of a drop certificate, but we prove equality for completeness and to provide more intuition.

\begin{lemma}\label{lem:drop_certificates}
Let $X \subseteq R$ and $K$ be a component connecting $X$. Then $\drop(X)$ is the maximum value of a drop certificate for $X$, i.e., 
\[ \textstyle
   \drop(X) \coloneqq \max \Bigl\{  2 \sum_{S\in \tilde{\mathcal{S}}} d^{S} : \tilde{\mathcal{S}}\text{ is a drop certificate for }X \Bigr\}.
\]
\end{lemma}

\begin{proof}
Recall that the graph $G^t[R]$  for $t\in \mathbb{R}_{\ge 0}$ is the graph $G[R]$  restricted to the edges of cost at most $2t$.
Moreover, for $t\in \mathbb{R}_{\ge 0}$ we write $\mathcal{S}^t$ to denote the partition of $R$ into the vertex sets of the connected components of $G^t[R]$.
Any terminal MST has cost $\mst(G[R]) = 2  \int_{0}^{t_{\max}}  \bigl(|\mathcal{S}^t| - 1 \bigr)\ dt$.

For any  $t\in \mathbb{R}_{\ge 0}$ we write $\mathcal{S}_X^t$ to denote the partition of $R/X$ into the vertex sets of the connected components of $G^t[R]/X$.
Then any minimum-cost spanning tree of $G[R]/X$ has cost $\mst(G[R]/X) = 2  \int_{0}^{t_{\max}}  \bigl( |\mathcal{S}_X^t| - 1 \bigr)\  dt$.
Therefore, we have
\[
\drop(X)\ =\ \mst(G[R]) - \mst(G[R]/X)  \ =\ 2  \int_{0}^{t_{\max}}  \Bigl(|\mathcal{S}^t| -   |\mathcal{S}_X^t|\Bigr) \ dt.
\]
We will show that for any drop certificate $\tilde{\mathcal{S}}$ for $X$ and any time $t\in [0,t_{\max}]$  we have
\begin{equation}\label{eq:inequality_proof_drop_certificate}
 |\mathcal{S}^t| -   |\mathcal{S}_X^t| \ \ge\ \bigl|  \bigl\{ S \in \tilde{\mathcal{S}} : d^S \ge t \bigr\} \bigr|
\end{equation}
and that for every $X\subseteq R$ there exists a drop certificate that satisfies for all $t\in [0,t_{\max}]$ the inequality \eqref{eq:inequality_proof_drop_certificate} with equality.
Because the value of a drop certificate $\tilde{\mathcal{S}}$ is
\[
2 \sum_{S\in \tilde{\mathcal{S}}} d^{S} \ =\ 2  \int_{0}^{t_{\max}} \bigl|  \bigl\{ S \in \tilde{\mathcal{S}} : d^S \ge t \bigr\} \bigr|\ dt,
\]
this will complete the proof.
\bigskip

We now consider a drop certificate $\tilde{\mathcal{S}}$ for $X$ and a time $t\in [0,t_{\max}]$ and prove \eqref{eq:inequality_proof_drop_certificate}.
Note that the partition $\mathcal{S}_X^t$ of $R/X$ arises from the partition $\mathcal{S}^t$ of $R$ by merging the parts of  $\mathcal{S}^t$ that have a nonempty intersection with $X$.
Let $S_1, \dots, S_k \in \mathcal{S}^t$ be the parts of $\mathcal{S}^t$ that have a nonempty intersection with $X$ and write  $X_i \coloneqq S_i \cap X$ for each $i\in\{1,\dots,k\}$.
Because $\mathcal{S}$ is a laminar family, every set $S\in  \tilde{\mathcal{S}}\subseteq \mathcal{S}$  that separates two terminals $s,s' \in X_i \subseteq S_i$, i.e., where $|S\cap \{s,s'\}| =1$, must be a proper subset of $S_i$.
Because the drop certificate $\tilde{\mathcal{S}}$ contains a set separating  $s$ and $s'$ for all terminal pairs $s,s' \in X_i$, it must contain at least $|X_i|-1$ proper subsets of $S_i$.
Each such proper subset $S\in \mathcal{S}$ of $S_i \in \mathcal{S}^t$ satisfies $d^S < t$.
We conclude that the number of sets $S$ in the drop certificate $\tilde{\mathcal{S}}$  with $d^S < t$ is at least $\sum_{i=1}^k \bigl( |X_i| -1\bigr) = |X| - k$ and hence
\[
\bigl|  \bigl\{ S \in \tilde{\mathcal{S}} : d^S \ge t \bigr\} \bigr| \ \le\ |\tilde{\mathcal{S}}| - \bigl( |X| -k\bigr) \ =\ |X| -1 - \bigl( |X| -k\bigr) \ =\  k-1\ =\  |\mathcal{S}^t| -   |\mathcal{S}_X^t| ,
\]
which shows \eqref{eq:inequality_proof_drop_certificate}.

\bigskip
Finally, we consider a set $X\subseteq R$ and construct a drop certificate $\tilde{\mathcal{S}}$ for $X$ satisfying \eqref{eq:inequality_proof_drop_certificate} with equality for all $t\in [0,t_{\max}]$.
We again view the partition $\mathcal{S}^t$ as a partition changing over time, starting with $t=0$ and ending with $t=t_{\max}$.
At time $t=0$, every element of $X$ is in a different part of $\mathcal{S}^t$.
(Here, we use that because our Steiner tree instance is locally $\gamma$-MST-optimal, the distance of every two terminals is positive.)
At time $t=t_{\max}$ the partition $\mathcal{S}^t$ has only one part containing all elements of $X$.
Going from time $t=0$ to time $t=t_{\max}$, the partition $\mathcal{S}^t$ changes at discrete time steps by merging several of its parts to a single new part.
Whenever several parts $S_1,\dots, S_k$ with a nonempty intersection with $X$ merge into a single part at time $t$, we include the parts $S_1,\dots, S_{k-1}$ into $\tilde{\mathcal{S}}$.
Note that the deactivation time of these sets $S_1,\dots, S_{k-1}$ is equal to the current time $t$ 
and hence the set family $\tilde{\mathcal{S}}$ we obtain in this way satisfies \eqref{eq:inequality_proof_drop_certificate} with equality for all $t\in [0,t_{\max}]$.
Because we add $k-1$ sets whenever the number of parts of $\mathcal{S}^t$ with a nonempty intersection with $X$ decreases by $k-1$, the total number of sets included in $\tilde{\mathcal{S}}$ is indeed $|X|-1$. 
Moreover, for any two distinct terminals $s,s'\in X$, we consider the time $t=\tmerge(s,s')$, i.e., the first time $t$ where $s$ and $s'$ belong to the same part of $\mathcal{S}^t$.
At this time $t$, the two parts $S_i$ and $S_j$ containing $s$ and $s'$, respectively, got merged and hence we added at least one of them, say $S_i$ to $\tilde{\mathcal{S}}$.
Then $S_i$ separates $s$ and $s'$, i.e, $|S\cap \{s,s'\}| =1$, which proves  that $\tilde{\mathcal{S}}$ is indeed a drop certificate for $X$.
\end{proof}

The following is the key property of locally $\gamma$-MST-optimal instances that we will use to prove Lemma~\ref{lem:dual_feasible}.

\begin{corollary}\label{cor:lower_bound_components}
Let $X \subseteq R$ and $\tilde{\mathcal{S}} \subseteq \mathcal{S}$ be a drop certificate for $X$.
Then every component $K$ connecting the terminals in $X$ has cost 
$c(K)> \frac{2}{1+\gamma} \sum_{S\in \tilde{\mathcal{S}}} d^{S}$.
\end{corollary}
\begin{proof}
Because our instance $(G,R)$ of the Steiner tree problem is locally $\gamma$-MST-optimal, we have $c(K)> \frac{1}{1+\gamma} \cdot \drop(X)$ for every component $K$ connecting $X$.
By Lemma~\ref{lem:drop_certificates}, the value $2 \sum_{S\in \tilde{\mathcal{S}}} d^{S}$ of any drop certificate $\tilde{\mathcal{S}}$ for the terminal set $X$ is a lower bound on $\drop(X)$.
\end{proof}

\subsection{\boldmath $S$-tight paths}
\label{sec:S-tight}

We now analyze the dual growth procedure described in Section~\ref{sec:construction_dual_solution} and prove Lemma~\ref{lem:dual_feasible}.
From now on, we will only work with the bidirected graph $(V^{\epsilon}, \overrightarrow{E}^{\epsilon})$, but not with the original graph $G$ anymore.

We need to show that for any time $t$ and for every active set $S\in \mathcal{S}^t$, the set $U^t_S$ does not contain the root~$r$.
Recall that $U^t_S$ is the set of vertices $v$ that are reachable from the set $S$ by $\delta$-tight edges at time $t$.
For a vertex $v$, we denote by $t^{Sv}$ the first time by which $v$ is reachable from $S$ by $\delta$-tight edges.

By the definition of $U^t_S$, for every vertex $v\in U^t_S$ there is an $S$-$v$ path of $\delta$-tight edges at time $t$.
(An $S$-$v$ path is an $s$-$v$ path for some $s\in S$.)
There might be many such paths and we will consider one such path $P$ with particular properties.
Informally speaking, we want the path $P$ to be a ``first path by which we reached $v$ from $S$''.
Moreover, also for every other vertex $u$ of $P$, we want the $s$-$u$ subpath of $P$ to be a ``first path by which we reached $u$ from $S$''.
More precisely, for a directed path $P$ and vertices $a$ and $b$ of $P$ with $a$ lying before $b$ on $P$, we denote by $P_{[a,b]}$, the $a$-$b$ subpath of $P$. We want the $s$-$v$ path $P$ (with $s\in S$) to have the following property.
For every vertex~$u$ of~$P$, all edges of the $s$-$u$ subpath $P_{[s,u]}$ of $P$ should be $\delta$-tight already at time $t^{Su}$.

An equivalent way of expressing the two above mentioned properties of $P$ is to require all edges of~$P$ to be \emph{$S$-tight}, which we define as follows.

\begin{definition}
An edge $e=(u,v) \in \overrightarrow{E}^{\epsilon}$  is $S$-tight if $e$ is $\delta$-tight at time $t^{Sv}$ and $t^{Su} \leq t^{Sv}$.
\end{definition}

\begin{definition}
For a set $S\in \mathcal{S}$ and a vertex $v\in V^{\epsilon}$, we say that an $S$-$v$ path is $S$-tight if all its edges are $S$-tight.
\end{definition}

We highlight that for every vertex $u$ of an $S$-tight path $P$ starting at $s\in S$, the subpath $P_{[s,u]}$ is also $S$-tight.
A simple induction on the length of the path $P$ implies that indeed every $S$-tight path $P$ has the property that for every vertex $u$ of $P$ all edges of $P_{[s,u]}$ are $\delta$-tight at time $t^{Su}$.

Moreover, we observe that whenever there exists some $S$-$v$ path consisting of $\delta$-tight edges, then there also exists an $S$-tight $S$-$v$ path.

\begin{lemma}\label{lem:reachability_implies_S-tight_path}
If for a set $S\subseteq R$, a vertex $w$ is reachable from $S$ by $\delta$-tight edges, then there exists an $S$-tight path (from $S$) to $w$:
\end{lemma}
\begin{proof}
We prove the statement for every vertex $w$ in an order of non-decreasing $t^{Sw}$.
Thus, when proving the statement for a vertex $w$, we may assume that it holds for every vertex $v$ with $t^{Sv} < t^{Sw}$.

We consider any $S$-$w$ path $Q$ consisting of $\delta$-tight edges at time $t^{Sw}$, which exists by definition of the time $t^{Sw}$.
Then we have $t^{Su} \le t^{Sw}$ for every vertex $u$ of $Q$.
If $t^{Sw} =0$, the path $Q$ is $S$-tight and hence we may assume that this is not the case.
We consider the last vertex $v$ of $Q$ with $t^{Sv} < t^{Sw}$.
Then there exists an $S$-tight $S$-$v$ path $Q'$.
Because $t^{Su} \le t^{Sv} < t^{Sw}$ for every vertex $u$ of $Q'$, the only vertex of $Q_{[v,w]}$ that is contained in $Q'$ is the vertex $v$.
Appending the path $Q_{[v,w]}$ to the $S$-tight path $Q'$ yields the desired $S$-tight $S$-$w$ path $P$.
\end{proof}

In particular, Lemma~\ref{lem:reachability_implies_S-tight_path} implies that for every vertex $v$ in the set $U^t_S$ for $S\in \mathcal{S}^t$, there exists some $S$-tight path ending at $v$.

The key step of our proof of Lemma~\ref{lem:dual_feasible} will  be to bound the maximum length of $S$-tight paths.
For a path $P$, we write $c(P)$ to denote its length, i.e., the total cost of its edges.
Because our goal is to show that the root $r$ is not contained in $U_S^t$ for active sets $S$ at time $t$, we will only consider $S$-tight paths existing before the deactivation time $d^S$ of the set $S$. 

\begin{lemma}\label{lem:distance_bound}
Let $\beta \coloneqq \betaval$.
Let $t\in [0,t_{\max})$ and $S\in \mathcal{S}$ with $t < d^S$ and $r\notin S$.
Then every $S$-tight path $P$ has length
\[
c(P) \le (1+\beta\delta)\cdot t.
\]
\end{lemma}

We prove Lemma~\ref{lem:distance_bound} in Section~\ref{sec:distance_bound}.
First, we show that it implies Lemma~\ref{lem:dual_feasible}.

\dualfeasible*
\begin{proof}Consider a time $t\in [0,t_{\max})$ and an active set $S\in \mathcal{S}^t$ with $r\notin S$.
For every vertex $v\in U^t_S$, there exists an $S$-tight $S$-$v$ path by Lemma~\ref{lem:reachability_implies_S-tight_path}.
Then Lemma~\ref{lem:distance_bound} implies that this path has length at most $(1+\beta\delta)\cdot t$ and hence
\begin{equation}\label{eq:dist_root}
\dist(S,v)\ \coloneqq\ \min_{s\in S} \dist(s,v)\ \le\ \highlight{(1+\beta\delta)\cdot t \ <\ 2 \cdot t}.
\end{equation}
Because $r\notin S$, we have $d^S \le \tmerge(s,r)$ for all $s\in S$ by Observation~\ref{obs:deactivation_time_lower_bounds_merge_time}.
Thus, by Observation~\ref{obs:dist_lb_merge_time}, we have
$\dist(S,r) \ \ge\ 2 \cdot d^S > 2 \cdot t$.
This implies $r\notin U^t_S$. 
\end{proof}
We  remark (although we will not need this fact) that the same argument implies that for any active set $S$ at time $t$, no terminal $\overline{s}\in R\setminus S$ is contained in $U^t_S$, i.e., the terminals contained in $U^t_S$ are exactly those in~$S$.
The special case $\overline{s}=r$ is the statement on Lemma~\ref{lem:dual_feasible}.

The reason for choosing the constant $\delta=\deltaval$ smaller than required for \eqref{eq:dist_root} is that we will need this small value of $\delta$ to guarantee various structural properties in the proof of Lemma~\ref{lem:distance_bound}.

It now remains to prove Lemma~\ref{lem:distance_bound}.

\section{Proof of Lemma~\ref{lem:distance_bound}}
\label{sec:distance_bound}

We continue to work with the graph  $(V^{\epsilon}, \overrightarrow{E}^{\epsilon})$  throughout this section, where we prove Lemma~\ref{lem:distance_bound}.
We observe that the set of $\delta$-tight edges changes only a finite number of times at discrete time steps.
In particular, the set of $S$-tight paths for sets $S\in \mathcal{S}$ changes only a finite number of times at discrete time steps.
This will allow us to prove Lemma~\ref{lem:distance_bound} by induction over time.
In fact, we will prove a stronger statement than Lemma~\ref{lem:distance_bound} by induction.
This will be useful because we then have a stronger induction hypothesis that we can use in our induction step.
We start by explaining the statement we prove by induction.

\subsection{Inductive argument to construct components}
\label{sec:induction}

To explain the statement we prove by induction, we need the following notation.
For a set $X$ of terminals and a terminal $s\in X$, we write
\[ \textstyle
   \drop(X,s) \coloneqq \max \Bigl\{  2 \sum_{S\in \tilde{\mathcal{S}}} d^{S} : \tilde{\mathcal{S}}\text{ is a drop certificate for }X\text{ and }\tilde{\mathcal{S}} \subseteq 2^{V\setminus \{s\}}  \Bigr\}
\]
to denote the maximum value of a drop certificate consisting only of sets not containing the vertex $s$.
Considering such drop certificates with sets not containing a particular vertex $s$ will be useful when extending a component to a larger component and constructing a drop certificate for the larger component from the drop certificate of the original component.
We remark that by Corollary~\ref{cor:lower_bound_components} we have $c(K) - \tfrac{1}{1+\gamma} \cdot \drop(X, s) > 0$ for every component $K$ connecting a terminal set $X$ with $s\in X$.

Let us fix constants $\alpha \coloneqq \alphaval$, $\mu\coloneqq \muval$, and $\lambda\coloneqq \lambdaval$.
We are now ready to state the claim that we will prove by induction.
For an illustration of~\eqref{item:strong_existing_component}, see Figure~\ref{fig:claim_induction}.

\begin{claim}\label{claim:main}
Consider a time $t \in [0, t_{\max})$, a set $S \in \mathcal{S}$ with $t <  d^{S}$ and $r\notin S$, and an $S$-tight path $P$ (at time $t$).
Let $s$ be the start vertex of $P$, which is a terminal, and let $w$ be the end vertex of $P$, which can be an arbitrary vertex.
Then we have the following:
\begin{enumerate}[(a)]\itemsep6pt
\item \label{item:strong_existing_component}
There exists a terminal set $X\subseteq R$ with $s\in X$  and a component $K$ connecting $X\cup \{w\}$ such that
\begin{align*}
c(P) \ \le&\ (1+\delta) \cdot t + \lambda \cdot \drop(X, s) \\
c(K) - \tfrac{1}{1+\gamma} \cdot \drop(X, s) \ \le&\  (1+\delta) \cdot t - \mu \cdot \drop(X,s)  \\
\tmerge(s,s') \ <&\  (1+\alpha\delta)\cdot t &\text{ for all } s' \in X.
\end{align*}

\item \label{item:actual_distance_bound}
 We have $c(P) \le (1+\beta\delta) \cdot t$.
\end{enumerate}
\end{claim}
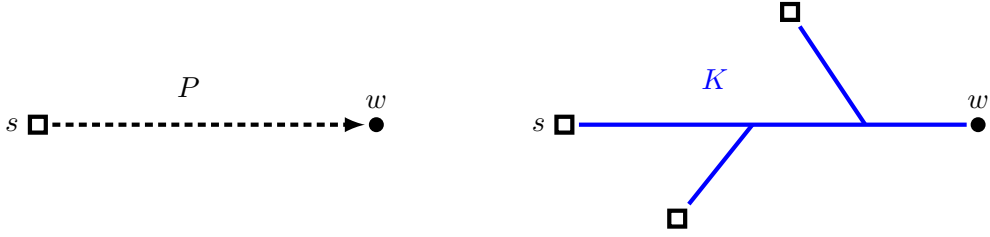
\begin{figure}
\begin{center}
\begin{tikzpicture}[very thick, yscale=0.5, xscale=1]
\tikzset{terminal/.style={
ultra thick,draw,fill=none,rectangle,minimum size=0pt, inner sep=3pt, outer sep=2.5pt}
}
\tikzset{steiner/.style={
fill=black,circle,inner sep=0em,minimum size=0pt, inner sep=2pt, outer sep=1.5pt}
}

\begin{scope}[shift={(-7,0)}]
\node[terminal] (start) at (-1.5, -4) {};
\node[left=3pt] at (start) {$s$};
\node[steiner] (u1) at (3,-4) {};
\node[above=2pt] at (u1) {$w$};
\begin{scope}[ultra thick, black, densely dashed, ->, >=latex]
\draw (start) to (u1);
\end{scope}
\node at (0.5,-3) {$P$};
\end{scope}

\node[terminal] (start) at (-1.5, -4) {};
\node[left=3pt] at (start) {$s$};
\node[terminal] (s1) at (0, -6.5) {};
\node[terminal] (s2) at (1.5, -1) {};
\node[steiner] (u1) at (4,-4) {};
\node[above=2pt] at (u1) {$w$};

\coordinate (m1) at (1,-4);
\coordinate (m2) at (2.5,-4);

\begin{scope}[ultra thick, blue]
\draw (start) -- (m1) -- (m2) -- (u1);
\draw (s1) -- (m1);
\draw (s2) -- (m2);
\end{scope}

\node[blue] at (0.5,-2.8) {$K$};

\end{tikzpicture}
\caption{\label{fig:claim_induction}
Illustration of \eqref{item:strong_existing_component} of Claim~\ref{claim:main}.
On the left, we see the path $P$ from $s$ to $w$.
On the right, we see the component $K$ connecting $w$ and a set $X$ of terminals (shown as squares).
}
\end{center}
\end{figure}

Let us provide a few observations about these technical statements.
First, we observe that \eqref{item:actual_distance_bound} is exactly what we claimed in Lemma~\ref{lem:distance_bound} and thus proving \eqref{item:strong_existing_component} and \eqref{item:actual_distance_bound} will complete the proof of Lemma~\ref{lem:distance_bound}.
Moreover, we observe that \eqref{item:actual_distance_bound}  is a direct consequence of \eqref{item:strong_existing_component} because by Corollary~\ref{cor:lower_bound_components} we have 
\[
0 \le c(K) - \tfrac{1}{1+\gamma} \cdot \drop(X, s)\ \le\  (1+\delta) \cdot t - \mu \cdot \drop(X,s),
\]
implying $\drop (X,s) \le \frac{1+\delta}{\mu} \cdot t$ and hence
\[
c(P) \le (1+\delta) \cdot t + \lambda \cdot \drop(X, s)\ \le\ \highlight{(1+\delta) \cdot (1+\tfrac{\lambda}{\mu}) \cdot t \ \le\ (1+\beta\delta) \cdot t},
\]
where we refer to Table~\ref{table:constants} for the values of the constants used in the last inequality.

Let us also briefly comment on the condition $\tmerge(s,s') \ <\  (1+\alpha\delta)\cdot t$ for all $s' \in X$.
It will be used in our proof when combining two components $K_1$ and $K_2$, connecting terminal sets $X_1$ and $X_2$, respectively to a new component connecting $X_1 \cup X_2$.
(The two components $K_1$ and $K_2$ arise by applying the induction hypothesis to different paths $P$.)
Then we will use this condition to argue that the terminal sets $X_1$ and $X_2$ are disjoint and to show that we can use drop certificates for $X_1$ and $X_2$ to construct a new drop certificate for $X_1 \cup X_2$.
See the proof of Lemma~\ref{lem:lower_bound_meeting_time} for details.

\bigskip

Recall  that the set of $\delta$-tight edges , and thus the set of $S$-tight paths for sets $S\in \mathcal{S}$, changes only a finite number of times at discrete time steps.
This allows us to prove Claim~\ref{claim:main} by induction over time.
At time~$0$ only edges of cost $0$ are $\delta$-tight and hence the claim trivially holds (with $X = \{s\}$).
We now consider a time $t^*> 0$ and may assume by induction that Claim~\ref{claim:main} holds for all times $t < t^*$ (for every set $S\in \mathcal{S}$ with $d^S \ge t$ and $r\notin S$ and for every $S$-tight path existing at time $t$).

To prove  Claim~\ref{claim:main} for the current time $t^*$ and a fixed set $S^*\in \mathcal{S}$ with $d^{S^*} > t^*$ and $r\notin S^*$, we apply induction over the number of edges of the path.
We fix an $S^*$-tight path  $P^*$ existing at time $t^*$ and may assume by the induction hypothesis that Claim~\ref{claim:main} holds for all $S^*$-tight paths with strictly fewer edges than $P^*$ (existing at the current time $t^*$).
We have already observed above that \eqref{item:actual_distance_bound} follows from \eqref{item:strong_existing_component} and Corollary~\ref{cor:lower_bound_components}.
Hence, it suffices to prove~\eqref{item:strong_existing_component}.

\bigskip
We start by bounding the total contribution of sets containing the start vertex $s^*$ of $P^*$ to edges of $P^*$ (Section~\ref{sec:base_contribution}).
Next, in Section~\ref{sec:structure_contributing}, we analyze the structure of other sets contributing to $P^*$ and prove that these can be partitioned into chains $\Cscr^1,\dots,\Cscr^k$ as outlined in Section~\ref{sec:outline_analysis}.
In Section~\ref{sec:separating}, we then prove that we can subdivide the path $P^*$ into two subpaths (overlapping in only one edge which has length at most $\epsilon$ and is hence negligible) such that on the first subpath only sets from the chains $\Cscr^1,\dots, \Cscr^{k-1}$ (and sets containing $s^*$) contributed and on the second subpath only sets from the chain $\Cscr^k$ (and sets containing~$s^*$) contributed.
To the first subpath we will apply the induction hypothesis and the length of the second subpath will be analyzed in Section~\ref{sec:single_chain}.
Finally, Section~\ref{sec:completing_induction} describes how to put these ingredients together to prove~\eqref{item:strong_existing_component}. 

\subsection{Bounding the contribution of sets containing \boldmath $s^*$}
\label{sec:base_contribution}

Let $s^*$ be the start vertex of the path $P^*$ and let $w^*$ be the end vertex of $P^*$.
In this section we analyze the total contribution of sets containing $s^*$ to the path $P^*$ (until the current time $t^*$).

The following simple observation will be useful.

\begin{observation}\label{obs:monotonicity_reaching_times}
For a vertex $v$ and two sets $A, B \in \mathcal{S}$ with $A\subseteq B$, we have $t^{Bv} \leq t^{Av}$.
\end{observation}

\begin{lemma}\label{lem:contribution_basic}
For every edge $(v,w)$ of the path $P^*$, the total contribution of terminal sets $S\in \mathcal{S}$ with $s^*\in S$ is at most $t^{S^*w}-t^{S^*v}$.
\end{lemma}
\begin{proof}
Consider a set $S\in \mathcal{S}$ with $s^*\in S$ that contributed to an edge of $P^*$ (before the current time $t^*$), implying $a^S < t^* < d^{S^*}$ (where the last inequality holds by the assumption of Claim~\ref{claim:main}).
We conclude that $S^*$ cannot be a proper subset of $S$.
Hence, because $s^*\in S \cap S^*$ and $\mathcal{S}$ is a laminar family, we have $S\subseteq S^*$.
Then by Observation~\ref{obs:monotonicity_reaching_times}, we have for every vertex $v$,
\[
  t^{\{s^*\}v} \ \ge\ t^{Sv} \ \ge\ t^{S^*v}.
\]
Because $P^*$ is an $S^*$-tight path starting at the terminal $s^*\in S^*$, we have for every vertex $v$ of $P^*$ that $t^{\{s^*\}v} = t^{S^*v}$ and hence $t^{Sv} = t^{S^*v}$.
Because $P^*$ is $S^*$-tight, we have  $t^{S^*v}  \le t^{S^*w}$ for every edge $(v,w)$ of $P^*$.
By the definition of the set $U^t_S$ for time $t$, an edge $(v,w)$ is only outgoing of $U^t_S$ if $t^{S^*v} = t^{Sv} \le t < t^{Sw}=t^{S^*w} $.
Thus, at most one outgoing edge of $U^t_S$ is contained in $P^*$.
Because at every time $t$ there is at most one active set $S \in \mathcal{S}^t$ containing $s^*$, this completes the proof.
\end{proof}

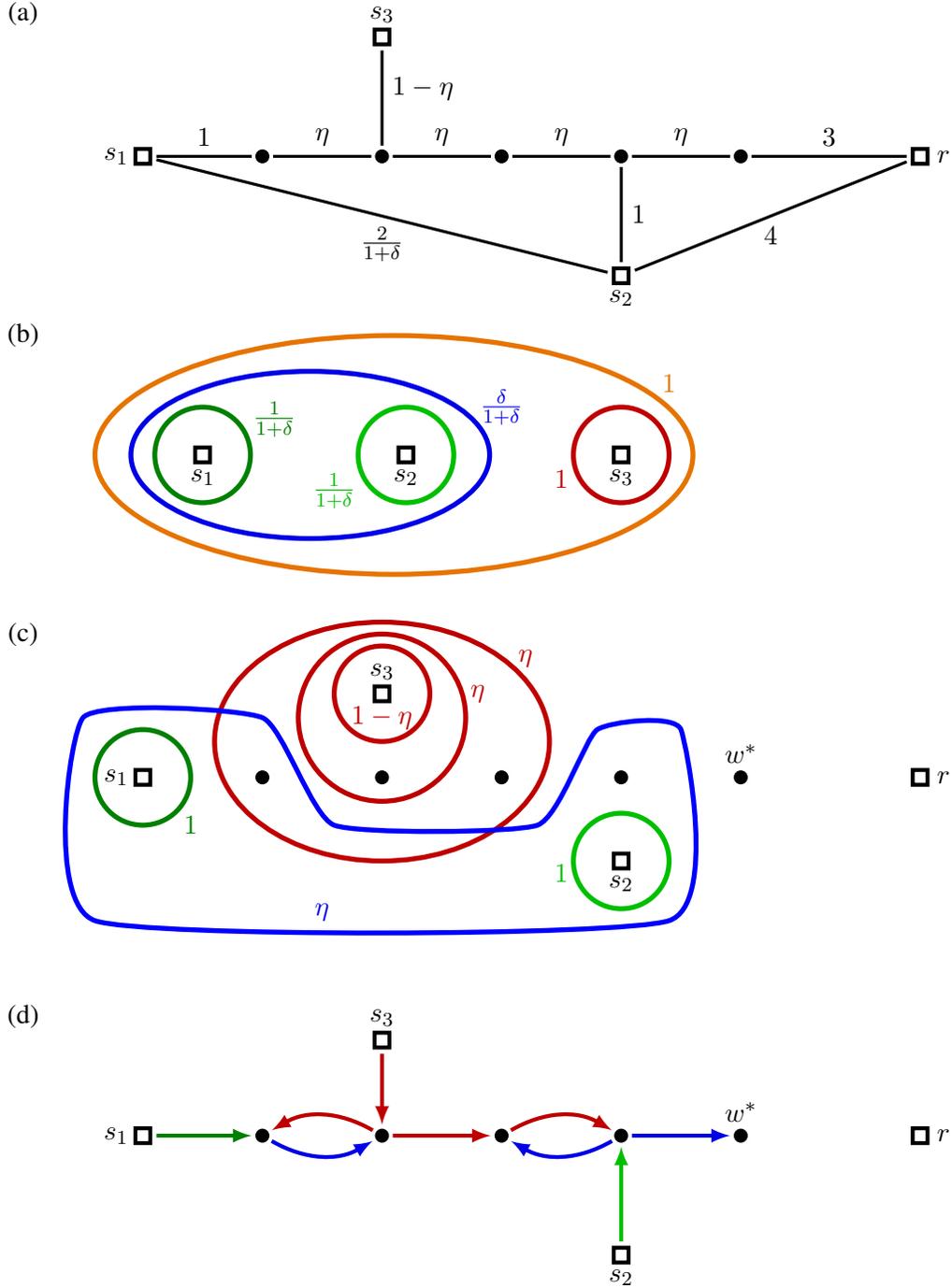
\begin{figure}
\begin{center}
\begin{tikzpicture}[scale=1.7]
\tikzset{terminal/.style={
ultra thick,draw,fill=none,rectangle,minimum size=0pt, inner sep=3pt, outer sep=2.5pt}
}
\tikzset{steiner/.style={
fill=black,circle,inner sep=0em,minimum size=0pt, inner sep=2pt, outer sep=1.5pt}
}
\tikzset{dual/.style={line width=2pt}}

\begin{scope}[shift={(0,15)}]
\node at (-1,2.2) {(a)};

\begin{scope}[every node/.style=terminal]
\node (s1) at (0,1) {};
\node (s2) at (4,0) {};
\node (s3) at (2,2) {};
\node (r) at (6.5,1) {};
\end{scope}
\node[left=2pt] at (s1) {$s_1$};
\node[below=2pt] at (s2) {$s_2$};
\node[above=2pt] at (s3) {$s_3$};
\node[right=3pt] at (r) {$r$};

\begin{scope}[every node/.style=steiner]
\node (v1) at (1,1) {};
\node (v2) at (2,1) {};
\node (v3) at (3,1) {};
\node (v4) at (4,1) {};
\node (w) at (5,1) {};
\end{scope}

\begin{scope}[very thick]
\draw (s1) -- node[above] {$1$} (v1) -- node[above] {$\eta$}  (v2) -- node[above] {$\eta$}  (v3) --   node[above] {$\eta$} (v4) --  node[above] {$\eta$} (w) -- node[above] {$3$} (r);
\draw (s3) -- node [right, pos=0.4]{$1-\eta$} (v2);
\draw (s2) -- node [right]{$1$} (v4);
\draw (s1) -- node[below] {$\frac{2}{1+\delta}$} (s2); 
\draw (s2) -- node[below] {$4$} (r);
\end{scope}

\end{scope}

\begin{scope}[shift={(0,12.5)}]
\node at (-1,2) {(b)};

\begin{scope}[every node/.style=terminal]
\node (s1) at (0.5,1) {};
\node (s2) at (2.2,1) {};
\node (s3) at (4,1) {};
\end{scope}
\node[below=2pt] at (s1) {$s_1$};
\node[below=2pt] at (s2) {$s_2$};
\node[below=2pt] at (s3) {$s_3$};

\begin{scope}[dual]
\draw[green!50!black] (s1) ellipse (0.4 and 0.4);
\draw[green!75!black] (s2) ellipse (0.4 and 0.4);
\draw[red!75!black] (s3) ellipse (0.4 and 0.4);
\draw[blue!90!black] (1.4,1) ellipse (1.5 and 0.7);
\draw[orange!90!black] (2.1,1) ellipse (2.5 and 1);
\end{scope}
\node[green!50!black] at (1.1,1.3) {$\frac{1}{1+\delta}$};
\node[green!75!black] at (1.6,0.7) {$\frac{1}{1+\delta}$};
\node[blue!90!black] at (3,1.4) {$\frac{\delta}{1+\delta}$};
\node[red!75!black] at (3.5,0.8) {$1$};
\node[orange!90!black] at (4.4,1.6) {$1$};

\end{scope}

\begin{scope}[shift={(0,9.8)}]
\node at (-1,2.2) {(c)};

\begin{scope}[every node/.style=terminal]
\node (s1) at (0,1) {};
\node (s2) at (4,0.3) {};
\node (s3) at (2,1.7) {};
\node (r) at (6.5,1) {};
\end{scope}
\node[left=2pt] at (s1) {$s_1$};
\node[below=2pt] at (s2) {$s_2$};
\node[above=2pt] at (s3) {$s_3$};
\node[right=3pt] at (r) {$r$};

\begin{scope}[every node/.style=steiner]
\node (v1) at (1,1) {};
\node (v2) at (2,1) {};
\node (v3) at (3,1) {};
\node (v4) at (4,1) {};
\node (w) at (5,1) {};
\end{scope}
\node[above=2pt] at (w) {$w^*$};

\begin{scope}[dual]
\draw[green!50!black] (s1) ellipse (0.4 and 0.4);
\draw[green!75!black] (s2) ellipse (0.4 and 0.4);
\draw[red!75!black] (s3) ellipse (0.4 and 0.4);
\draw[red!75!black] (2,1.5) ellipse (0.7 and 0.7);
\draw[red!75!black] (2,1.3) ellipse (1.4 and 1);
\draw [blue] plot [smooth cycle, tension=0.3] coordinates {(-0.4,-0.2)(-0.5,1.5) (1,1.5) (1.6,0.6) (3.3,0.6) (3.8, 1.4) (4.5, 1.4) (4.4,-0.2) };
\end{scope}

\node[red!75!black] at (2,1.5) {$1-\eta$};
\node[red!75!black] at (2.8,1.7) {$\eta$};
\node[red!75!black] at (3.2,2) {$\eta$};
\node[green!50!black] at (0.4,0.6) {$1$};
\node[green!75!black] at (3.5,0.2) {$1$};
\node[blue!90!black] at (1.5,-0.13) {$\eta$};
\end{scope}

\begin{scope}[shift={(0,6.8)}]
\node at (-1,2) {(d)};

\begin{scope}[every node/.style=terminal]
\node (s1) at (0,1) {};
\node (s2) at (4,0) {};
\node (s3) at (2,1.8) {};
\node (r) at (6.5,1) {};
\end{scope}
\node[left=2pt] at (s1) {$s_1$};
\node[below=2pt] at (s2) {$s_2$};
\node[above=2pt] at (s3) {$s_3$};
\node[right=3pt] at (r) {$r$};

\begin{scope}[every node/.style=steiner]
\node (v1) at (1,1) {};
\node (v2) at (2,1) {};
\node (v3) at (3,1) {};
\node (v4) at (4,1) {};
\node (w) at (5,1) {};
\end{scope}
\node[above=2pt] at (w) {$w^*$};

\begin{scope}[ultra thick, ->, >=latex]
\draw[green!50!black] (s1) to (v1);
\draw[green!75!black] (s2) to (v4);
\draw[red!75!black] (s3) to (v2);
\draw[red!75!black, bend right] (v2) to (v1);
\draw[red!75!black] (v2) to (v3);
\draw[red!75!black, bend left] (v3) to (v4);
\draw[blue!90!black, bend right] (v1) to (v2);
\draw[blue!90!black, bend left] (v4) to (v3);
\draw[blue!90!black] (v4) to (w);
\end{scope}

\end{scope}

\end{tikzpicture}
\caption{\label{fig:tight_path_needed}
An example illustrating that the conditions $d^{S^*} > t^*$ and $P^*$ being $S^*$-tight are necessary. 
We choose $\eta$ such that $0 < \eta \le \delta$.
The instance is given by the (metric closure of) the graph shown in (a).
The sets in $\Sscr$ not containing $r$ are shown in (b). 
The number next to a set $S$ shows the length $d^S-a^S$ of the time interval where $S$ is active.
(c) shows the (scaled) dual solution $(1+\delta)z^{t^*}$ at time $t^*= \frac{1+\eta}{1+\delta}$. 
The sets are shown in the same color as their corresponding set in (b).
(d) shows the $\delta$-tight edges at time $t^*$, where edges are colored by the set contributing to it. (In this example, this is only one set for each edge.)
There is only one path from $s_1$ to $w^*$ consisting of $\delta$-tight edges. This path is $\{s_1\}$-tight (but the deactivation time of $\{s_1\}$ is less than $t^*$) but it is not $\{s_1,s_2\}$-tight.
The contribution of sets containing $s_1$ is $\frac{1+2\eta}{1+\delta} > t^*$.
}
\end{center}
\end{figure}

Lemma~\ref{lem:contribution_basic} implies that the total contribution of the terminal sets $S\in \mathcal{S}$ with $s^*\in S$ to the path $P^*$ is at most 
\[
\sum_{(v,w)\in P^*} (t^{S^*w}-t^{S^*v}) = t^{S^*w^*}-t^{S^*s^*} \le t^* - 0 = t^*.
\]
Note that we crucially used both $d^{S^*} \ge t^*$ and the fact that $P^*$ is $S^*$-tight. See Figure~\ref{fig:tight_path_needed} for an example that these conditions are necessary.
\bigskip

We conclude that if every set that contributed to $P^*$ contains $s^*$, we have $t^* \ge \frac{1}{1+\delta} \cdot c(P^*)$ because all edges of $P^*$ are $\delta$-tight.
Then $c(P^*) \le (1+\delta) \cdot t^*$, which implies Claim~\ref{claim:main}  for $K=P^*$ and $X=\{s^*\}$.

The main difficulty when proving Claim~\ref{claim:main} is to bound the contribution of sets not containing the terminal~$s^*$.
Next, we will investigate the structure of these sets.

\subsection{Structure of the sets contributing to the path \boldmath $P^*$}
\label{sec:structure_contributing}

In this section we investigate the structure of the sets that contributed to edges of the path~$P^*$ (before the current time $t^*$).
We denote the collection of sets that contributed to $P^*$ but do not contain the vertex $s^*$ by $\mathcal{C} \subseteq \mathcal{S}$.
In Section~\ref{sec:base_contribution}, we have shown Claim~\ref{claim:main} in the special case where $\mathcal{C}= \emptyset$.
From now on, we assume that $\mathcal{C}$ is nonempty.

We will prove that we can partition $\mathcal{C}$ into chains $\mathcal{C}^1, \dots, \mathcal{C}^k$ such that any two sets from different chains are disjoint; see Figure~\ref{fig:example_contributing_sets}.
Moreover, we will prove that every set from $\mathcal{C}^1\cup \dots \cup\mathcal{C}^{k-1}$ is a subset of $S^*$.

\bigskip

Consider a set $S\in \mathcal{C}$. 
Then, because $S\in \mathcal{S}$ and $s^*\notin S$, the merge time $\tmerge(s,s^*)$ is the same for all terminals $s \in S$.
We denote this time by $\tmerge(S, s^*)$.
Because for every set $S\in \mathcal{C}$ we have $s^* \notin S$, Observation~\ref{obs:deactivation_time_lower_bounds_merge_time} implies the following.

\begin{observation}\label{obs:lower_bound_merge_time_sets}
For every set $S\in \mathcal{C}$, we have $\tmerge(S, s^*) \ge d^S$.
\end{observation}

We now prove that a set $S\in\mathcal{C}$ cannot contribute to edges of the path $P^*$ that are too far away from the start vertex $s^*$ of $P^*$.

\begin{lemma}\label{lem:simple_bound_path_length_to_contribution}
If a set $S\in \mathcal{C}$ contributed to an edge $(u,v)$ of the path $P^*$, then
 \[
c(P^*_{[s^*,u]})\ \le\ (1+\beta\delta) \cdot \tmerge(S,s^*).
\]
\end{lemma}
\begin{proof}
First, we show $t^{S^*u} \le \tmerge(S,s^*)$.
If $ \tmerge(S,s^*) \ge t^*$, this follows from $t^{S^*u} \le t^* \le \tmerge(S, s^*)$, where we used that all edges of the path $P^*_{[s^*,u]}$ are $\delta$-tight at the current time $t^*$.

If $\tmerge(S,s^*) < t^*$, the fact that $t^* < d^{S^*}$ (by the assumptions of Claim~\ref{claim:main}) implies $S \subsetneq S^*$ by Observation~\ref{obs:deactivation_time_lower_bounds_merge_time}.
Hence, using Observation~\ref{obs:monotonicity_reaching_times} and Observation~\ref{obs:lower_bound_merge_time_sets} we then get $t^{S^*u} \le t^{Su} \le d^{S} \le \tmerge(S,s^*) $, where for the second last inequality we used that $S$ contributed to the edge $(u,v)$ before time $d^{S}$.
We conclude $t^{S^*u} \le \tmerge(S,s^*)$ in both cases.

Because all edges of the path $P^*_{[s^*,u]}$ are $S^*$-tight (using that $P^*$ is $S^*$-tight), the path $P^*_{[s,u]}$ is an $S^*$-tight path and it has strictly fewer edges than $P^*$ (using that $(u,v)$ is an edge of $P^*$).
Therefore, we can apply \eqref{item:actual_distance_bound} from the induction hypothesis to obtain 
$c(P^*_{[s^*, u]}) \le (1+\beta\delta)\cdot t^{S^*u} \le (1+\beta\delta)\cdot \tmerge(S,s^*)$.
\end{proof}

We now prove that for every set $S\in \mathcal{C}$, its merge time $\tmerge(S, s^*)$ is approximately equal to its deactivation time $d^{S}$.
Recall that by Observation~\ref{obs:lower_bound_merge_time_sets} we have $\tmerge(S,s^*) \ge d^S$.
The following lemma shows that the deactivation time $d^S$ cannot be much smaller than the merge time $\tmerge(S, s^*)$.

\begin{lemma}\label{lem:merge_time_approx_deactivation_time}
For every set $S\in \mathcal{C}$, we have
\[
\min\{t^*,d^S\}\ >\ \tfrac{1-\beta\delta}{1+\beta\delta}\cdot  \tmerge(S,s^*).
\]
\end{lemma}
\begin{proof}
By the definition of $\tmerge(S,s^*)$ and Observation~\ref{obs:dist_lb_merge_time}, we have $\dist(S,s^*) \ge 2 \cdot \tmerge(S, s^*)$.
Let $(u,v)$ be an edge of $P^*$ to which $S$ contributed.
Then $t^{Su} < \min \{t^*,d^{S}\}$.
By Lemma~\ref{lem:reachability_implies_S-tight_path} there existed an $S$-tight path $P$ from $S$ to $u$ at time $t^{Su} < \min \{t^*,d^{S}\}$.
Using \eqref{item:actual_distance_bound} of the induction hypothesis applied to $P$ at time $t^{Su}$ and Lemma~\ref{lem:simple_bound_path_length_to_contribution} we conclude
\begin{align*}
2 \cdot \tmerge(S,s^*)\ \le&\ \dist(S, s^*)\\
 \le&\ \dist(S,u)  + c(P^*_{[s^*,u]}) \\
 <&\  (1+\beta\delta) \cdot \min\{t^*,d^{S}\} +(1+\beta\delta)\cdot t_{\mathrm{merge}}(S, s^*),
\end{align*}
implying $\min\{t^*, d^{S} \} > \tfrac{1-\beta\delta}{1+\beta\delta}\cdot  \tmerge(S,s^*)$.
\end{proof}

Next, we prove that any two disjoint sets from $\mathcal{C}$ have a very different merge time with $s^*$.
Let us fix a constant 
\[
M\ \coloneqq\  \frac{\frac{2}{1+\gamma}-1-\beta\delta}{\frac{(1+\beta\delta)^2}{1-\beta\delta}-(\frac{2}{1+\gamma}-1-\beta\delta)}\approx \Mval.
\]
For intuition, we advise the reader to think about $M$ as any sufficiently large constant.
(In fact it would be possible to choose $M$ arbitrarily large by choosing $\delta$ and $\gamma$ small enough, which comes at the cost of a worse but still better-than-$2$ integrality gap upper bound for \ref{eq:bcr-tree}.)

\begin{lemma}\label{lem:different_merge_times}
Let $S_1,S_2 \in \mathcal{C}$ with $S_1 \cap S_2 =\emptyset$, where without loss of generality  we have $\tmerge(S_1,s^*) \le \tmerge(S_2,s^*)$.
Then we have $\tmerge(S_2,s^*) > M \cdot \tmerge(S_1, s^*)$.
\end{lemma}
\begin{proof}

Let $e_1=(v_1,w_1),e_2=(v_2,w_2)$ be the first edges of $P^*$ where the sets $S_1,S_2$ contributed, respectively.
Let $j\in \{1,2\}$ such that $v_j$ is no earlier on $P^*$ than the other vertex from $v_1,v_2$.
By Lemma~\ref{lem:simple_bound_path_length_to_contribution}, we have
\begin{equation}\label{lem:different_merge_times:eq1}
c(P^*_{[s^*,v_j]}) \ \le\ (1+\beta\delta)\cdot \tmerge(S_j,s^*)\ \le\ (1+\beta\delta)\cdot \tmerge(S_2,s^*).
\end{equation}
For $i\in\{1,2\}$, the set $S_i$ contributed to the edge $(v_i,w_i)$ and hence $t^{S_iv_i} < \min\{t^*, d^{S_i} \}$.
By Lemma~\ref{lem:reachability_implies_S-tight_path}, there exists an $S_i$-tight path from $S_i$ to $v_i$ at time $t^{S_iv_i} < \min\{ t^*, d^{S_i}\}$.
Applying \eqref{item:actual_distance_bound} of the induction hypothesis to this $S_i$-tight path yields
\begin{equation}\label{lem:different_merge_times:eq2}
\dist(S_i,v_i) \ \leq\ (1+\beta\delta)t^{S_iv_i}\ <\ (1+\beta\delta) \cdot d^{S_i}.
\end{equation}
By combing $P_{[s^*,v_j]}$ with a shortest $S_1$-$v_1$ path and a shortest $S_2$-$v_2$ path we obtain a component $K$ connecting terminals $s^*,s_1,s_2$ with $s_1\in S_1$ and $s_2\in S_2$.
By  \eqref{lem:different_merge_times:eq1} and \eqref{lem:different_merge_times:eq2} the component $K$ has cost
\[
c(K) \ <\ (1+\beta\delta)\cdot (d^{S_1} + d^{S_2}) + (1+\beta\delta)\cdot \tmerge(S_2,s^*).
\]
Because $S_1$ and $S_2$ are disjoint and both do not contain $s^*$,  we conclude that $\{S_1, S_2\}$ is a drop certificate for $\{ s^*,s_1,s_2\}$.
By Corollary~\ref{cor:lower_bound_components}, the cost of the component $K$ must be more than $\tfrac{2}{1+\gamma} \cdot (d^{S_1} + d^{S_2})$, implying
\begin{align*}
   (1+\beta\delta)\cdot \tmerge(S_2,s^*)\ >&\ \bigl(\tfrac{2}{1+\gamma} - 1 -\beta\delta\bigr) \cdot (d^{S_1} + d^{S_2}) \\
   >&\ \tfrac{1-\beta\delta}{1+\beta\delta}\bigl(\tfrac{2}{1+\gamma} - 1 -\beta\delta\bigr) \cdot \bigl(\tmerge(S_1,s^*) + \tmerge(S_2,s^*)\bigr),
\end{align*}
where we used Lemma~\ref{lem:merge_time_approx_deactivation_time} in the second inequality.
We conclude
\[
\tmerge(S_2,s^*) \ >\ \frac{\frac{2}{1+\gamma}-1-\beta\delta}{\highlight{\frac{(1+\beta\delta)^2}{1-\beta\delta}-(\frac{2}{1+\gamma}-1-\beta\delta)}} \cdot \tmerge(S_1,s^*)= M \cdot \tmerge(S_1,s^*). 
\]
\end{proof}

We are now ready to prove that we can indeed partition $\mathcal{C}$ into chains as described above.

\begin{lemma}\label{lem:chain_structure}
We can partition $\mathcal{C}$ into chains $\mathcal{C}^1, \dots, \mathcal{C}^k$ such that for any two sets $S_i \in \mathcal{C}^i$ and $S_j \in \mathcal{C}^j$ with $i < j$ we have that
\begin{itemize}\itemsep0pt
\item $S_i$ is disjoint from $S_j$, and
\item $\tmerge(S_j,s^*) > M \cdot \tmerge(S_i,s^*)$.
\end{itemize}
\end{lemma}
\begin{proof}
We partition $\mathcal{C}$ into set families $\mathcal{C}^1, \dots, \mathcal{C}^k$ by there merge time with $s^*$, meaning that 
\begin{itemize}\itemsep0pt
\item any two sets $S$ and $S'$ in the same set family $\mathcal{C}^i$ have the same merge time, i.e., we have $\tmerge(S,s^*) = \tmerge(S',s^*)$, and
\item any two sets $S$ and $S'$ in the different set families $\mathcal{C}^i$ have different merge times, i.e., $\tmerge(S,s^*) \neq \tmerge(S',s^*)$.
\end{itemize}
We choose the numbering of $\mathcal{C}^1, \dots, \mathcal{C}^k$ such that for any two sets $S_i \in \mathcal{C}^i$ and $S_j \in \mathcal{C}^j$ with $i < j$, we have $\tmerge(S_i, s^*) < \tmerge(S_j, s^*)$.

Now consider two sets $S,S' \in \mathcal{C}\subseteq \mathcal{S}$ with $S\subseteq S'$.
Then because $s^* \notin S'$, we have $\tmerge(S,s^*) = \tmerge(S', s^*)$.
Because $\mathcal{S}$ is a laminar family, this show that for any two sets $S_,S' \in \mathcal{C}$ with $\tmerge(S,s^*) \neq \tmerge(S', s^*)$, we have $S\cap S' = \emptyset$.
In particular, any two sets $S_i \in \mathcal{C}^i$ and $S_j \in \mathcal{C}^j$ with $i < j$ are disjoint.
By the choice of the numbering of $\mathcal{C}^1, \dots, \mathcal{C}^k$ and Lemma~\ref{lem:different_merge_times}, we also have $\tmerge(S_j,s^*) > M \cdot \tmerge(S_i,s^*)$.

It remains to prove that each set family $\mathcal{C}^i$ is a chain.
Suppose this is not the case.
Then  because $\mathcal{C}^i \subseteq \mathcal{S}$ is a laminar family, $\mathcal{C}^i$ contains two disjoint sets, which by  Lemma~\ref{lem:different_merge_times} have a different merge time with $s^*$.
This contradicts our choice of $\mathcal{C}^1, \dots, \mathcal{C}^k$.
\end{proof}

Finally, we show that every set from each one of the first $k-1$ chains is a (proper) subset of $S^*$. 

\begin{lemma}\label{lem:fist_chains_contained}
For every set $S\in \mathcal{C}^1 \cup \dots \cup \mathcal{C}^{k-1}$, we have $S\subsetneq S^*$.
\end{lemma}
\begin{proof}
Let $S \in  \mathcal{C}^1 \cup \dots \cup \mathcal{C}^{k-1}$ and $S_k \in \mathcal{C}^k$.
Then
\[
t^* \overset{Lem. \ref{lem:merge_time_approx_deactivation_time}}{>}\ \tfrac{1-\beta\delta}{1+\beta\delta}\cdot\tmerge(S_k,s^*) \ \overset{Lem. \ref{lem:chain_structure}}{\geq}\  \highlight{\tfrac{1-\beta\delta}{1+\beta\delta}\cdot M \cdot \tmerge(S,s^*) \ \geq\ \tmerge(S,s^*)},
\]
where we used the values of the constants $\beta$, $\delta$, and $M$ in the last inequality (see Table~\ref{table:constants}).
Because $s^* \in S^*$ and $d^{S^*} > t^*$ (by assumption of Claim~\ref{claim:main}), this implies $S \subsetneq S^*$ by Observation~\ref{obs:deactivation_time_lower_bounds_merge_time}.
\end{proof}

\subsection{Separating the contribution of the last chain}
\label{sec:separating}

In this section we show that we can subdivide the path $P^*$ (from the terminal $s^*$ to the vertex $w^*$) into two subpaths where on the first subpath the chain $\mathcal{C}^k$ did not contribute and on the second subpath the chains $\mathcal{C}^1, \dots , \mathcal{C}^{k-1}$ did not contribute.
Later, we will apply the induction hypothesis to the first subpath (Section~\ref{sec:completing_induction}). 
For the second subpath only the chain $\mathcal{C}^k$ (and of course sets containing the vertex $s^*$) contributed, which will help us to analyze the length of this subpath (see Section~\ref{sec:single_chain}).

To subdivide the path $P^*$ into the above-mentioned subpaths, we consider the last edge $(u_1,u_2)$ of $P^*$ with $t^{S^*u_1} \le \tfrac{\tmerge^k}{1+\alpha\delta}$.
We will prove that the chain $\mathcal{C}^k$ did not contribute to the subpath $P^*_{[s^*,u_2]}$ and the chains $\mathcal{C}^1 \cup \dots \cup \mathcal{C}^{k-1}$ did not contribute to the subpath $P^*_{[u_2,w^*]}$.
For an illustration and overview of our analysis, see Figure~\ref{fig:analysis_division_subpath}.

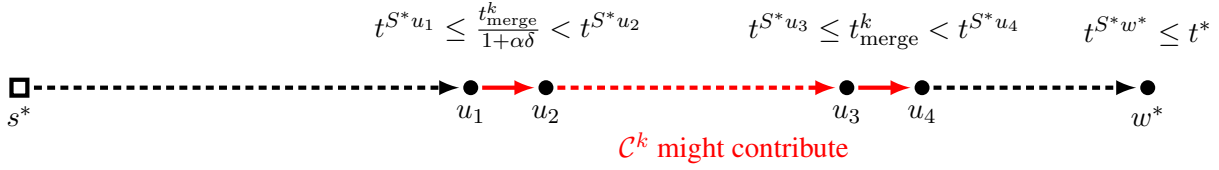
\begin{figure}
\begin{center}
\begin{tikzpicture}
\tikzset{terminal/.style={
ultra thick,draw,fill=none,rectangle,minimum size=0pt, inner sep=3pt, outer sep=2.5pt}
}
\tikzset{steiner/.style={
fill=black,circle,inner sep=0em,minimum size=0pt, inner sep=2pt, outer sep=1.5pt}
}

\node[terminal] (start) at (-1,0) {};
\node[steiner] (u1) at (5,0) {};
\node[steiner] (u2) at (6,0) {};
\node[steiner] (u3) at (10,0) {};
\node[steiner] (u4) at (11,0) {};
\node[steiner] (w) at (14,0) {};

\node[below=3pt] at (start) {$s^*$};
\node[below=3pt] at (u1) {$u_1$};
\node[below=3pt] at (u2) {$u_2$};
\node[below=3pt] at (u3) {$u_3$};
\node[below=3pt] at (u4) {$u_4$};
\node[below=3pt] at (w) {$w^*$};

\begin{scope}[->, >=latex, ultra thick]
\draw[densely dashed] (start) to (u1);
\draw[red] (u1) to node[above=10pt, black] {$t^{S^*u_1} \le \tfrac{\tmerge^k}{1+\alpha\delta} < t^{S^*u_2}$} (u2);
\draw[densely dashed, red] (u2) to (u3);
\draw[red] (u3) to  node[above=10pt, black] {$t^{S^*u_3} \le \tmerge^k < t^{S^*u_4}$}  (u4);
\draw[densely dashed] (u4) to (w);
\end{scope}
\node [above=12pt] at (w) {$t^{S^*w^*} \le t^*$};
\node[red] at (8.5,-0.8){$\mathcal{C}^k$ might contribute};
\end{tikzpicture}
\caption{ \label{fig:analysis_division_subpath}
Illustration of the path $P^*$ in the case where $\tmerge^k < t^{S^*w^*}$. 
(Without this assumption, there might be some degeneracy, where several of the depicted vertices will be defined to be identical.)
Single edges are shown as solid lines, whereas paths with potentially several edges are dashed.\newline
Recall that by the definition of an $S^*$-tight path, the times $t^{S^*v}$ for vertices $v$ of $P^*$ increase monotonically along the path $P^*$.
The edge $(u_1, u_2)$ is defined so that $t^{S^*u_1} \le \tfrac{\tmerge^k}{1+\alpha\delta} < t^{S^*u_2}$.
We will later also consider the edge $(u_3,u_4)$ chosen so that $t^{S^*u_3} \le \tmerge^k < t^{S^*u_4}$.
We will show that $\mathcal{C}^k$ only contributed to the subpath $P^*_{[u_1,u_4]}$ (red) .
We will also show that the chains $\mathcal{C}^1 \cup \dots \cup \mathcal{C}^{k-1}$ did not contribute to the subpath $P^*_{[u_2,w^*]}$ (see Lemma~\ref{lem:only_last_chain_contributes_to_new_subpath}).\newline
To bound the length of $P^*_{[s^*, u_1]}$ we will use the induction hypothesis. 
The length of  $P^*_{[u_2,u_4]}$ will be bounded by $3\cdot (1+\delta) \cdot (\tmerge^k - \tfrac{\tmerge^k}{1+\alpha\delta}) + 2\epsilon$ (see Lemma~\ref{lem:bound_subpath_potential}).
The length of $P^*_{[u_4, w^*]}$, where no set from $\mathcal{C}$ contributes, will be bounded by $(1+\delta) \cdot (t^* - \tmerge^k)$ (see Lemma~\ref{lem:bound_length_of_trivial_subpath}).
}
\end{center}
\end{figure}

\bigskip

First, we prove that indeed no set from the first $k-1$ chains contributed to the path $P^*_{[u_2, w^*]}$.
For every two sets $S,S'$ in the same chain $\mathcal{C}^j$, we have $\tmerge(S,s^*) = \tmerge(S',s^*)$ because $s^*$ is contained neither in $S$ nor in $S'$.
We denote this time by  $\tmerge^j$ and we then have $\tmerge^j = \tmerge(S,s^*) $ for every $S\in \mathcal{C}^j$. 

\begin{lemma}\label{lem:only_last_chain_contributes_to_new_subpath}
No set $S\in \mathcal{C}^1 \cup \dots \cup \mathcal{C}^{k-1}$ contributed to any edge of $P^*_{[u_2, w^*]}$.
\end{lemma}
\begin{proof}
The statement is trivial if $P^*_{[u_2, w^*]}$ has no edges.
Otherwise, consider an edge $(v,w)$ of $P^*_{[u_2, w^*]}$ and a set $S\in \mathcal{C}^j$ with $j < k$.
By the definition of the edge $(u_1,u_2)$, we have $t^{S^*v} > \tfrac{\tmerge^k}{1+\alpha\delta}$, which by Lemma~\ref{lem:chain_structure} implies
\[
t^{S^*v}\ >\ \highlight{\tfrac{M}{1+\alpha\delta} \cdot \tmerge^j \ \ge\ \tmerge^j}\ =\  \tmerge(S,s^*) \ \ge\  d^S,
\]
where we used Observation~\ref{obs:lower_bound_merge_time_sets} in the last inequality.
By Lemma~\ref{lem:fist_chains_contained}, we have $S\subsetneq S^*$ and hence by Observation~\ref{obs:monotonicity_reaching_times} we get $t^{Sv} \ge t^{S^*v} > d^S$, implying that $S$ did not contributed to the edge $(v,w)$.
\end{proof}

We next show that $\mathcal{C}^k$ does not contribute to the edges of $P^*_{[s^*,u_1]}$. 
The key technical  statement we will use to prove this is the following.

\begin{lemma}\label{lem:lower_bound_meeting_time}
Let $\tilde{S}\in \mathcal{C}^k$ and $v\neq w^*$ a vertex of the path $P^*$.
If $t^{\tilde{S}v} < t^*$, then $\max\{ t^{\tilde{S}v}, t^{S^*v} \} > \tfrac{\tmerge^k}{1+\alpha\delta}$.
\end{lemma}
\begin{proof}
Suppose for the sake of deriving a contradiction that we have $ t^{\tilde{S}v} \le \tfrac{\tmerge^k}{1+\alpha\delta}$ and $t^{S^*v} \le \tfrac{\tmerge^k}{1+\alpha\delta}$.
Because $v \neq w^*$, the path $P^*_{[s^*,v]}$ is an $S^*$-tight path with strictly fewer edges than $P^*$ at time $t^{S^*v} \le t^*$ and hence we can apply \eqref{item:strong_existing_component} of the induction hypothesis to $P^*_{[s^*,v]}$.
We obtain a terminal set $X^*\subseteq R$ with $s^*\in X^*$  and a component $K^*$ connecting $X^*\cup \{v\}$ such that
\begin{equation}\label{eq:properties_K_star}
\begin{aligned}
c(P^*_{[s^*,v]}) \ \le&\ (1+\delta) \cdot t^{S^*v} + \lambda \cdot \drop(X^*, s^*) \\
c(K^*) - \tfrac{1}{1+\gamma} \cdot \drop(X^*, s^*) \ \le&\  (1+\delta) \cdot t^{S^*v} - \mu \cdot \drop(X^*,s^*) \\
\tmerge(s^*,s) \ <&\  (1+\alpha\delta) \cdot t^{S^*v}& \text{ for all }s \in X^*
\end{aligned}
\end{equation}
By Lemma~\ref{lem:reachability_implies_S-tight_path}, there is an $\tilde{S}$-tight path $\tilde{P}$ from a terminal $\tilde{s}\in \tilde{S}$ to $v$ at time $t^{\tilde{S}v}$.
Because $t^{\tilde{S}v} < t^*$, we can apply \eqref{item:strong_existing_component} of the induction hypothesis to the path $\tilde{P}$.
We obtain a terminal set $\tilde{X}\subseteq R$ with $\tilde{s}\in \tilde{X}$  and a component $\tilde{K}$ connecting $\tilde{X}\cup \{v\}$ such that
\begin{equation}\label{eq:properties_K_tilde}
\begin{aligned}
c(\tilde{P}) \ \le&\ (1+\delta) \cdot t^{\tilde{S}v} + \lambda \cdot \drop(\tilde{X}, \tilde{s}) \\
c(\tilde{K}) - \tfrac{1}{1+\gamma} \cdot \drop(\tilde{X}, \tilde{s}) \ \le&\  (1+\delta) \cdot t^{\tilde{S}v} - \mu \cdot \drop(\tilde{X},\tilde{s})  \\
\tmerge(\tilde{s},s) \ <&\  (1+\alpha\delta) \cdot t^{\tilde{S}v}& \text{ for all }s \in \tilde{X}
\end{aligned}
\end{equation}

The union $K^*\cup \tilde{K}$ of the components $K^*$ and $\tilde{K}$ is a component connecting $X^* \cup \tilde{X}$.
We will show $ c(K^* \cup \tilde{K})-  \tfrac{1}{1+\gamma}  \cdot \drop(X^* \cup \tilde{X},s^*) \le 0$, contradicting Corollary~\ref{cor:lower_bound_components}.

In order to derive a lower bound on $\drop(X^* \cup \tilde{X},s^*)$, we first prove the following claim.

\begin{claim}\label{claim:drops_combined_meeting_time}
$\drop(X^* \cup \tilde{X},s^*) \ge \drop(\tilde{X}, \tilde{s}) + \drop(X^*, s^*) + 2\cdot d^{\tilde{S}}$.
\end{claim}
\begin{claimproof}
Let $\mathcal{S}^* \subseteq \mathcal{S}$ be a drop certificate for $X^*$ such that $s^* \notin S'$ for all $S'\in \mathcal{S}^*$ and $\drop(X^*,s^*) = 2 \sum_{S\in \mathcal{S}^*} d^{S}$ (which exists by the definition of $\drop(X^*,s^*)$).
Similarly, we let $\tilde{\mathcal{S}} \subseteq \mathcal{S}$ be a drop certificate for $\tilde{X}$ such that $\tilde{s} \notin S$ for all $S\in \tilde{\mathcal{S}}$ and $\drop(\tilde{X}, \tilde{s}) = 2 \sum_{S\in \tilde{\mathcal{S}}} d^{S}$.
We will prove that $\mathcal{S}^* \cup \tilde{\mathcal{S}} \cup \{\tilde{S}\}$ is a drop certificate for $X^* \cup \tilde{X}$.
\bigskip

We claim that $X\cap \tilde{X} = \emptyset$ and that for any two distinct terminals $s_1,s_2\in  X^* \cup \tilde{X}$, there exists a set $S\in \mathcal{S}^* \cup \tilde{\mathcal{S}} \cup \{\tilde{S}\}$ with $|S\cap \{s_1,s_2\}| = 1$.
To see this, consider two terminals $s_1,s_2 \in X^* \cup \tilde{X}$.
If we have $s_1,s_2 \in X^*$, then either $s_1=s_2$ or there exists a set $S\in \mathcal{S}^*$ with $|S\cap \{s_1,s_2\}| = 1$.
Similarly, if  $s,s' \in \tilde{X}$, then either $s_1=s_2$ or there exists a set $S\in \tilde{\mathcal{S}}$ with $|S\cap \{s_1,s_2\}| = 1$.

Now consider the case where one of the terminals $s_1,s_2$ is contained in $X^*$ and the other in $\tilde{X}$, say $s_1 \in X^*$ and $s_2\in \tilde{X}$.
Using $\tmerge(s^*,s) \ <\  (1+\alpha\delta) \cdot t^{S^*v}$ for all $s \in X^*$, we get
\begin{equation}\label{eq:first_merge_time_lemma_meeting}
\tmerge(s^*, s_1) \ <\ (1+\alpha\delta) \cdot t^{S^*v} \le (1+\alpha\delta) \cdot \tfrac{\tmerge^k}{1+\alpha\delta}\ =\ \tmerge^k,
\end{equation}
and similarly, using $\tmerge(\tilde{s},s) \ <\  (1+\alpha\delta) \cdot t^{\tilde{S}v}$ for all $s \in \tilde{X}$, we get
\begin{equation}\label{eq:second_merge_time_lemma_meeting}
\tmerge(\tilde{s}, s_2) \ <\ (1+\alpha\delta) \cdot t^{\tilde{S}v} \le (1+\alpha\delta) \cdot \tfrac{\tmerge^k}{1+\alpha\delta}\ =\ \tmerge^k.
\end{equation}
Using $\tmerge(s^*,\tilde{s}) = \tmerge^k$, this implies $\tmerge(s_1,s_2) \ge \tmerge^k$  as otherwise by Observation~\ref{obs:dist_lb_merge_time}
\[
\tmerge^k \ =\ \tmerge(s^*,\tilde{s})\ \le\ \max \Bigl\{  \tmerge(s^*, s_1),  \tmerge(s_1,s_2), \tmerge( s_2,\tilde{s})  \Bigr\} \ <\ \tmerge^k.
\]
In particular, we get $s_1 \neq s_2$, implying $X^*\cap \tilde{X} = \emptyset$.
If $s_1=s^*$ and $s_2=\tilde{s}$, then $s_2 \in \tilde{S}$ and $s_1 \notin \tilde{S}$.
Now consider the case where $s_1\neq s^*$.
Then the drop certificate $\mathcal{S}^*$ contains a set $S$ with $s_1\in S$ and $s^* \notin S$.
Using Observation~\ref{obs:deactivation_time_lower_bounds_merge_time} and $\tmerge(s_1,s_2) \ge \tmerge^k$, we get
\[
d^S\ \le\ \tmerge(s^*, s_1) \ \overset{\eqref{eq:first_merge_time_lemma_meeting}}{<}\ \tmerge^k\ \le\  \tmerge(s_1,s_2),
\]
 implying $s_2 \notin S$.
Similarly, if $s_2 \neq \tilde{s}$, the drop certificate $\tilde{\mathcal{S}}$ contains a set $S$ with $s_2\in S$ and $\tilde{s}\notin S$.
Then 
\[
d^S\ \le\ \tmerge(\tilde{s}, s_2)\ \overset{\eqref{eq:second_merge_time_lemma_meeting}}{<}\ \tmerge^k \ \le\  \tmerge(s_1,s_2),
\]
implying $s_1 \notin S$.

We have shown that for any two distinct terminals $s_1,s_2\in X^* \cup \tilde{X}$, there is a set $S \in \mathcal{S}^* \cup \tilde{\mathcal{S}} \cup \{\tilde{S}\}$ with $|S\cap \{s_1,s_2\}| = 1$, which in particular implies $|\mathcal{S}^* \cup \tilde{\mathcal{S}} \cup \{\tilde{S}\}| \ge |X^* \cup X| -1$.
Moreover, we have shown $X^*\cap \tilde{X} = \emptyset$ and hence
\[
|X^* \cup X| -1\ \le\ |\mathcal{S}^* \cup \tilde{\mathcal{S}} \cup \{\tilde{S}\}|  
\ \le\  |\mathcal{S}^*| + |\tilde{\mathcal{S}}| + 1\ =\ |X^*| - 1 + |\tilde{X}| - 1 + 1\ =\ |X^* \cup X| -1.
\]
This implies that we have equality throughout and hence $\mathcal{S}^*$, $\tilde{\mathcal{S}}$, and $\{ \tilde{S}\}$ are disjoint.
We conclude that $\mathcal{S}^* \cup \tilde{\mathcal{S}} \cup \{\tilde{S}\}$ is indeed a drop certificate for $X^* \cup \tilde{X}$ and has value $\drop(\tilde{X}, \tilde{s}) + \drop(X^*, s^*) + 2\cdot d^{\tilde{S}}$.
\end{claimproof}

Having shown Claim~\ref{claim:drops_combined_meeting_time}, we next
prove a lower bound on $\drop(\tilde{X}, \tilde{s}) + \drop(X^*, s^*)$.
Recall that $t^{\tilde{S}v}\leq \tfrac{\tmerge^k}{1+\alpha\delta}$ and $t^{S^*v}\leq \tfrac{\tmerge^k}{1+\alpha\delta}$ by the assumption we made for the sake of deriving a contradiction.
Using $\tilde{S}\in \mathcal{C}^k$ and Observation~\ref{obs:dist_lb_merge_time}, we get $2 \cdot \tmerge^k = 2  \cdot \tmerge(\tilde{s},s^*) \le \dist(\tilde{s},s^*)$.
Hence, 
\begin{equation}\label{eq:lower_bound_sum_of_drops}
\begin{aligned}
2 \cdot \tmerge^k \ \le&\ c(\tilde{P}) + c(P^*_{[s^*,v]})  \\
\le&\ (1+\delta) \cdot t^{\tilde{S}v} + \lambda \cdot \drop(\tilde{X}, \tilde{s}) +(1+\delta) \cdot t^{S^*v} + \lambda \cdot \drop(X^*, s^*) \\
\le&\ 2 \cdot (1+\delta) \cdot \tfrac{\tmerge^k}{1+\alpha\delta} + \lambda \cdot \bigl(\drop(\tilde{X}, \tilde{s}) + \drop(X^*, s^*) \bigr). 
\end{aligned}
\end{equation}
Using the properties~\eqref{eq:properties_K_star} and~\eqref{eq:properties_K_tilde} of the components $K^*$ and $\tilde{K}$ guaranteed by the induction hypothesis and Claim~\ref{claim:drops_combined_meeting_time}, we get
\begin{align*}
&c(K^* \cup \tilde{K})   - \tfrac{1}{1+\gamma} \cdot \drop(\tilde{X} \cup X^*,s^*)\\
&\le\ c(K^*) - \tfrac{1}{1+\gamma} \cdot \drop(X^*,s^*) + c(\tilde{K}) - \tfrac{1}{1+\gamma} \cdot \drop(\tilde{X},s^*) - \tfrac{2}{1+\gamma} \cdot d^{\tilde{S}} \\
&\le\ (1+\delta) \cdot t^{S^*v} - \mu \cdot \drop(X^*,s^*)
+ (1+\delta) \cdot t^{\tilde{S}v} - \mu \cdot \drop(\tilde{X},\tilde{s}) - \tfrac{2}{1+\gamma} \cdot d^{\tilde{S}}\\
&\le\ 2 \cdot (1+\delta) \cdot \tfrac{\tmerge^k}{1+\alpha\delta} - \mu \cdot \bigl(\drop(\tilde{X}, \tilde{s}) + \drop(X^*, s^*) \bigr) - \tfrac{2}{1+\gamma} \cdot d^{\tilde{S}}.
\end{align*}
Using the lower bound~\eqref{eq:lower_bound_sum_of_drops} on $\drop(\tilde{X}, \tilde{s}) + \drop(X^*, s^*)$ and Lemma~\ref{lem:merge_time_approx_deactivation_time}, we conclude
\begin{align*}
&c(K^* \cup \tilde{K})   - \tfrac{1}{1+\gamma} \cdot \drop(\tilde{X} \cup X^*,s^*)\\
&\le\ 2 \cdot (1+\delta) \cdot \tfrac{\tmerge^k}{1+\alpha\delta} - \tfrac{2\mu}{\lambda} \cdot \bigl(\tmerge^k - (1+\delta)\cdot  \tfrac{\tmerge^k}{1+\alpha\delta} \bigr) - \tfrac{2}{1+\gamma} \cdot d^{\tilde{S}}\\ 
&=\ \tfrac{2}{1+\alpha\delta} \cdot \bigl( 1+ \delta - \tfrac{\mu}{\lambda} \cdot (\alpha-1)\cdot \delta \bigr) \cdot \tmerge^k - \tfrac{2}{1+\gamma} \cdot d^{\tilde{S}} \\
&\le\ \highlight{\Bigl( \tfrac{2}{1+\alpha\delta} \cdot \bigl( 1+ \delta - \tfrac{\mu}{\lambda} \cdot (\alpha-1)\cdot \delta\bigr) - \tfrac{2}{1+\gamma} \cdot \tfrac{1-\beta\delta}{1+\beta\delta} \Bigr) \cdot \tmerge^k} \\
& \highlight{\le\ 0},
\end{align*}
contradicting Corollary~\ref{cor:lower_bound_components}.
(Recall Table~\ref{table:constants} for the values of the constants used in the last inequality.)
\end{proof}

\begin{corollary}\label{cor:merge_time_and_function_f}
We have $\tmerge^k < (1+\alpha\delta)\cdot t^*$.
\end{corollary}
\begin{proof}
Consider the first edge $(u,v)$ of the path $P^*$ where a set $S\in \mathcal{C}^k$ contributed.
Then $t^{Su} < t^*$ and hence Lemma~\ref{lem:lower_bound_meeting_time} implies $\max\{ t^{Su}, t^{S^*u} \} > \tfrac{\tmerge^k}{1+\alpha\delta}$.
Because $u$ is a vertex of the path $P^*$, which is $S^*$-tight at time $t^*$, we have $t^{S^*u} \le t^*$, implying $\tfrac{\tmerge^k}{1+\alpha\delta} <\max\{ t^{Su}, t^{S^*u} \}  \le t^*$.
\end{proof}

From Lemma~\ref{lem:lower_bound_meeting_time} and Corollary~\ref{cor:merge_time_and_function_f} we can conclude that no set $S\in \mathcal{C}^k$  contributed to any edge of the path $P^*_{[s^*,u_2]}$.

\begin{corollary}\label{cor:last_chain_contributes_only_to_new_subpath}
No set $S\in \mathcal{C}^k$  contributed to any edge of $P^*_{[s^*,u_1]}$.
\end{corollary}
\begin{proof}
By the choice of the edge $(u_1,u_2)$, we have $t^{S^*u_1} \le \frac{\tmerge^k}{1+\alpha\delta}$ and hence because $P^*$ is $S^*$-tight, for every vertex $u$ of the path $P^*_{[s^*,u_1]}$, we have $t^{S^*u} \le t^{S^*u_1} \le \frac{\tmerge^k}{1+\alpha\delta}$.
Consider an edge $(v,w)$ of the path $P^*_{[s^*,u_1]}$.
Because $P^*$ is an $S^*$-tight path, this edge was $\delta$-tight already at time $t^{S^*w}$.

Suppose for the sake of deriving a contradiction that $S\in \mathcal{C}^k$  contributed to the edge $(v,w)$.
Then $t^{Sv} < t^*$ and hence by Lemma~\ref{lem:lower_bound_meeting_time}, we have $\max\{t^{Sv}, t^{S^*v}\} > \frac{\tmerge^k}{1+\alpha\delta}$, implying $t^{Sv} > \frac{\tmerge^k}{1+\alpha\delta}$ (because $t^{S^*v} \le \frac{\tmerge^k}{1+\alpha\delta}$).
However, the edge $(v,w)$ was $\delta$-tight already at time $t^{S^*w} \le \frac{\tmerge^k}{1+\alpha\delta} < t^{Sv}$, contradicting the fact that $S$ contributed to $(v,w)$.
\end{proof}

\subsection{Analyzing the contribution of a single chain}
\label{sec:single_chain}

The goal of this section is to prove an upper bound on the length of the path $P^*_{[u_2,w^*]}$.
To this end we will first define a potential function $\pi : V^{\epsilon} \to \mathbb{R}_{\ge 0}$.
Then we will prove the following key property.
For an edge $e=(v,w)$ we denote by $z_e^k$ the contribution of the chain $\mathcal{C}^k$ to $e$ and by $\widetilde{z}_e$ the total contribution of other sets (until the current time $t^*$). 
We will prove
\begin{equation}\label{eq:key_property_of_potential}
 \pi(w) -\pi(v) \le \widetilde{z}_e - z_e^k
\end{equation}
for every $\delta$-tight edge $e=(v,w)$.

This will allow us to bound the contribution of the chain $\mathcal{C}^k$ to subpaths of $P^*$, in terms of the potential difference of its endpoints.
More precisely, because every edge of $P^*$ is $\delta$-tight, we have $\widetilde{z}_e = \tfrac{c(e)}{1+\delta} - z_e^k$ for each edge $e$ of $P^*$. 
Hence for every subpath $P^*_{[u,v]}$, we have
\[ 
 \pi(v) -\pi(u)\ \le\ \sum_{e\in P^*_{[u,v]}} \bigl(\widetilde{z}_e - z_e^k \bigr)\ =\ \tfrac{1}{1+\delta} \cdot c\bigl(P^*_{[u,v]}\bigr) - 2 \sum_{e\in P^*_{[u,v]}} z_e^k,
\]
which implies 
\[
\sum_{e\in P^*_{[u,v]}}z_e^k \ \le\  \tfrac{1}{2} \cdot \Bigl( \tfrac{1}{1+\delta} \cdot c(P^*_{[u,v]}) + \pi(u) - \pi(v)  \Bigr).
\]
\bigskip

We now define the potential function $\pi$.
For each set $S \in \mathcal{C}^k$ and every vertex $v$, we define
\[
\pi (S, v) \coloneqq \max\Bigl\{ 0, \ \min \Bigl\{ t^*,\  d^{S}\Bigr\} - \max \Bigl\{ t^{Sv},\  a^{S} \Bigr\}\Bigr\}.
\]
Intuitively, this is (an upper bound on) the length of the time interval in which $S$ contributed to a subpath of~$P^*$ starting at the vertex $v$: the set $S$ could not contribute before its activation time $a^{S}$ or before time $t^{Sv}$ and it could not contribute after it was deactivated.
We now define 
\[
\pi(v)\ \coloneqq\ \sum_{S\in \mathcal{C}^k} \pi(S, v).
\]
See Figure~\ref{fig:potential} for an example.

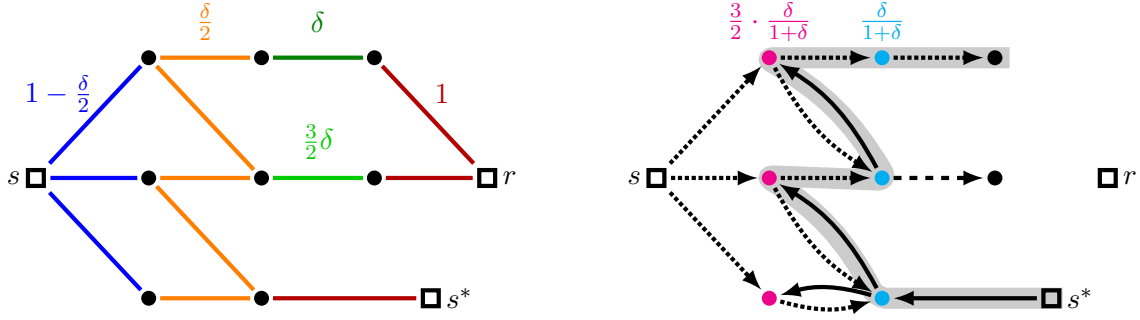
\begin{figure}
\begin{center}
\begin{tikzpicture}[xscale=1.5, yscale=1.6]
\tikzset{terminal/.style={
ultra thick,draw,fill=none,rectangle,minimum size=0pt, inner sep=3pt, outer sep=2.5pt}
}
\tikzset{steiner/.style={
fill,circle,inner sep=0em,minimum size=0pt, inner sep=2pt, outer sep=1.5pt}
}
\tikzset{dual/.style={line width=2pt}}

\begin{scope}[shift={(0,0)}]

\begin{scope}[every node/.style=terminal]
\node (s) at (0,1) {};
\node (r) at (4,1) {};
\node (star) at (3.5,0) {};
\end{scope}
\node[left=2pt] at (s) {$s$};
\node [right=2pt] at (r) {$r$};
\node [right=2pt] at (star) {$s^*$};

\begin{scope}[every node/.style=steiner]
\node (a0) at (1,0) {};
\node (a1) at (1,1) {};
\node (a2) at (1,2) {};
\node (b0) at (2,0) {};
\node (b1) at (2,1) {};
\node (b2) at (2,2) {};
\node (c1) at (3,1) {};
\node (c2) at (3,2) {};
\end{scope}

\begin{scope}[ultra thick]
\begin{scope}[blue]
\draw (s) -- (a0);
\draw (s) -- (a1);
\draw (s) -- (a2);
\end{scope}
\begin{scope}[orange]
\draw (a0) -- (b0);
\draw (a1) -- (b1);
\draw (a2) -- (b2);
\draw  (a2) -- (b1);
\draw (a1) -- (b0);
\end{scope}
\draw[green!50!black] (b2) -- (c2);
\draw[green!80!black] (b1) -- (c1);
\begin{scope}[red!70!black]
\draw (b0) -- (star);
\draw (c1)-- (r);
\draw (c2) -- (r);
\end{scope}
\end{scope}

\node[blue] at (0.2,1.7) {$1-\tfrac{\delta}{2}$};
\node[orange] at (1.5,2.3) {$\tfrac{\delta}{2}$};
\node[green!50!black] at (2.5,2.3) {$\delta$};
\node[green!80!black] at (2.5,1.3) {$\tfrac{3}{2}\delta$};
\node[red!70!black] at (3.6,1.7) {$1$};

\end{scope}

\begin{scope}[shift={(5.5,0)}]

\draw[line width=8pt, opacity=0.2, rounded corners]
(3.6,0) -- (1.95,0) to[bend right=15] (0.95,1.02) -- (2.05,0.98) to[bend right=15] (0.95,2) -- (3.13,2);

\begin{scope}[every node/.style=terminal]
\node (s) at (0,1) {};
\node (r) at (4,1) {};
\node (star) at (3.5,0) {};
\end{scope}
\node[left=2pt] at (s) {$s$};
\node [right=2pt] at (r) {$r$};
\node [right=2pt] at (star) {$s^*$};

\begin{scope}[every node/.style=steiner]
\begin{scope}[magenta]
\node (a0) at (1,0) {};
\node (a1) at (1,1) {};
\node (a2) at (1,2) {};
\end{scope}
\begin{scope}[cyan]
\node (b0) at (2,0) {};
\node (b1) at (2,1) {};
\node (b2) at (2,2) {};
\end{scope}
\node (c1) at (3,1) {};
\node (c2) at (3,2) {};
\end{scope}
\node[magenta] at (1,2.3) {$\frac{3}{2}\cdot \frac{\delta}{1+\delta}$};
\node[cyan] at (2,2.3) {$\frac{\delta}{1+\delta}$};

\begin{scope}[ultra thick, ->, >=latex]
\begin{scope}[densely dotted]
\draw (s) to (a0);
\draw (s) to (a1);
\draw (s) to (a2);
\draw[bend right=15] (a0) to (b0);
\draw (a1) to[bend right=15] (b0);
\draw (a1) to (b1);
\draw (a2) to[bend right=15] (b1);
\draw (a2) to (b2);
\draw (b2) to (c2);
\end{scope}
\begin{scope}[dashed]
\draw (b1) to (c1);
\end{scope}
\begin{scope}
\draw (star) to (b0);
\draw[bend right=15] (b0) to (a0);
\draw (b0) to[bend right=15] (a1);
\draw (b1) to[bend right=15] (a2);
\end{scope}
\end{scope}

\end{scope}

\end{tikzpicture}
\caption{\label{fig:potential}
The left part of the figure shows an instance with three terminals $s$, $s^*$, and $r$. 
The costs of the edges are shown in the corresponding color.
For $t\in [0,1)$, we have  $\mathcal{S}^t = \{ \{s\}, \{s^*\}, \{r\}\}$ and for $t\in [1, t_{\max}) =[1,1+\frac{\delta}{2})$, we have $\mathcal{S}^t = \{ \{s, s^*\}, \{r\}\}$.\newline
The right part of the figure shows the $\delta$-tight edges at time $t^*=1$ and an $S^*$-tight path of length $1+3\delta$ highlighted in gray.
On solid edges only the set $S^*=\{s^*\}$ contributed, on dotted edges only the set $S=\{s\}$ contributed, and on the dashed edge both $S$ and $S^*$ contributed.
The color of a vertex $v$ indicates its potential $\pi(v)=\pi(S,v)$, where all black Steiner nodes and all terminals have potential zero.
}
\end{center}
\end{figure}

\begin{lemma}\label{lem:key_property_potential}
The potential function $\pi$ fulfills \eqref{eq:key_property_of_potential} for every $\delta$-tight edge $e=(v,w)$.
\end{lemma}
\begin{proof}
We consider an edge $\{v,w\}\in E^{\epsilon}$ and prove that both orientations of it fulfill \eqref{eq:key_property_of_potential} if they are $\delta$-tight.
More precisely, we will show that
\begin{equation}\label{eq:potential_dir_vw}
 \pi(w) -\pi(v)\ \le\ \widetilde{z}_{(v,w)} - z_{(v,w)}^k
\end{equation} 
if $(v,w)$ is $\delta$-tight and
\begin{equation}\label{eq:potential_dir_wv}
 \pi(v) -\pi(w)\ \le\ \widetilde{z}_{(w,v)} - z_{(w,v)}^k
\end{equation} 
if $(w,v)$ is $\delta$-tight.
\bigskip

We call a set $S\in \mathcal{C}^k$ \emph{irrelevant} (for $\{v,w\}$) if it satisfies at least one of the following conditions:
\begin{itemize}\itemsep0pt
\item $t^* \le a^{S}$
\item $t^* \le \min\bigl\{ t^{Sv}, t^{Sw}\bigr\}$
\item $a^{S} \ge \max \bigl\{ t^{Sv}, t^{Sw}\bigr\}$
\item $d^{S} \le \min \bigl\{ t^{Sv}, t^{Sw}\bigr\}$.
\end{itemize}
We observe that if $S$ is irrelevant, then its contribution (until the current time $t^*$) to the edges $(v,w)$  and $(w,v)$ is $0$.
Moreover, we then also have $\pi(S,v) = \pi(S,w)$.
\bigskip

We call  $S\in \mathcal{C}^k$ \emph{relevant} if it is not irrelevant.
Let us assume that the chain $\mathcal{C}^k$ contains at least one relevant set (otherwise the claim trivially holds), and let $S_{\min}$ be the minimal such set. We assume $t^{S_{\min}v} \le t^{S_{\min}w}$ without loss of generality (by symmetry).
We denote by $z^{S}_e$ the contribution of a set $S$ to an edge $e$ (until the current time $t^*$) and claim that
\begin{equation}\label{eq:relate_contribution_to_potential_diff}
z^{S}_{(v,w)}\ =\ \pi(S,v) -  \pi(S,w) 
\end{equation}
and
\begin{equation}\label{eq:no_contribution_to_wv}
z^{S}_{(w,v)}\ =\ 0
\end{equation}
for every set $S\in \mathcal{C}^k$.
\bigskip

Before proving \eqref{eq:relate_contribution_to_potential_diff} and \eqref{eq:no_contribution_to_wv}, we observe that they imply \eqref{eq:potential_dir_vw} and \eqref{eq:potential_dir_wv}.
Summing \eqref{eq:relate_contribution_to_potential_diff} for all sets $S\in \mathcal{C}^k$, we obtain $z^k_{(v,w)} = \pi(v) - \pi(w)$.
Because $\widetilde{z}_{(v,w)} \ge 0$, this implies \eqref{eq:potential_dir_vw}.
Summing \eqref{eq:no_contribution_to_wv} for all sets $S\in \mathcal{C}^k$, we obtain $z^k_{(w,v)}=0$.
Hence, if $(w,v)$ is $\delta$-tight, we have 
\[
\pi(v) - \pi(w) =  z^k_{(v,w)} \le \frac{c(\{v,w\})}{1+\delta} = \frac{c(\{v,w\})}{1+\delta} - z^k_{(w,v)} - z^k_{(w,v)} = \widetilde{z}_{(w,v)} - z^k_{(w,v)}.
\]
\bigskip

It remains to prove \eqref{eq:relate_contribution_to_potential_diff} and \eqref{eq:no_contribution_to_wv} for every set $S\in \mathcal{C}^k$.
Observe that every irrelevant set $S$  satisfies \eqref{eq:relate_contribution_to_potential_diff} and \eqref{eq:no_contribution_to_wv} because both sides of \eqref{eq:relate_contribution_to_potential_diff} and \eqref{eq:no_contribution_to_wv} are zero.

Now consider a relevant set $S$.
First, we show $t^{Sv} \le t^{Sw}$.
If $S = S_{\min}$, this holds by the above assumption (which we made without loss of generality).
Otherwise $S_{\min} \subsetneq S$.
Then we have $a^{S} \ge d^{S_{\min}}$ and
$t^{Sv} \le t^{S_{\min}v}$ (by Observation~\ref{obs:monotonicity_reaching_times}).
Because $S$ and $S_{\min}$ are relevant, this implies
\[
t^{Sv} 
\ \le\  t^{S_{\min}v} 
\ =\ \min \bigl\{ t^{S_{\min}v} ,  t^{S_{\min}w}   \bigr\} 
\ <\  d^{S_{\min}} 
\ \le\ a^{S}
\ <\ \max\bigl\{   t^{Sv} , t^{Sw} \bigr\},
\]
implying $t^{Sv} < t^{Sw}$.

Having shown $t^{Sv} \le t^{Sw}$, we conclude that $z^{S}_{(w,v)} =0$ and hence \eqref{eq:no_contribution_to_wv}.
It remains to prove \eqref{eq:relate_contribution_to_potential_diff}.
We observe that if the set $S$ contributed to the edge $(v,w)$, then it started to contribute at time $\max \bigl\{ t^{Sv}, a^{S} \bigr\}$ and stopped contributing at time $\min \bigl\{ t^*, d^{S}, t^{Sw} \bigr\}$.
More precisely, we have
\begin{equation}\label{eq:contribution_direction_one}
z^{S}_{(v,w)} = \max \Bigl\{ 0,\ \min \bigl\{ t^*, d^{S}, t^{Sw} \bigr\} - \max \bigl\{ t^{Sv}, a^{S} \bigr\} \Bigr\}.
\end{equation}
We now distinguish two cases.\\

\textbf{Case 1:} $t^{Sw} > \min\{t^*, d^{S}\}$.\\[2mm]
Then $\pi(S,w)=\max\bigl\{ 0, \ \min \bigl\{ t^*,\  d^{S}\bigr\} - \max \bigl\{ t^{Sw},\  a^{S} \bigr\}\bigr\}=0$ and hence 
\begin{align*}
\pi(S,v) - \pi(S,w) =&\ \pi(S,v) \\
=&\  \max\Bigl\{ 0, \ \min \bigl\{ t^*,\  d^{S}\bigr\} - \max \bigl\{ t^{Sv},\  a^{S} \bigr\}\Bigr\} \\
=&\  \max \Bigl\{ 0,\ \min \bigl\{ t^*, d^{S}, t^{Sw} \bigr\} - \max \bigl\{ t^{Sv}, a^{S} \bigr\} \Bigr\}\\
\overset{\eqref{eq:contribution_direction_one}}{=}&\  z^{S}_{(v,w)} .
\end{align*}

\textbf{Case 2:} $t^{Sw} \le \min\{t, d^{S}\}$.\\[2mm]
Then we have $t^{Sv} \le t^{Sw} \le \min\{t, d^{S}\}$.
Moreover, because $S$ is relevant, we have $a^{S} < t^*$ and hence $a^{S} \le \min\{t^*, d^{S}\}$. 
We conclude
\begin{equation}\label{eq:pot_difference}
\begin{aligned}
\pi(S,v) - \pi(S,w) =&\  \max\Bigl\{ 0, \ \min \bigl\{ t^*,\  d^{S}\bigr\} - \max \bigl\{ t^{Sv},\  a^{S} \bigr\}\Bigr\} \\
&\ - \max\Bigl\{ 0, \ \min \bigl\{ t^*,\  d^{S}\bigr\} - \max \bigl\{ t^{Sw},\  a^{S} \bigr\}\Bigr\} \\
=&\ \max \bigl\{ t^{Sw},\  a^{S} \bigr\} -  \max \bigl\{ t^{Sv},\  a^{S} \bigr\}
\end{aligned}
\end{equation}
and by \eqref{eq:contribution_direction_one} we have
\begin{equation}\label{eq:contribution_of_S}
\begin{aligned}
z^{S}_{(v,w)} =&\ \max \Bigl\{ 0,\ \min \bigl\{ t^*, d^{S}, t^{Sw} \bigr\} - \max \bigl\{ t^{Sv}, a^{S} \bigr\} \Bigr\} \\
=&\ \max \Bigl\{ 0,\ t^{Sw}  - \max \bigl\{ t^{Sv}, a^{S} \bigr\} \Bigr\}.
\end{aligned}
\end{equation}
Because $S$ is relevant, we have
$a^{S} < \max \bigl\{ t^{Sv},t^{Sw} \bigr\} = t^{Sw}$ and hence
\[
\pi(S,v) - \pi(S,w) \overset{\eqref{eq:pot_difference}}{=} t^{Sw}  - \max \bigl\{ t^{Sv}, a^{S}\bigr\} \overset{\eqref{eq:contribution_of_S}}{=} z^{S}_{(v,w)},
\]
where we also used $t^{Sv} \le t^{Sw}$ for the last equality.
\end{proof}

Recall that we will use the potential $\pi$ to obtain an upper bound on the length of $P^*_{[u_2,w^*]}$.
To this end, we first observe that Lemma~\ref{lem:lower_bound_meeting_time} implies an upper bound on the potential $\pi(u_2)$. 

\begin{lemma}\label{lem:upper_bound_potential}
We have $\pi (u_2) \le \max \Bigl\{ 0,\ \min \bigl\{ t^*,\ \tmerge^k \bigr\} -   \tfrac{\tmerge^k}{1+\alpha\delta} \Bigr\} + \tfrac{\epsilon}{1+\delta}$.
\end{lemma}
\begin{proof}
Recall that $(u_1,u_2)$ is the last edge of $P^*$ with $t^{S^*u_1} \le \tfrac{\tmerge^k}{1+\alpha\delta}$.
We will prove the following upper bound on $\pi(u_1)$:
\begin{equation}\label{lem:upper_bound_potential:eq1}
\pi(u_1) \leq \max \Bigl\{ 0,\ \min \bigl\{ t^*,\ \tmerge^k \bigr\} -   \tfrac{\tmerge^k}{1+\alpha\delta} \Bigr\}.     
\end{equation}
This implies the claim by applying Lemma~\ref{lem:key_property_potential} to $e=(u_1,u_2)$: 
\[
\pi(u_2)\ \le\ \pi(u_1) + \widetilde{z}_e - z_e^k\ \le\ \pi(u_1) + \tfrac{c(e)}{1+\delta}\ \le\  \pi(u_1) + \tfrac{\epsilon}{1+\delta}.
\]

If $\pi(u_1)=0$, the claimed bound \eqref{lem:upper_bound_potential:eq1} is trivially satisfied. 
Hence, we may assume $\pi(u_1) > 0$.
Let $\{ S \in \mathcal{C}^k : \pi(S, u_1) > 0 \} = \{ S_1, S_2, \dots, S_l \}$ with $S_1 \subsetneq S_2 \subsetneq \dots \subsetneq S_l$.
Then $\pi(u_1) =\sum_{i=1}^l \pi(S_i, u_1)$ and for $i\in\{1,\dots,l\}$,
\begin{align*}
\pi(S_i, u_1) \ =&\  \max\bigl\{ 0, \ \min \bigl\{ t^*,\  d^{S_i}\bigr\} - \max \bigl\{ t^{S_i u_1},\  a^{S_i} \bigr\}\bigr\} \\
=&\  \min\bigl\{ t^*,\  d^{S_i}\bigr\} - \max \bigl\{ t^{S_i u_1},\  a^{S_i} \bigr\}.
\end{align*}

For $i\in \{1,\dots, l-1\}$, we have $a^{S_{i+1}} \ge d^{S_i}$ because $S_i \subsetneq S_{i+1}$.
Hence, we have $\max \bigl\{ t^{S_{i+1}u_1},\  a^{S_{i+1}} \bigr\} \ge \min\bigl\{ t^*,\  d^{S_i}\bigr\}$ and thus
\begin{align*}
\pi(u_1) =&\ \sum_{i=1}^l \pi(S_i, u_1) \\
=&\ \sum_{i=1}^l    \bigl( \min\bigl\{ t^*,\  d^{S_i}\bigr\} - \max \bigl\{ t^{S_i u_1},\  a^{S_i} \bigr\}\bigr) \\
 \le&\ \min\bigl\{ t,\  d^{S_l}\bigr\} - \max \bigl\{ t^{S_1 u_1},\  a^{S_1} \bigr\}.
\end{align*}
The definition of the edge $(u_1,u_2)$ implies $t^{S^*u_1} \le \tfrac{\tmerge^k}{1+\alpha\delta}$.
Hence, by Lemma~\ref{lem:lower_bound_meeting_time} for every set $S\in \mathcal{C}^k$ we have $t^{Su_1} > \tfrac{\tmerge^k}{1+\alpha\delta}$  (unless $t^{Su_1} \ge t^*$, but then $t^{Su_1} \ge t^* >  \tfrac{\tmerge^k}{1+\alpha\delta}$ by Corollary~\ref{cor:merge_time_and_function_f}).
By Observation~\ref{obs:lower_bound_merge_time_sets}, we have $\tmerge^k\geq d^{S_l}$.
We conclude
\begin{align*}
\pi(u_1) \le&\ \min\bigl\{ t^*,\  d^{S_l}\bigr\} - \max \bigl\{ t^{S_1u_1},\  a^{S_1} \bigr\}  \\
\le&\  \min\bigl\{ t^*,\  \tmerge^k \bigr\} - \max \left\{ \tfrac{\tmerge^k}{1+\alpha\delta},\  a^{S_1} \right\} \\
\le&\  \min\bigl\{ t^*,\  \tmerge^k \bigr\} - \tfrac{\tmerge^k}{1+\alpha\delta}. 
\end{align*}
\end{proof}

In order to prove an upper bound on the length of the path $P^*_{[u_2,w^*]}$, we consider the last vertex $u_3$ of $P^*_{[u_2,w^*]}$ with $t^{S^*u_3} \le \tmerge^k$ (recall Figure~\ref{fig:analysis_division_subpath} for an illustration).
If no such vertex exists, we let $u_3=u_2$. We will upper bound the length of the subpaths $P^*_{[u_2,u_3]}$ and $P^*_{[u_3,w^*]}$ separately in the next two lemmas.
First, we use the potential function $\pi$ to bound the length of $P^*_{[u_2,u_3]}$.

 \begin{lemma}\label{lem:bound_subpath_potential}
We have
 \[
 c(P^*_{[u_2,u_3]}) \ \le\ (1+\delta) \cdot 2 \cdot \Bigl( t^k_{\mathrm{merge}} - \tfrac{\tmerge^k}{1+\alpha\delta} \Bigr) +  (1+\delta) \cdot \Bigl( \min\bigl\{ t^*,\  \tmerge^k \bigr\} - \tfrac{\tmerge^k}{1+\alpha\delta} \Bigr) + \epsilon.
 \]
 \end{lemma}
 \begin{proof}
Recall that  $(u_1,u_2)$ was the last edge of $P^*$ with $t^{S^*u_1} \le \tfrac{\tmerge^k}{1+\alpha\delta}$.
We may assume $u_2 \neq u_3$ because otherwise the claimed bound on $c(P^*_{[u_2,u_3]})$ is trivially satisfied (using $t^* \ge \tfrac{\tmerge^k}{1+\alpha\delta}$ by Corollary~\ref{cor:merge_time_and_function_f}).
Then $t^{S^*u_3} \le t^k_{\mathrm{merge}}$ by the choice of $u_3$.
Moreover, the edge $(u_1,u_2)$ is not the last edge of $P^*$ and thus
$t^{S^*u_2} > \tfrac{\tmerge^k}{1+\alpha\delta}$.
 
 By Lemma~\ref{lem:only_last_chain_contributes_to_new_subpath}, the only sets that contributed to edges of $P^*_{[u_2,u_3]}$, are sets containing the terminal $s^*$ and sets from the chain $\mathcal{C}^k$.
Therefore, $\widetilde{z}_e \le t^{S^*w} -t^{S^*v}$ for every edge $e=(v,w)$ of $P^*_{[u_2,u_3]}$ by Lemma~\ref{lem:contribution_basic}.
Using that every edge of $P^*$ is $\delta$-tight and satisfies \eqref{eq:key_property_of_potential} by Lemma~\ref{lem:key_property_potential}, we conclude
\begin{align*}
c(P^*_{[u_2,u_3]})\ =&\ (1+\delta)\cdot \sum_{e\in P^*_{[u_2,u_3]}}\bigl(\tilde{z}_e+z^k_e\bigr) \\
\overset{Lem. \ref{lem:key_property_potential}}{\le} &\ (1+\delta) \cdot \Bigl( \sum_{e\in P^*_{[u_2,u_3]}} 2 \cdot \widetilde{z}_e  + \pi(u_2) -\pi(u_3)  \Bigr) \\
\le&\  (1+\delta) \cdot \Bigl( 2\bigl(t^{S^*u_3} - t^{S^*u_2}\bigr) + \pi(u_2) -\pi(u_3)  \Bigr) \\
\le&\ (1+\delta) \cdot \Bigl( 2\bigl( t^k_{\mathrm{merge}} - \tfrac{\tmerge^k}{1+\alpha\delta} \bigr) + \pi(u_2)  - 0\Bigr) \\
\overset{Lem. \ref{lem:upper_bound_potential}}{\le}&\ (1+\delta) \cdot 2 \cdot \Bigl( t^k_{\mathrm{merge}} - \tfrac{\tmerge^k}{1+\alpha\delta} \Bigr) +  (1+\delta) \cdot \max\Bigl\{0,\min\bigl\{ t^*,\  \tmerge^k \bigr\} - \tfrac{\tmerge^k}{1+\alpha\delta}\Bigr\}  + \epsilon\\
= &\  (1+\delta) \cdot 2 \cdot \Bigl( t^k_{\mathrm{merge}} - \tfrac{\tmerge^k}{1+\alpha\delta} \Bigr) +  (1+\delta) \cdot \Bigl( \min\bigl\{ t^*,\  \tmerge^k \bigr\} - \tfrac{\tmerge^k}{1+\alpha\delta} \Bigr) + \epsilon,
\end{align*}
where in the last equality we used $t^*\geq \tfrac{\tmerge^k}{1+\alpha\delta}$ by Corollary~\ref{cor:merge_time_and_function_f}.
\end{proof}

We next upper bound the length of $P^*_{[u_3,w^*]}$.
To this end, we will exploit the fact that the chain $\mathcal{C}^k$ can only contribute to the first edge of the path $P^*_{[u_3,w^*]}$, but not to any other edge of $P^*_{[u_3,w^*]}$ (recall Figure~\ref{fig:analysis_division_subpath} for an illustration).

  \begin{lemma}\label{lem:bound_length_of_trivial_subpath}
If $u_3 \neq w^*$, we have 
 \[
 c(P^*_{[u_3,w^*]}) \ <\ (1+\delta)\cdot \bigl( t^* -   \tmerge^k \bigr) + \epsilon.
 \]
 \end{lemma}
 \begin{proof}
Because $u_3 \neq w^*$, the choice of $u_3$ implies $\tmerge^k < t^{S^*w^*} \le t^*$.
 Hence, because $d^{S^*} > t^*$ (by assumption), every set $S\in \mathcal{C}^k$ is a subset of the set $S^*$ by Observation~\ref{obs:deactivation_time_lower_bounds_merge_time}.
 Therefore,  by Observation~\ref{obs:monotonicity_reaching_times} we have $t^{Sv} \ge t^{S^*v}$ for every vertex $v$ and every set $S\in \mathcal{C}^k$.
 
 Let $u_4$ be the successor of $u_3$ on the path $P^*$.
 Then by the choice of $u_3$, we have $t^{S^*v} > \tmerge^k $ for every vertex $v$ of $P^*_{[u_4,w^*]}$.
 For every set $S\in \mathcal{C}^k$ and every vertex $v$ of $P^*_{[u_4,w^*]}$ we therefore have $t^{Sv} \ge t^{S^*v} > \tmerge^k \ge d^S$ (using Observation~\ref{obs:lower_bound_merge_time_sets} in the last inequality).
 This shows that no set from the chain $\mathcal{C}^k$ contributed to any edge of $P^*_{[u_4,w^*]}$.
 Hence, by Lemma~\ref{lem:only_last_chain_contributes_to_new_subpath}, no set from $\mathcal{C}$ contributed to $P^*_{[u_4,w^*]}$.
 This implies that for every edge $e=(v,w)$ of $P^*_{[u_4,w^*]}$ we have $t^{S^*w} - t^{S^*v} \ge \frac{c(e)}{1+\delta}$ by Lemma~\ref{lem:contribution_basic}.
 We conclude
 \begin{align*}
 c(P^*_{[u_3,w^*]}) \ \le&\  c(P^*_{[u_4,w^*]}) + \epsilon \\
  \le&\ \sum_{(v,w)\in P^*_{[u_4,w^*]}} (1+\delta) \cdot (t^{S^*w} - t^{S^*v})  + \epsilon \\
  =&\ (1+\delta) \cdot (t^{S^*w^*} - t^{S^*u_4})  + \epsilon \\
  <&\ (1+\delta)\cdot \bigl( t^* -  \tmerge^k  \bigr) + \epsilon,
 \end{align*}
where we used that $t^{S^*u_4} > \tmerge^k $ by the definition of $u_3$.
 \end{proof} 
 
 Combining Lemma~\ref{lem:bound_subpath_potential} and Lemma~\ref{lem:bound_length_of_trivial_subpath} we obtain the main result of this section.
 
 \begin{lemma}\label{lem:bound_length_new_subpath}
  We have 
 \[
 c(P^*_{[u_2, w^*]}) \ \le\ (1+\delta) \cdot t^* +  (1+\delta) \cdot 2 \cdot \Bigl( t^k_{\mathrm{merge}} - \tfrac{\tmerge^k}{1+\alpha\delta} \Bigr) + 2\epsilon -  (1+\delta)\cdot \tfrac{\tmerge^k}{1+\alpha\delta}.
 \]
 \end{lemma}
 \begin{proof}
 If $u_3=w^*$, Lemma~\ref{lem:bound_subpath_potential} implies
  \begin{align*}
 c(P^*_{[u_2, w^*]}) \ =&\ c(P^*_{[u_2,u_3]}) \\
  \le&\ (1+\delta) \cdot 2 \cdot \Bigl( t^k_{\mathrm{merge}} - \tfrac{\tmerge^k}{1+\alpha\delta} \Bigr) +  (1+\delta) \cdot \Bigl( \min\bigl\{ t^*,\  \tmerge^k \bigr\} - \tfrac{\tmerge^k}{1+\alpha\delta} \Bigr) + \epsilon \\
  <&\  (1+\delta) \cdot t^* +  (1+\delta) \cdot 2 \cdot \Bigl( t^k_{\mathrm{merge}} - \tfrac{\tmerge^k}{1+\alpha\delta} \Bigr) + 2\epsilon -  (1+\delta)\cdot \tfrac{\tmerge^k}{1+\alpha\delta}.
 \end{align*}
 
Otherwise, Lemma~\ref{lem:bound_subpath_potential} and Lemma~\ref{lem:bound_length_of_trivial_subpath} imply
 \begin{align*}
 c(P^*_{[u_2, w^*]}) \ =&\ c(P^*_{[u_2,u_3]}) +  c(P^*_{[u_3, w^*]}) \\
  \le&\ (1+\delta) \cdot 2 \cdot \Bigl( t^k_{\mathrm{merge}} - \tfrac{\tmerge^k}{1+\alpha\delta} \Bigr) +  (1+\delta) \cdot \Bigl( \min\bigl\{ t^*,\  \tmerge^k \bigr\} - \tfrac{\tmerge^k}{1+\alpha\delta} \Bigr) + \epsilon \\
  &\ + (1+\delta)\cdot \bigl( t^* - \tmerge^k \bigr) + \epsilon  \\
  \le&\  (1+\delta) \cdot t^* +  (1+\delta) \cdot 2 \cdot \Bigl( t^k_{\mathrm{merge}} - \tfrac{\tmerge^k}{1+\alpha\delta} \Bigr) + 2\epsilon -  (1+\delta)\cdot \tfrac{\tmerge^k}{1+\alpha\delta}.
 \end{align*}
 \end{proof}

\subsection{Completing the proof of Lemma~\ref{lem:distance_bound}}\label{sec:completing_induction}

In this section we prove~\eqref{item:strong_existing_component} of Claim~\ref{claim:main} for the $S^*$-tight path $P^*$at the current time $t^*$.
This completes the proof of Lemma~\ref{lem:distance_bound} and hence of Lemma~\ref{lem:dual_feasible} and Theorem~\ref{thm:main}.

Because $(u_1,u_2)$ is an edge of $P^*$, the $S^*$-tight path $P^*_{[s^*,u_1]}$  has strictly fewer edges than $P^*$ and hence we can apply the induction hypothesis to it.
By the definition of $u_1$, we have $t^{S^*u_1} \le \tfrac{\tmerge^k}{1+\alpha\delta}$ and hence $P^*_{[s^*,u_1]}$ was $S^*$-tight at time  $\tfrac{\tmerge^k}{1+\alpha\delta}$.
By \eqref{item:strong_existing_component} of the induction hypothesis, we obtain a terminal set $X$ with $s^*\in X$  and a component $K$ connecting $X\cup \{u_1\}$ such that
\begin{align*}
c(P^*_{[s^*,u_1]}) \ \le&\ (1+\delta) \cdot  \tfrac{\tmerge^k}{1+\alpha\delta} + \lambda \cdot \drop(X, s^*), \\
c(K) - \tfrac{1}{1+\gamma} \cdot \drop(X, s^*) \ \le&\  (1+\delta) \cdot \tfrac{\tmerge^k}{1+\alpha\delta} - \mu \cdot \drop(X,s^*) \\
\tmerge(s^*,s) \ <&\  (1+\alpha\delta) \cdot \tfrac{\tmerge^k}{1+\alpha\delta} = \tmerge^k & \text{ for all }s \in X.
\end{align*}

Now consider a set $\overline{S}\in \mathcal{C}^k$.
Because $\overline{S}$ contributes to some edge of $P^*$, Corollary~\ref{cor:last_chain_contributes_only_to_new_subpath} implies that there exists a vertex $v$ of $P^*_{[u_1,w^*]}$ with $t^{\overline{S}v} < \min\{t^*, d^{\overline{S}}\}$.
Hence, by Lemma~\ref{lem:reachability_implies_S-tight_path} at time $t^{\overline{S}v}$ there exists an $\overline{S}$-tight path $\overline{P}$ from some terminal $\overline{s}\in \overline{S}$ to the vertex $v$ on the path $P^*_{[u_1,w^*]}$.
By \eqref{item:actual_distance_bound} of the induction hypothesis applied to the path $\overline{P}$, we have 
\begin{equation}\label{eq:bound_new_connection}
c(\overline{P}) < (1+\beta\delta) \cdot \min\{t^*, d^{\overline{S}}\}.
\end{equation}
\bigskip

To prove \eqref{item:strong_existing_component} for the $S^*$-tight path $P^*$at the current time $t^*$ we need to show that there exists a terminal set $X^*$ with $s^*\in X^*$  and a component $K^*$ connecting $X^*\cup \{w^*\}$ such that
\begin{align}
c(P^*) \ \le&\ (1+\delta) \cdot t^* + \lambda \cdot \drop(X^*, s^*) 
 \label{item:cost_path_to_show}\\
c(K^*) - \tfrac{1}{1+\gamma} \cdot \drop(X^*, s^*) \ \le&\  (1+\delta) \cdot t^* - \mu \cdot \drop(X^*,s^*) 
\label{item:cost_comp_to_show}\\
\tmerge(s^*,s) \ <&\  (1+\alpha\delta) \cdot t^* & \text{ for all }s \in X^*
\label{item:bound_merge_time_to_show}
\end{align}
\begin{figure}
\begin{center}
\begin{tikzpicture}[very thick, yscale=0.5, xscale=1]
\tikzset{terminal/.style={
ultra thick,draw,fill=none,rectangle,minimum size=0pt, inner sep=3pt, outer sep=2.5pt}
}
\tikzset{steiner/.style={
fill=black,circle,inner sep=0em,minimum size=0pt, inner sep=2pt, outer sep=1.5pt}
}

\node[terminal] (start) at (-1.5, -4) {};
\node[left=3pt] at (start) {$s^*$};
\node[terminal] (s1) at (0, -6.5) {};
\node[terminal] (s2) at (1.5, -1) {};
\node[terminal] (s3) at (5.5, -7.5) {};
\node[left=3pt] at (s3) {$\overline{s}$};
\node[steiner] (u1) at (4,-4) {};
\node[above=2pt] at (u1) {$u_1$};
\node[steiner] (v) at (6.5,-4) {};
\node[above=2pt] at (v) {$v$};
\node[steiner] (w) at (10,-4) {};
\node[above=2pt] at (w) {$w^*$};

\coordinate (m1) at (1,-4);
\coordinate (m2) at (2.5,-4);

\begin{scope}[ultra thick, blue]
\draw (start) -- (m1) -- (m2) -- (u1);
\draw (s1) -- (m1);
\draw (s2) -- (m2);
\end{scope}

\node[blue] at (0.5,-2.8) {$K$};

\begin{scope}[->, >=latex, ultra thick, densely dashed, orange]
\draw (u1) to (v);
\draw (v) to (w);
\end{scope}

\node[orange!90!black] at (8,-2.6) {$P^*_{[u_1,w^*]}$};

\begin{scope}[->, >=latex, ultra thick, densely dotted, green!70!black]
\draw (s3) to (v);
\end{scope}

\node[green!70!black] at (6.5,-5.9) {$\overline{P}$};

\end{tikzpicture}
\caption{\label{fig:completing_induction}
An illustration of the component $K^*$ that we construct from the component $K$, the path $P^*_{[u_1,w^*]}$, and the path $\overline{P}$.
The terminals in $X^*$, which are connected by the component $K^*$, are shown as squares.
}
\end{center}
\end{figure}
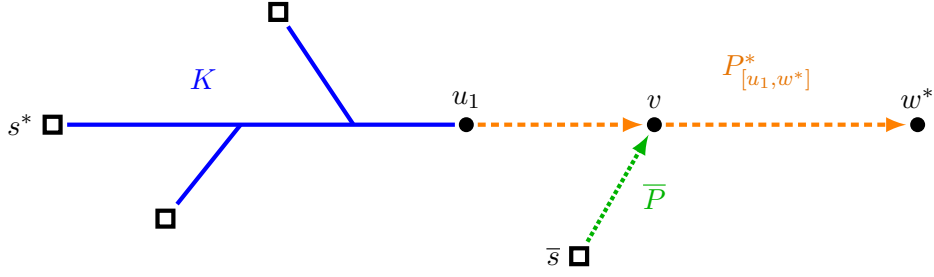

We define $X^* \coloneqq X \cup \{\overline{s}\}$ and let $K^*$ be the union of the component $K$, the path $P^*_{[u_1,w^*]}$, and the path~$\overline{P}$.
See Figure~\ref{fig:completing_induction}.
Then $K^*$ is a component connecting $X^*\cup \{w^*\}$.
It remains to show \eqref{item:cost_path_to_show} -- \eqref{item:bound_merge_time_to_show}.
\bigskip

We first show \eqref{item:bound_merge_time_to_show}.
Recall that for $s\in X$ we have $\tmerge(s^*,s) < \tmerge^k$.
Moreover, because $\overline{s} \in \overline{S} \in \mathcal{C}^k$, we have $\tmerge(s^*,\overline{s}) = \tmerge^k$.
By Corollary~\ref{cor:merge_time_and_function_f}, we have $\tmerge^k < (1+\alpha\delta) \cdot t^*$, completing the proof of~\eqref{item:bound_merge_time_to_show}.
\bigskip

Before proving \eqref{item:cost_path_to_show} and \eqref{item:cost_comp_to_show}, we first derive a lower bound on $\drop(X^*, s^*)$.
Let $\mathcal{S}^*$ be a drop certificate for $X$ of value $\drop(X,s^*)$ such that $s^*\notin S$ for all $S\in \mathcal{S}^*$ (which exists by the definition of $\drop(X,s^*)$).

We now show that $\mathcal{S}^* \cup \{\overline{S}\}$ is a drop certificate for $X^*$.
Note that we have $s^* \notin \overline{S}$ because $\overline{S} \in \mathcal{C}^k$.
For any two terminals $s_1,s_2 \in X$, the drop certificate $\mathcal{S}^*$ contains a set $S$ with $|S\cap \{s_1,s_2\}| =1$.
Consider now a terminal $s\in X$. 
Because $\tmerge(s,s^*) < \tmerge^k$ and  $\tmerge(s^*, \overline{s}) =\tmerge^k$, Observation~\ref{obs:merge_time_max} implies $\tmerge(s,\overline{s}) \ge \tmerge^k \ge d^{\overline{S}}$, where the second inequality follows from Observation~\ref{obs:lower_bound_merge_time_sets}.
Hence, we have $s \notin \overline{S}$ by Observation~\ref{obs:deactivation_time_lower_bounds_merge_time}.
In particular, this implies that $s\neq \overline{s}$ and $\mathcal{S}^* \cup \{\overline{S}\}$ is a drop certificate for $X^*$ .
Because $s^*\notin S$ for all $S\in \mathcal{S}^* \cup \{\overline{S}\}$ we conclude 
\begin{equation}\label{eq:final_lower_bound_drop}
\drop(X^*, s^*) \ge \drop(X,s^*) + 2 \cdot d^{\overline{S}}.
\end{equation}
\bigskip

We now prove \eqref{item:cost_path_to_show}.
Notice that $c(P^*_{[u_1,w^*]}) \le \epsilon + c(P^*_{[u_2,w^*]})$ and recall that we have $c(P^*_{[s^*,u_1]}) \ \le\ (1+\delta) \cdot  \tfrac{\tmerge^k}{1+\alpha\delta} + \lambda \cdot \drop(X, s^*)$ by the induction hypothesis.
Then we obtain \eqref{item:cost_path_to_show} as follows:
\begin{align*}
c(P^*) \ \le&\ c(P^*_{[s^*,u_1]}) + \epsilon + c(P^*_{[u_2,w^*]}) \\
\le&\ (1+\delta) \cdot  \tfrac{\tmerge^k}{1+\alpha\delta} + \lambda \cdot \drop(X, s^*) + \epsilon + c(P^*_{[u_2,w^*]})\\
\overset{Lem. \ref{lem:bound_length_new_subpath}}{\le}&\   (1+\delta) \cdot  \tfrac{\tmerge^k}{1+\alpha\delta} + \lambda \cdot \drop(X, s^*)+ \epsilon\\
&\ + (1+\delta) \cdot t^* +  (1+\delta) \cdot 2 \cdot \bigl( t^k_{\mathrm{merge}} - \tfrac{\tmerge^k}{1+\alpha\delta} \bigr) + 2\epsilon - (1+\delta) \cdot  \tfrac{\tmerge^k}{1+\alpha\delta} \\
=&\ (1+\delta) \cdot t^* +  (1+\delta) \cdot 2 \cdot \bigl( t^k_{\mathrm{merge}} - \tfrac{\tmerge^k}{1+\alpha\delta} \bigr) + 3 \epsilon + \lambda \cdot \drop(X, s^*)\\
=&\ (1+\delta) \cdot t^* +  2 \cdot \tfrac{1+\delta}{1+\alpha\delta} \cdot \alpha \delta \cdot \tmerge^k + 3 \epsilon + \lambda \cdot \drop(X, s^*)\\
\overset{Lem. \ref{lem:merge_time_approx_deactivation_time}}{\le}&\ (1+\delta) \cdot t^* +  2 \cdot \tfrac{1+\delta}{1+\alpha\delta} \cdot \alpha\delta \cdot \tfrac{1+\beta\delta}{1-\beta\delta} \cdot d^{\overline{S}} + 3 \epsilon + \lambda \cdot \drop(X, s^*)\\
\overset{\eqref{eq:choice_of_epsilons}}{\le}&\ (1+\delta) \cdot t^* +  \highlight{2 \cdot \tfrac{1+\delta}{1+\alpha\delta} \cdot \alpha\delta \cdot \tfrac{1+\beta\delta}{1-\beta\delta} \cdot d^{\overline{S}} + 3 \epsilon'\cdot d^{\overline{S}} }+ \lambda \cdot \drop(X, s^*)\\
\le&\ (1+\delta) \cdot t^* + \highlight{2\lambda\cdot d^{\overline{S}} } + \lambda \cdot \drop(X, s^*)\\
\overset{\eqref{eq:final_lower_bound_drop}}{\le}&\ (1+\delta) \cdot t^* +\lambda \cdot \drop(X^*, s^*),
\end{align*}
where in the second-last inequality we used the values of the constants given in Table~\ref{table:constants}.
\bigskip

Finally, we prove \eqref{item:cost_comp_to_show}. 
Recall that $c(K)\leq (1+\delta) \cdot \tfrac{\tmerge^k}{1+\alpha\delta} +\bigl( \tfrac{1}{1+\gamma} - \mu \bigr) \cdot \drop(X,s^*)$ by the induction hypothesis.
Using again $c(P^*_{[u_1,w^*]}) \le \epsilon + c(P^*_{[u_2,w^*]})$, we can conclude \eqref{item:cost_comp_to_show} as follows:
\begin{align*}
c(K^*) \ \le&\ c(K) + \epsilon + c(P^*_{[u_2,w^*]}) + c(\overline{P})\\
\le&\ (1+\delta) \cdot \tfrac{\tmerge^k}{1+\alpha\delta} +\bigl( \tfrac{1}{1+\gamma} - \mu \bigr) \cdot \drop(X,s^*) + \epsilon + c(P^*_{[u_2,w^*]}) + c(\overline{P})\\
\overset{\eqref{eq:bound_new_connection}}{\leq} &\ (1+\delta) \cdot \tfrac{\tmerge^k}{1+\alpha\delta} +\bigl( \tfrac{1}{1+\gamma} - \mu \bigr) \cdot \drop(X,s^*) + \epsilon + c(P^*_{[u_2,w^*]}) + (1+\beta\delta) \cdot \min\{t^*, d^{\overline{S}}\}\\
\overset{Lem. \ref{lem:bound_length_new_subpath}}{\le}&\   (1+\delta) \cdot \tfrac{\tmerge^k}{1+\alpha\delta} +\bigl( \tfrac{1}{1+\gamma} - \mu \bigr) \cdot \drop(X,s^*)  + \epsilon \\
&\ + (1+\delta) \cdot t^* +  (1+\delta) \cdot 2 \cdot \bigl( t^k_{\mathrm{merge}} - \tfrac{\tmerge^k}{1+\alpha\delta} \bigr) + 2\epsilon -  (1+\delta)\cdot \tfrac{\tmerge^k}{1+\alpha\delta} \\
&\ + (1+\beta\delta) \cdot \min\{t^*, d^{\overline{S}}\} \\
\le&\  (1+\delta) \cdot t^* +\bigl( \tfrac{1}{1+\gamma} - \mu \bigr) \cdot \drop(X,s^*) + 3\epsilon + (1+\beta\delta) \cdot d^{\overline{S}}  +  (1+\delta) \cdot 2 \cdot \bigl( t^k_{\mathrm{merge}} - \tfrac{\tmerge^k}{1+\alpha\delta} \bigr) \\
=&\ (1+\delta) \cdot t^* +\bigl( \tfrac{1}{1+\gamma} - \mu \bigr) \cdot \drop(X,s^*) + 3\epsilon + (1+\beta\delta) \cdot d^{\overline{S}}  +  2\alpha \delta \cdot \tfrac{1+\delta}{1+\alpha\delta} \cdot \tmerge^k \\
\overset{Lem. \ref{lem:merge_time_approx_deactivation_time}}{\le}&\  (1+\delta) \cdot t^* +\bigl( \tfrac{1}{1+\gamma} - \mu \bigr) \cdot \drop(X,s^*) + 3\epsilon + (1+\beta\delta) \cdot d^{\overline{S}}  +  2\alpha\delta \cdot \tfrac{1+\delta}{1+\alpha\delta} \cdot \tfrac{1+\beta\delta}{1-\beta\delta} \cdot d^{\overline{S}}\\
\overset{\eqref{eq:choice_of_epsilons}}{\leq} &\
(1+\delta) \cdot t^* +\bigl( \tfrac{1}{1+\gamma} - \mu \bigr) \cdot \drop(X,s^*) + \highlight{3\eps'\cdot d^{\overline{S}} + (1+\beta\delta) \cdot d^{\overline{S}}  +  2\alpha\delta \cdot \tfrac{1+\delta}{1+\alpha\delta} \cdot \tfrac{1+\beta\delta}{1-\beta\delta} \cdot d^{\overline{S}}}\\
\leq &\ (1+\delta) \cdot t^* +\bigl( \tfrac{1}{1+\gamma} - \mu \bigr) \cdot \drop(X,s^*) + \highlight{\bigl( \tfrac{1}{1+\gamma} - \mu \bigr)\cdot 2d^{\overline{S}}}\\
\overset{\eqref{eq:final_lower_bound_drop}}{\leq} &\  (1+\delta) \cdot t^* +\bigl( \tfrac{1}{1+\gamma} - \mu \bigr) \cdot \drop(X^*,s^*),
\end{align*}
where in the second-last inequality we again plugged in the constants shown in Table~\ref{table:constants} and in the last inequality we used \highlight{$\tfrac{1}{1+\gamma} - \mu \ge 0$}.

This completes the proof of Lemma~\ref{lem:distance_bound} and hence also of Lemma~\ref{lem:dual_feasible} and Theorem~\ref{thm:main}.

\appendix
\section{Necessity of non-laminar dual solutions}
\label{sec:non-laminar}

A solution $y$ of \eqref{eq:dual-bcr-tree} has \emph{laminar support} if the support $\mathcal{L} \coloneqq \{ U: y_U > 0\}$ of $y$ is a laminar family, i.e., for all sets $A,B\in \mathcal{L}$ we have $A\subseteq B$, $B\subseteq A$, or $A\cap B = \emptyset$.

In this section we provide a family of instances of the Steiner tree problem where the maximum value of a solution to \eqref{eq:dual-bcr-tree} with laminar support is arbitrarily close to $\frac{1}{2}\opt$, where $\opt$ denotes the minimum cost of a Steiner tree.
This shows that if for some fixed constant $\rho> 0$ (independent of the specific instance), we want to find a Steiner tree $T$ and a solution $y$ to  \eqref{eq:dual-bcr-tree} with $c(T) \le (2-\rho)\cdot \sum_{\emptyset \neq U\subseteq V\setminus \{r\}} y_U$ , we cannot restrict ourselves to dual solutions $y$ with laminar support.

We consider the family of instances illustrated in Figure~\ref{fig:example_nonlaminar}.
\begin{figure}[h]
\begin{center}
\begin{tikzpicture}
\tikzset{terminal/.style={
ultra thick,draw,fill=none,rectangle,minimum size=0pt, inner sep=3pt, outer sep=2.5pt}
}
\tikzset{steiner/.style={
fill=black,circle,inner sep=0em,minimum size=0pt, inner sep=2pt, outer sep=1.5pt}
}

\def\numn{20}
\def\numk{5}
\def\rad{2}

\foreach \i in {1,...,\numn} {
  \pgfmathsetmacro\r{90+(\i)*360/(\numn)}
  \pgfmathparse{Mod(\i,\numn / \numk)==0?1:0}
      \ifthenelse{\pgfmathresult > 0}{
         \node[terminal] (t\i) at (\r:\rad) {};
      }{
         \node[steiner] (t\i) at (\r:\rad) {};
      }
}
\foreach \i in {1,...,\numn} {
  \pgfmathsetmacro\j{int(Mod(\i, \numn) + 1)}
   \draw[very thick] (t\i) -- (t\j);
}
\node[above=2pt] (r) at (t\numn) {$r$};

\end{tikzpicture}
\caption{\label{fig:example_nonlaminar}
Illustration of a family of Steiner tree instances, where the graph $G$ is a cycle with $n$ vertices, out of which $k$ are terminals (shown as squares). In the depicted example, $n=20$ and $k=5$.
We choose $n$ to be divisible by $k$ and every $\frac{n}{k}$th vertex along the cycle is a terminal.
Every edge has cost one.
}
\end{center}
\end{figure}
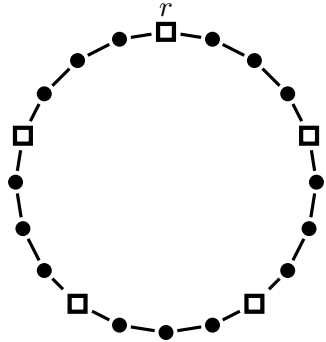

We have an instance for every integer $k \ge 2$ and every integer $n\ge k$ that is an integer multiple of $k$.
The cost of an optimum Steiner tree for such an instance is $\opt=\frac{k-1}{k}\cdot n$.
We show that for $k$ and $\frac{n}{k}$ large enough the maximum value of a solution to \eqref{eq:dual-bcr-tree} with laminar support gets arbitrarily close to $\frac{1}{2}\opt$.

\begin{lemma}\label{lem:bound_laminar}
Let $y$ be a solution to \eqref{eq:dual-bcr-tree} with laminar support.
Then $y$ has value
\[
 \sum_{\emptyset \neq U \subseteq V\setminus\{r\}} y_U \ \le\ \frac{n}{2}+k.
\]
\end{lemma}
\begin{proof}
Let $\mathcal{L} \coloneqq \{ U: y_U > 0\}$ denote the support of $y$.
By assumption, $\mathcal{L}$ is a laminar family.

We first show that we may assume that every element of $\mathcal{L}$ is a set of consecutive vertices of the cycle $G$.
If this is not the case, consider a minimal element $U$ of $\Lscr$ violating this condition. 
Then we can write the cut $\delta^+(U)$ as the disjoint union of cuts $\delta^+(U_1), \dots, \delta^+(U_p)$, where each set $U_i \subseteq V\setminus \{r\}$ is a nonempty set of consecutive vertices of the cycle $G$ and $p \ge 2$.
Then decreasing the value of $y_U$ to zero and increasing $y_{U_1}$ by the same amount, we maintain a solution to \eqref{eq:dual-bcr-tree} of the same value.
Moreover, we maintain laminarity of the support of $y$ by the minimal choice of $U$.
Hence, we may assume that every element of $\mathcal{L}$ is a set of consecutive vertices of $G$.

Let $\mathcal{L}^R \coloneqq \{ U \cap R : U\in \mathcal{L}\}$.
Then $\mathcal{L}^R$ is a laminar family because $\mathcal{L}$ is a laminar family.
Because $\mathcal{L}^R$ is a laminar family of subsets of $R$, it has at most $2k$ elements.
For each set $X\in \mathcal{L}^R$, the family $\{ U \in \mathcal{L} : U\cap R = X\}$ is a chain.
We write $U_X$ to denote the maximal element of this chain.
Now consider an edge $e=\{v,w\}$ of $G$, where the contribution of $y$ to each of $(v,w)$ and $(w,v)$ is positive, i.e, $ \sum_{U : (v,w)\in \delta^+(U)} y_U > 0 $ and 
$\sum_{U : (w,v)\in \delta^+(U)} y_U \ >\ 0$. 
Then $e \in \delta(U_X)$ for two disjoint sets $X \in \mathcal{L}^R$.
Because $\delta(U)$ contains exactly two edges (in $G$) for every set $U\in \mathcal{L}$ and $|\mathcal{L}^R|\le 2k$, this implies that there are at most $2k$ edges $\{v,w\}$ with this property.
We conclude 
\[
 \sum_{\{v,w\}\in E}\left( \sum_{U : (v,w)\in \delta^+(U)} y_U + \sum_{U : (w,v)\in \delta^+(U)} y_U \right) \ \le\ n + 2k.
\]
Because every cut $\delta^+(U)$ with $U\in \mathcal{L}$ has exactly two elements (in $\overrightarrow{E}$),  this implies that the value of $y$ is at most $\frac{n}{2} + k$.
\end{proof}

Lemma~\ref{lem:bound_laminar} implies that on the instances from Figure~\ref{fig:example_nonlaminar}, every solution $y$ to \eqref{eq:dual-bcr-tree} with laminar support has value at most $\frac{n}{2}+k = (\frac{1}{2} - \frac{k}{n}) \cdot n$.
For $k$ and $\frac{n}{k}$ being large enough, the ratio of this value $(\frac{1}{2} - \frac{k}{n}) \cdot n$ and the cost $\opt=\frac{k-1}{k}\cdot n$ of an optimum Steiner tree becomes arbitrarily close to $\frac{1}{2}$.

Of course, for this particular family of instances a simple preprocessing could eliminate the Steiner nodes of degree $2$.
However, we might be given only the metric closure of the cycle $G$, or the graph $G$ might have some additional edges that have high enough cost such that they do not change the cost of an optimum Steiner tree.
Then considering dual solutions with laminar support is still not sufficient and the simple preprocessing eliminating degree-2 Steiner nodes does not help anymore. 

\subsubsection*{Acknowledgment}
We thank the Banff International Research Station and the organizers of the BIRS workshop ``Approximation Algorithms and the Hardness of Approximation'' in September 2023, during which fruitful discussions related to this work have taken place.

\bibliographystyle{alpha}

\begin{thebibliography}{AGH{\etalchar{+}}24b}

\bibitem[ABHK11]{ABHK11}
Aaron Archer, Mohammad~Hossein Bateni, MohammadTaghi Hajiaghayi, and Howard~J.
  Karloff.
\newblock Improved approximation algorithms for prize-collecting {S}teiner tree
  and {TSP}.
\newblock {\em {SIAM} J. Comput.}, 40(2):309--332, 2011.

\bibitem[AGH{\etalchar{+}}24a]{AGHJM24soda}
Ali Ahmadi, Iman Gholami, MohammadTaghi Hajiaghayi, Peyman Jabbarzade, and
  Mohammad Mahdavi.
\newblock 2-approximation for prize-collecting {S}teiner forest.
\newblock In {\em Proceedings of the 2024 {ACM-SIAM} Symposium on Discrete
  Algorithms, ({SODA})}, pages 669--693, 2024.

\bibitem[AGH{\etalchar{+}}24b]{AGHJM24}
Ali Ahmadi, Iman Gholami, MohammadTaghi Hajiaghayi, Peyman Jabbarzade, and
  Mohammad Mahdavi.
\newblock Prize-collecting {S}teiner tree: {A} 1.79 approximation.
\newblock In {\em Proceedings of the 56th Annual {ACM} Symposium on Theory of
  Computing, ({STOC})}, pages 1641--1652, 2024.

\bibitem[AKR95]{AKR91}
Ajit Agrawal, Philip~N. Klein, and R.~Ravi.
\newblock When trees collide: An approximation algorithm for the generalized
  {S}teiner problem on networks.
\newblock {\em {SIAM} J. Comput.}, 24(3):440--456, 1995.

\bibitem[BD95]{BD95}
Al~Borchers and Ding{-}Zhu Du.
\newblock The k-{S}teiner ratio in graphs.
\newblock In {\em Proceedings of the 27th Annual {ACM} Symposium on Theory of
  Computing ({STOC})}, pages 641--649, 1995.

\bibitem[BGA23]{BGJ23sicomp}
Jaroslaw Byrka, Fabrizio Grandoni, and Afrouz~Jabal Ameli.
\newblock Breaching the 2-approximation barrier for connectivity augmentation:
  {A} reduction to {S}teiner tree.
\newblock {\em {SIAM} J. Comput.}, 52(3):718--739, 2023.

\bibitem[BGRS13]{BGRS13}
Jaroslaw Byrka, Fabrizio Grandoni, Thomas Rothvo{\ss}, and Laura Sanit{\`{a}}.
\newblock Steiner tree approximation via iterative randomized rounding.
\newblock {\em J. {ACM}}, 60(1):6:1--6:33, 2013.

\bibitem[BGSW93]{BGSW93}
Daniel Bienstock, Michel~X. Goemans, David Simchi{-}Levi, and David~P.
  Williamson.
\newblock A note on the prize collecting traveling salesman problem.
\newblock {\em Math. Program.}, 59:413--420, 1993.

\bibitem[CC08]{CC08}
Miroslav Chleb{\'{\i}}k and Janka Chleb{\'{\i}}kov{\'{a}}.
\newblock The {S}teiner tree problem on graphs: Inapproximability results.
\newblock {\em Theor. Comput. Sci.}, 406(3):207--214, 2008.

\bibitem[CCC{\etalchar{+}}99]{CCCDGGL99}
Moses Charikar, Chandra Chekuri, To{-}Yat Cheung, Zuo Dai, Ashish Goel, Sudipto
  Guha, and Ming Li.
\newblock Approximation algorithms for directed {S}teiner problems.
\newblock {\em J. Algorithms}, 33(1):73--91, 1999.

\bibitem[CDV11]{CDV11}
Deeparnab Chakrabarty, Nikhil~R. Devanur, and Vijay~V. Vazirani.
\newblock New geometry-inspired relaxations and algorithms for the metric
  {S}teiner tree problem.
\newblock {\em Math. Program.}, 130(1):1--32, 2011.

\bibitem[CKP10]{CKP10}
Deeparnab Chakrabarty, Jochen K{\"{o}}nemann, and David Pritchard.
\newblock Hypergraphic {LP} relaxations for {S}teiner trees.
\newblock In {\em Proceedings of the 14th International Conference on Integer
  Programming and Combinatorial Optimization, ({IPCO})}, pages 383--396, 2010.

\bibitem[CTZ21]{CTZ21}
Federica Cecchetto, Vera Traub, and Rico Zenklusen.
\newblock Bridging the gap between tree and connectivity augmentation: unified
  and stronger approaches.
\newblock In {\em Proceedings of the 53rd Annual {ACM} Symposium on Theory of
  Computing ({STOC})}, pages 370--383, 2021.

\bibitem[DW71]{DW71}
Stuart~E. Dreyfus and Robert~A. Wagner.
\newblock The {S}teiner problem in graphs.
\newblock {\em Networks}, 1(3):195--207, 1971.

\bibitem[Edm67]{E67}
Jack Edmonds.
\newblock Optimum branchings.
\newblock {\em Journal of Research of the National Bureau of Standards},
  B71:233--240, 1967.

\bibitem[FKOS16]{FKOS16}
Andreas~Emil Feldmann, Jochen K{\"{o}}nemann, Neil Olver, and Laura
  Sanit{\`{a}}.
\newblock On the equivalence of the bidirected and hypergraphic relaxations for
  {S}teiner tree.
\newblock {\em Math. Program.}, 160(1-2):379--406, 2016.

\bibitem[FV20]{FV20}
Bartosz Filipecki and Mathieu~Van Vyve.
\newblock Stronger path-based extended formulation for the {S}teiner tree
  problem.
\newblock {\em Networks}, 75(1):3--17, 2020.

\bibitem[GAT22]{GJT22stoc}
Fabrizio Grandoni, Afrouz~Jabal Ameli, and Vera Traub.
\newblock Breaching the 2-approximation barrier for the forest augmentation
  problem.
\newblock In {\em Proceedings of the 54th Annual {ACM} Symposium on Theory of
  Computing ({STOC})}, pages 1598--1611, 2022.

\bibitem[GLL19]{GLL19}
Fabrizio Grandoni, Bundit Laekhanukit, and Shi Li.
\newblock \emph{O}(log\({}^{\mbox{2}}\) \emph{k} / log log
  \emph{k})-approximation algorithm for directed {S}teiner tree: a tight
  quasi-polynomial-time algorithm.
\newblock In {\em Proceedings of the 51st Annual {ACM} Symposium on Theory of
  Computing ({STOC})}, pages 253--264, 2019.

\bibitem[GM93]{GM93}
Michel~X. Goemans and Young{-}Soo Myung.
\newblock A catalog of {S}teiner tree formulations.
\newblock {\em Networks}, 23(1):19--28, 1993.

\bibitem[GORZ12]{GORZ12}
Michel~X. Goemans, Neil Olver, Thomas Rothvo{\ss}, and Rico Zenklusen.
\newblock Matroids and integrality gaps for hypergraphic {S}teiner tree
  relaxations.
\newblock In {\em Proceedings of the 44th {ACM} Symposium on Theory of
  Computing Conference ({STOC})}, pages 1161--1176, 2012.

\bibitem[GW95]{GW95}
Michel~X. Goemans and David~P. Williamson.
\newblock A general approximation technique for constrained forest problems.
\newblock {\em {SIAM} J. Comput.}, 24(2):296--317, 1995.

\bibitem[HDJAS23]{hyattdenesik_et_al:LIPIcs.ICALP.2023.79}
Dylan Hyatt-Denesik, Afrouz Jabal~Ameli, and Laura Sanit\`{a}.
\newblock {Finding Almost Tight Witness Trees}.
\newblock In {\em 50th International Colloquium on Automata, Languages, and
  Programming (ICALP 2023)}, pages 79:1--79:16, 2023.

\bibitem[HJ06]{HJ06}
Mohammad~Taghi Hajiaghayi and Kamal Jain.
\newblock The prize-collecting generalized {S}teiner tree problem via a new
  approach of primal-dual schema.
\newblock In {\em Proceedings of the 17th Annual {ACM-SIAM} Symposium on
  Discrete Algorithms ({SODA})}, pages 631--640, 2006.

\bibitem[Jai01]{J01}
Kamal Jain.
\newblock A factor 2 approximation algorithm for the generalized {S}teiner
  network problem.
\newblock {\em Combinatorica}, 21(1):39--60, 2001.

\bibitem[KOP{\etalchar{+}}17]{KOP0SV17}
Jochen K{\"{o}}nemann, Neil Olver, Kanstantsin Pashkovich, R.~Ravi, Chaitanya
  Swamy, and Jens Vygen.
\newblock On the integrality gap of the prize-collecting {S}teiner forest {LP}.
\newblock In {\em Proceedings of the International Conference on Approximation,
  Randomization, and Combinatorial Optimization ({APPROX/RANDOM})}, pages
  17:1--17:13, 2017.

\bibitem[KZ97]{KZ97}
Marek Karpinski and Alexander Zelikovsky.
\newblock New approximation algorithms for the {S}teiner tree problems.
\newblock {\em J. Comb. Optim.}, 1(1):47--65, 1997.

\bibitem[LL22]{LL22}
Shi Li and Bundit Laekhanukit.
\newblock Polynomial integrality gap of flow {LP} for directed {S}teiner tree.
\newblock In {\em Proceedings of the 2022 {ACM-SIAM} Symposium on Discrete
  Algorithms ({SODA})}, pages 3230--3236, 2022.

\bibitem[Meh88]{Mehlhorn88}
Kurt Mehlhorn.
\newblock A faster approximation algorithm for the {S}teiner problem in graphs.
\newblock {\em Inf. Process. Lett.}, 27(3):125--128, 1988.

\bibitem[PS00]{PS00}
Hans~J{\"{u}}rgen Pr{\"{o}}mel and Angelika Steger.
\newblock A new approximation algorithm for the {S}teiner tree problem with
  performance ratio 5/3.
\newblock {\em J. Algorithms}, 36(1):89--101, 2000.

\bibitem[RV99]{RajagopalanV99}
Sridhar Rajagopalan and Vijay~V. Vazirani.
\newblock On the bidirected cut relaxation for the metric {S}teiner tree
  problem.
\newblock In {\em Proceedings of the Tenth Annual {ACM-SIAM} Symposium on
  Discrete Algorithms ({SODA})}, pages 742--751, 1999.

\bibitem[RZ05]{RZ05}
Gabriel Robins and Alexander Zelikovsky.
\newblock Tighter bounds for graph {S}teiner tree approximation.
\newblock {\em {SIAM} J. Discret. Math.}, 19(1):122--134, 2005.

\bibitem[TZ21]{TZ21}
Vera Traub and Rico Zenklusen.
\newblock A better-than-2 approximation for weighted tree augmentation.
\newblock In {\em Proceedings of the 62nd {IEEE} Annual Symposium on
  Foundations of Computer Science ({FOCS})}, pages 1--12, 2021.

\bibitem[TZ22]{TZ22}
Vera Traub and Rico Zenklusen.
\newblock Local search for weighted tree augmentation and {S}teiner tree.
\newblock In {\em Proceedings of the 2022 {ACM-SIAM} Symposium on Discrete
  Algorithms ({SODA})}, pages 3253--3272, 2022.

\bibitem[TZ23]{TZ23}
Vera Traub and Rico Zenklusen.
\newblock A (1.5+{\(\epsilon\)})-approximation algorithm for weighted
  connectivity augmentation.
\newblock In {\em Proceedings of the 55th Annual {ACM} Symposium on Theory of
  Computing ({STOC})}, pages 1820--1833, 2023.

\bibitem[Vic20]{Vicari20}
Robert Vicari.
\newblock Simplex based {S}teiner tree instances yield large integrality gaps
  for the bidirected cut relaxation.
\newblock {\em CoRR}, abs/2002.07912, 2020.

\bibitem[Zel93]{Z93}
Alexander Zelikovsky.
\newblock An 11/6-approximation algorithm for the network {S}teiner problem.
\newblock {\em Algorithmica}, 9(5):463--470, 1993.

\bibitem[Zel96]{zelikovsky_1996_better}
A.~Zelikovsky.
\newblock Better approximation bounds for the network and {E}uclidean {S}teiner
  tree problems.
\newblock Technical report, University of Virginia, 1996.
\newblock CS-96-06.

\bibitem[ZK02]{ZK02}
Leonid Zosin and Samir Khuller.
\newblock On directed steiner trees.
\newblock In {\em Proceedings of the Thirteenth Annual {ACM-SIAM} Symposium on
  Discrete Algorithms ({SODA})}, pages 59--63, 2002.

\end{thebibliography}
\newcommand{\etalchar}[1]{$^{#1}$}

\end{document}